\pdfoutput=1
\documentclass[11pt]{article}

\usepackage{autobreak, bbm, caption, enumerate, geometry, mathtools, mathrsfs, tikzsymbols, titling}

\usepackage{algorithm,algpseudocode}

\usepackage[backend=biber, isbn=false, backref, maxbibnames=99]{biblatex}
\DefineBibliographyStrings{english}{
	backrefpage = {p\adddot},
	backrefpages = {pp\adddot}
}

\usepackage{amsmath,amsthm,amssymb, tabu}
\usepackage{booktabs}
\usepackage{dsfont}
\usepackage{enumitem}
\usepackage{float}
\usepackage{graphicx}
\usepackage{mathtools}
\usepackage{enumitem}
\usepackage{verbatim}
\usepackage{thmtools,thm-restate}

\usepackage{soul}
\usepackage{suffix}
\usepackage{tikz}
\usepackage{xcolor}
\usepackage{xspace}

\definecolor{weborange}{rgb}{.8,.3,.3}
\definecolor{webblue}{rgb}{0,0,.8}
\definecolor{internallinkcolor}{rgb}{0,.5,0}
\definecolor{externallinkcolor}{rgb}{0,0,.5}

\definecolor{DarkBlue}{rgb}{0,0,0.8}
\definecolor{DarkOrange}{rgb}{0.8,0.4,0}
\def\mylinkcolor{DarkBlue}

\usepackage{hyperref}
\hypersetup{breaklinks=true, colorlinks=true, urlcolor=\mylinkcolor, linkcolor=\mylinkcolor, citecolor=\mylinkcolor}

\usetikzlibrary{calc,positioning}

\usepackage[capitalize]{cleveref}

\newcommand{\hnote}[1]{{\textcolor{blue}{\small (Henry: #1)}}}
\newcommand{\gnote}[1]{\textcolor{magenta}{\small (Greg: #1)}}

\newcommand{\hlabel}[1]{\phantomsection\label{#1}}
\renewcommand{\hnote}[1]{}
\renewcommand{\gnote}[1]{}

\definecolor{White}{rgb}{1,1,1}
\definecolor{Black}{rgb}{0,0,0}
\definecolor{LightGray}{rgb}{.8,.8,.8}
\colorlet{ChannelColor}{LightGray}
\colorlet{ChannelTextColor}{Black}
\colorlet{ReadoutColor}{White}

\numberwithin{equation}{section}

\theoremstyle{plain}
\newtheorem{thm}{Theorem}[section]
\newtheorem{lem}[thm]{Lemma}
\newtheorem{cor}[thm]{Corollary}
\newtheorem{prp}[thm]{Proposition}

\newtheorem{clm}[thm]{Claim}

\newtheorem*{rmd*}{Reminder}

\crefname{lem}{Lemma}{Lemmas}
\crefname{thm}{Theorem}{Theorems}
\crefname{cor}{Corollary}{Corollaries}
\crefname{clm}{Claim}{Claims}

\theoremstyle{definition}
\newtheorem{dfn}[thm]{Definition}

\newtheorem{prb}[thm]{Problem}
\crefname{prb}{Problem}{Problems}

\theoremstyle{remark}
\newtheorem*{rmk*}{Remark}

\let\originalleft\left
\let\originalright\right
\renewcommand{\left}{\mathopen{}\mathclose\bgroup\originalleft}
\renewcommand{\right}{\aftergroup\egroup\originalright}

\newcommand*{\C}{\ensuremath{\mathbb{C}}}
\newcommand*{\N}{\ensuremath{\mathbb{N}}}
\newcommand*{\Z}{\ensuremath{\mathbb{Z}}}
\newcommand*{\Q}{\ensuremath{\mathbb{Q}}}
\newcommand*{\R}{\ensuremath{\mathbb{R}}}
\newcommand*{\D}[1]{\mathbb D_{#1}}

\newcommand{\microspace}{\mspace{.5mu}}
\newcommand{\ket}[1]{\ensuremath{\left\lvert\microspace #1
		\microspace\right\rangle}}
\newcommand{\bra}[1]{\ensuremath{\left\langle\microspace #1
		\microspace\right\vert}}
\newcommand{\ketbra}[2]{\ensuremath{\left\lvert\microspace #1
		\microspace\right\rangle\! \left\langle \microspace #2 \microspace \right\rvert}}
\newcommand*{\kb}[1]{\ketbra{#1}{#1}}
\newcommand*{\ip}[2]{\left\langle #1 \middle| #2 \right\rangle}

\usetikzlibrary{shapes,arrows,quantikz}

\newcommand*{\adj}[1]{#1^\dagger}

\newcommand*{\bits}{\{0,1\}}

\newcommand*{\cc}[1]{\ensuremath{\mathsf{#1}}}
\newcommand*{\conj}[1]{#1^*}
\newcommand*{\cube}[1]{\bits^{#1}}

\newcommand*{\eps}{\varepsilon}
\newcommand*{\E}{\mathbb{E}}

\newcommand*{\hw}{\mrm{hw}}
\newcommand*{\Ind}[1]{\mathbbm1\!\{#1\}}
\newcommand*{\interact}{\mathord{\leftrightarrows}}
\newcommand*{\mai}[1]{\mathit{#1}}
\newcommand*{\mrm}[1]{\mathrm{#1}}

\newcommand*{\norm}[1]{\|#1\|}
\newcommand*{\Norm}[1]{\left\|#1\right\|}
\newcommand*{\poly}{\mathrm{poly}}
\newcommand*{\pr}[1]{\mathrm{Pr}\{#1\}}

\newcommand*{\PR}[1]{\mathrm{Pr}\left\{#1\right\}}

\newcommand*{\reg}{\mathsf}

\newcommand*{\unif}{\sim}
\newcommand{\uppart}[1]{^{\left(#1\right)}}
\newcommand*{\tr}{\mathrm{tr}}
\DeclareMathOperator*{\td}{td}
\newcommand*\wt\widetilde
\newcommand*{\zs}{0\cdots 0}
\DeclarePairedDelimiter{\ceil}{\lceil}{\rceil}
\DeclarePairedDelimiter{\floor}{\lfloor}{\rfloor}
\newcommand*{\stateQIP}[1]{\ensuremath{\cc{stateQIP}\left[#1\right]}}

\newcommand{\Paren}[1]{\left(#1\right)}

\newcommand{\class}[1]{\mathsf{#1}}

\newcommand{\NP}{\class{NP}}
\newcommand{\IP}{\class{IP}}
\newcommand{\QMA}{\class{QMA}}
\newcommand{\PSPACE}{\class{PSPACE}}
\newcommand{\QIP}{\class{QIP}}

\newcommand{\statePSPACE}{\cc{statePSPACE}\xspace}
\newcommand{\unitaryPSPACE}{\cc{unitaryPSPACE}\xspace}

\newcommand{\alg}[1]{\texttt{#1}}

\newcommand{\id}{I}

\newcommand{\flag}{\mrm{flag}}
\newcommand{\work}{\mrm{work}}
\newcommand{\out}{\mrm{out}}

\algblockdefx{Ctrl}{EndCtrl}[1]{\textbf{controlled on} #1}{\textbf{end control}}

\floatname{algorithm}{Procedure}
\crefname{algorithm}{Procedure}{Procedures}
\algnewcommand{\IIf}[1]{\State\algorithmicif\ #1\ \algorithmicthen}
\algnewcommand{\EndIIf}{\unskip\ \algorithmicend\ \algorithmicif}

\allowdisplaybreaks

\begin{document}
	
	\title{Interactive Proofs for Synthesizing \\ Quantum States and Unitaries}
	
	\author{Gregory Rosenthal\thanks{Department of Computer Science, University of Toronto. \texttt{rosenthal@cs.toronto.edu}. Supported by NSERC (PGS D).} 
		\and 
		Henry Yuen\thanks{Department of Computer Science, Columbia University. \texttt{hyuen@cs.columbia.edu}. Supported by an NSERC Discovery Grant, a Google Research Award, and AFOSR award FA9550-21-1-0040.}
	}
	\date{\vspace{-5ex}}
	\maketitle
	
	\begin{abstract}
		Whereas quantum complexity theory has traditionally been concerned with problems arising from \emph{classical} complexity theory (such as computing boolean functions), it also makes sense to study the complexity of inherently \emph{quantum} operations such as constructing quantum states or  performing unitary transformations. With this motivation, we define models of \emph{interactive proofs} for synthesizing quantum states and unitaries, where a polynomial-time quantum verifier interacts with an untrusted quantum prover, and a verifier who accepts also outputs an approximation of the target state (for the state synthesis problem) or the result of the target unitary applied to the input state (for the unitary synthesis problem); furthermore there should exist an ``honest" prover which the verifier accepts with probability 1.
		
		Our main result is a ``state synthesis'' analogue of the inclusion $\PSPACE \subseteq \IP$: any sequence of states computable by a polynomial-space quantum algorithm (which may run for exponential time) admits an interactive protocol of the form described above. Leveraging this state synthesis protocol, we also give a unitary synthesis protocol for polynomial space-computable unitaries that act nontrivially on only a polynomial-dimensional subspace. We obtain analogous results in the setting with multiple entangled provers as well.
	\end{abstract}
	
	\newpage
	
	\tableofcontents

	\section{Introduction}
\label{sec:intro}

In quantum computing and quantum information processing, there are tasks that are ``inherently quantum'', meaning that it does not even make sense for a classical computer to perform them. Such tasks include:
\begin{itemize}[leftmargin=*]
	\renewcommand\labelitemi{--}
	\item \emph{State synthesis}: given an implicit description of a quantum state, construct the state.
	\item \emph{State transformations}: given an implicit description of a quantum operation (e.g.\ a unitary), perform it on a given input state.
\end{itemize}
Many quantum protocols and algorithms are most naturally viewed as synthesizing a state, performing a state transformation, or both. For example, primitives in quantum cryptography such as quantum money~\cite{aaronson2009quantum} or quantum pseudorandom states~\cite{ji2018pseudorandom} revolve around constructing highly entangled, difficult-to-clone states. The class of algorithms known as variational quantum eigensolvers are meant to prepare grounds states of physical systems~\cite{cerezo2020variational}. A decoder for a quantum error-correcting code transforms noise-corrupted states into noise-free states~\cite{lidar2013quantum}. 

This motivates the study of the complexity of state synthesis and state transformations. The central question is the following: how difficult is it to prepare a given state or perform a given unitary transformation?\footnote{We highly recommend \cite{aaronson2016complexity} for an introduction to this nascent subject.} Unlike how search and decision problems can often be reduced to each other in classical computer science, such ``inherently quantum'' tasks cannot obviously be reduced to analogous classical decision or search problems. For example, it is unknown whether the ability to decide the local Hamiltonian problem (a problem that is complete for $\QMA$, the quantum analogue of $\NP$) in polynomial time implies the ability to efficiently \emph{construct} ground states of local Hamiltonians on a quantum computer.\footnote{In fact there is some evidence in the form of an oracle separation that efficient search-to-decision reductions for $\QMA$ do not exist~\cite{irani2021quantum}.} Part of the difficulty comes from the fact that an $n$-qubit quantum state, comprised of $2^n$ amplitudes, is an exponentially more complex object than an $n$-bit string.

We investigate \emph{interactive proofs for synthesizing states and unitaries}. In traditional models of interactive proofs (even the ones associated with quantum complexity classes such as $\cc{QIP}$ and $\cc{MIP}^*$), the goal of the verifier is to solve a decision problem with the help of an all-powerful prover. We propose a model of ``inherently quantum'' interactive proofs where the verifier's goal is to synthesize a quantum state or perform a quantum operation. The challenge is for the verifier to use the help of an untrusted prover to perform these tasks in a \emph{verifiable} way.

Theoretical computer science has been deeply influenced by interactive proofs. Many central concepts and subjects, ranging from zero-knowledge proofs~\cite{goldwasser1989knowledge} to hardness of approximation~\cite{arora1998proof}, came from their study. We believe that studying state synthesis and state transformation tasks in the interactive setting will similarly give us a powerful approach to understanding the complexity of quantum states and quantum operations.

We first discuss interactive state synthesis, and then discuss interactive unitary synthesis in \Cref{sec:intro-transforms}.
\subsection{Interactive state synthesis} \label{sec:intro-state}

Consider the following state synthesis problem: given a succinct description of a quantum circuit $C$ acting on $n$ qubits with depth up to $\exp(\poly(n))$, generate the state $\ket{\psi} = C \ket{0^n}$. That is, construct the state that results from applying the circuit $C$ to the all-zeros state. By succinct description, we mean that there is (for example) a Turing machine $M$ that on input $1 \leq j \leq \exp(\poly(n))$ written in binary, uses $\poly(n)$ space and outputs the $j$'th gate of $C$.

This state synthesis problem has a natural physics motivation. Let $H$ denote a local Hamiltonian on $n$ qubits and let $t$ be an integer that is at most $\exp(\poly(n))$. Letting $C = \exp(-iH t)$, the state $\ket{\psi}$ denotes the evolution of the all-zeros state with respect to the Hamiltonian $H$ for time $t$ under the Schr\"{o}dinger equation. An experimental physicist may want, for example, to get their hands on $\ket{\psi}$ in order to study the long-term dynamics of the Hamiltonian $H$ (e.g., perhaps the physicist wants to investigate the equilibration properties of the system after exponential time, or whether it exhibits certain kinds of chaotic behavior~\cite{swingle2018unscrambling}). 

This state synthesis problem can be solved in quantum polynomial space: a quantum computer running in exponential time but using polynomial space can generate the state $\ket{\psi}$ by simply computing the gates of $C$ and applying them one by one (and uncomputing each gate after applying it). Furthermore, this state synthesis problem is $\PSPACE$-hard; it was recently shown that $\PSPACE$-hard problems can be embedded into exponential-time, polynomial-space unitary computations (i.e.\ quantum computations without intermediate measurements)~\cite{fefferman2021eliminating,girish2021quantum}.

Since we do not expect quantum computers to be able to solve $\PSPACE$-hard problems in polynomial time, we do not expect this state synthesis problem to be solvable in polynomial time either. Here we ask whether a polynomial-time quantum computer can synthesize the state $\ket{\psi}$ with the help of an all-powerful, but untrusted, prover. We call this the \emph{Interactive State Synthesis} problem.

In a sense, we are asking whether there is a ``state complexity version'' of the $\IP = \PSPACE$ protocol. By this we mean the following: the celebrated result $\IP = \PSPACE$ of~\cite{lund1992algebraic,shamir1992ip} shows that a polynomial-time \emph{classical} verifier can verify membership in any $\PSPACE$ language by interacting with an all-powerful but untrusted prover. Whereas the protocol developed by~\cite{lund1992algebraic,shamir1992ip} (and what we'll henceforth call the ``$\IP = \PSPACE$ protocol'') is meant to verify \emph{decision problems}, it is straightforward to extend the protocol in order to solve the \emph{function} version of any $\PSPACE$ problem, where the goal is to produce the output string $s \in \{0,1\}^{\poly(n)}$ of a polynomial-space Turing machine $M$ that has run for time $t \leq \exp(\poly(n))$. To do so, the verifier can treat each output bit of $s$ as the answer to a well-defined decision problem (i.e.\ ``Is the $i$'th bit of the output of $M$, when run for time $t$, equal to $1$?''), and summon the $\IP = \PSPACE$ protocol for each of the output bits. 

In our state synthesis problem, the goal is not just to produce a string $s$ on $n$ bits but an entire \emph{quantum state} on $n$ qubits. Note that by leveraging the $\IP = \PSPACE$ protocol, a polynomial-time verifier can compute any amplitude of the target state $\ket{\psi}$ of its choice (up to a phase). This is because (as we show in \Cref{sec:tomography}) the quantity $|\alpha_x|^2 = |\ip{x}{\psi}|^2$ can be approximated in polynomial space; since $\IP = \PSPACE$ a polynomial-time verifier can interact with a prover to compute such an approximation also. However, in our setting the quantum verifier wants to certify that at the end of an interaction with an untrusted prover, a specified collection of $n$ qubits in the possession of the verifier is (close to) the coherent superposition $\ket{\psi} = \sum_x \alpha_x \ket{x}$.

This goal immediately raises a number of challenges: first, the verifier has to somehow obtain information about exponentially many amplitudes in polynomial time. Second, the verifier has to also check that the prover has not maliciously entangled itself with the $n$ target qubits at the end of the interaction; for example, the verifier should ensure that it doesn't have the first $n$ qubits of the entangled state $\sum_x \alpha_x \ket{x} \ket{\phi_x}$ where the states $\{ \ket{\phi_x} \}$ are in the possession of the prover.

Our main result is a positive solution to the Interactive State Synthesis problem. We say that a family $( \ket{\psi_n} )_{n \in \N}$ of quantum states where $\ket{\psi_n}$ has $n$ qubits is in \statePSPACE if there exists a space-uniform family of quantum circuits $(C_n)_{n \in \N}$ such that $C_n$ outputs $\ket{\psi_n}$. In turn, the circuit family $(C_n)_{n \in \N}$ is space-uniform if $C_n$ is of size at most $\exp(\poly(n))$ and uses $\poly(n)$ qubits, and if there exists a (classical) Turing machine that on input $(1^n,j)$, outputs the $j$'th gate of $C_n$ using $\poly(n)$ space.

\begin{thm}[Main theorem, informal]
	\label{thm:main-intro}
	Let $( \ket{\psi_n} )_{n \in \N}$ denote a family of quantum states in \statePSPACE. Then there exists an interactive protocol between a polynomial-time quantum verifier and an untrusted quantum prover that, on input $1^n$, constructs an approximation of $\ket{\psi_n}$. More precisely, the protocol has the following guarantees: for all $n \in \N$, when the verifier receives input $1^n$,
	\begin{itemize}[leftmargin=*]
		\item \emph{(Completeness)} There exists an ``honest'' prover that is accepted by the verifier with probability $1$, and for which the verifier outputs a density matrix $\rho$ such that
		\[
		\td(\rho,\kb{\psi_n}) \leq \exp(-\Omega(n))~.
		\]
		\item \emph{(Soundness)} For all prover strategies, 
		\[
		\Pr \left\{ \text{verifier accepts and outputs $\rho$ such that $\td(\rho,\kb{\psi_n}) \geq \frac{1}{\poly(n)}$} \right\} \leq \exp(-\Omega(n))~.
		\]
		Here, $\td(\cdot,\cdot)$ denotes trace distance. 
	\end{itemize}
\end{thm}

Intuitively, a state $\ket{\psi}$ can be synthesized via an interactive protocol between a verifier and a prover if (a) an ``honest'' prover can help the verifier synthesize $\ket{\psi}$ with high probability (i.e.\ the completeness property), and (b) for all provers, if the verifier accepts with non-negligible probability, then the output state is close to $\ket{\psi}$ (i.e.\ the soundness property). 

We note that \Cref{thm:main-intro} is \emph{not} implied by the result $\cc{QIP} = \PSPACE$~\cite{jain2011qip}. This is because $\cc{QIP}$, though defined in terms of quantum verifiers, is still a class of decision problems, and the result does not say anything about the complexity of performing state synthesis. We do, however, use the $\cc{QIP} = \cc{PSPACE}$ protocol as a subroutine in our state synthesis protocol. See \cref{sec:state-synthesis-overview} for a detailed overview of the proof of \cref{thm:main-intro}. 

Our precise statement of \Cref{thm:main-intro} is formulated in terms of a \emph{state complexity class}. Unlike standard complexity classes (such as $\cc{BQP}$ and $\cc{QIP}$) which are sets of languages, a state complexity class is a set of families $(\ket{\psi_n})_n$ of quantum states.\footnote{A more general definition of state complexity classes would allow for sets of states $\{ \ket{\psi_x} \}_{x \in \cube*}$ indexed by binary strings, but for simplicity we focus on families of states indexed by $\N$.} 

We introduce the class $\cc{stateQIP}$, which consists of all state families that can be generated by interactive protocols with completeness and soundness properties like those described in \Cref{thm:main-intro}.\footnote{Actually, strictly speaking we only define $\stateQIP{c,s}$ for completeness and soundness parameters $c$ and $s$.} \Cref{thm:main-intro} can then be succinctly recast as stating that $\statePSPACE \subseteq \cc{stateQIP}$. 

We also prove a partial converse, which can be interpreted as ``$\cc{stateQIP} \subseteq \cc{stateEXP}$":

\begin{thm}[informal]
	\label{thm:converse-intro}
	For all $(\ket{\psi_n})_n \in \cc{stateQIP}$ there exists a uniform family of exponential-size quantum circuits $Q = (Q_n)_n$ such that $Q_n$ outputs an approximation of $\ket{\psi_n}$. 
\end{thm}

By uniform, we mean that there is an exponential-time Turing machine $M$ that on input $(1^n,j)$ outputs the $j$-th gate of the circuit $Q_n$. We note that the uniformity condition makes this theorem nontrivial: while \emph{every} quantum state on $n$ qubits has a $\exp(\poly(n))$-size circuit that synthesizes it,\footnote{To a good approximation, using the Solovay-Kitaev theorem~\cite{dawson2006solovay}.} it is not necessarily the case that for an arbitrary state family $(\ket{\psi_n})_n$ there is a single Turing machine that specifies \emph{all} of the exponential size circuits synthesizing each $\ket{\psi_n}$.

This leaves open the intriguing question of whether a full converse can be proven. In other words, is it true that if $(\ket{\psi_n})_n$ is in $\cc{stateQIP}$, then in fact $(\ket{\psi_n})_n$ is (approximately) in \statePSPACE? This would establish a full state synthesis analogue of the $\IP = \PSPACE$ theorem.

We also prove that $\cc{stateR} = \cc{stateQMIP}$, where $\cc{stateR}$ denotes the class of state sequences $(\ket{\psi_n})_{n \in \N}$ such that $\ket{\psi_n}$ is on $n$ qubits and the description of an approximation of $\ket{\psi_n}$ is computable as a function of $n$, and $\cc{stateQMIP}$ is defined similarly to $\cc{stateQIP}$ but with multiple entangled provers instead of a single prover. The proof uses the result $\cc{QMIP} = \cc{RE}$~\cite{ji2020mip}, and is otherwise similar to the proof of \cref{thm:main-intro}. The reason that $\cc{stateQMIP} \subseteq \cc{stateR}$, whereas $\cc{QMIP} \nsubseteq \cc{R}$, relates to the distinction between search and decision problems.

\subsection{Interactive unitary synthesis}
\label{sec:intro-transforms}

We now consider the following unitary synthesis problem: given a pair $(C,\ket{\phi})$ where $C$ is a succinctly-described $n$-qubit circuit with depth at most $\exp(\poly(n))$ and $\ket{\phi}$ is an $n$-qubit state, generate the state $C \ket{\phi}$. This problem is solvable in quantum polynomial space: an algorithm can simply take the input state $\ket{\phi}$ and apply the circuit $C$ to it. 

A difference between this problem and the state synthesis problem presented in \cref{sec:intro-state} is that in the latter problem, the input $C$ provides an implicit classical description of the target state $\ket{\psi} = C \ket{0^n}$.
In the unitary synthesis problem, however, the input $\ket{\phi}$ to the circuit $C$ is provided \emph{in quantum form}. Even if an algorithm were given unlimited time, it would not in general be able to compute a classical description of $C \ket{\phi}$; this is because only one copy of the state $\ket{\phi}$ is provided. For this reason, the unitary synthesis problem appears more challenging than the state synthesis problem. 

Similarly to before, we ask whether a polynomial-time verifier can accomplish the unitary synthesis task with the help of a prover. We call this the \emph{Interactive Unitary Synthesis} problem. In addition to all of the difficulties that arise with the Interactive State Synthesis problem, here we also have to contend with the fact that the one copy of $\ket{\phi}$ is a precious resource: a verifier has to take great care to ensure that a malicious prover does not corrupt the state in an undetectable way.

We are not able to give a general solution to the Interactive Unitary Synthesis problem here. However we present solutions for nontrivial cases. We denote by \unitaryPSPACE the set of space-uniform sequences $(C_n)_{n \in \N}$ of unitary quantum circuits where each $C_n$ acts on $n$ qubits. The first nontrivial case deals with families $(C_n)_{n \in \N}$ in $\unitaryPSPACE$ that have \emph{polynomial action}: this means that $C_n$ acts nontrivially only on a $\poly(n)$-dimensional subspace.

Interesting families of unitaries $(C_n)_{n \in \N}$ that have polynomial action include \emph{reflections} $C_n = \id - 2\kb{\theta_n}$ where $( \ket{\theta_n} )_{n}$ is some family of states. These unitaries act nontrivially on a one-dimensional subspace (namely, the space spanned by $\ket{\theta_n}$). Since the states $( \ket{\theta_n} )_n$ might be extremely complicated (requiring exponential time to synthesize without the help of a prover, for example), applying the unitaries $(C_n)_{n}$ can still be quite nontrivial.

\begin{thm}[Interactive synthesis of polynomial-action unitaries, informal]
	\label{thm:poly-action-intro}
	Let $( C_n )_{n \in \N}$ denote a family of unitaries in \unitaryPSPACE with polynomial action. Then there exists an interactive protocol between a polynomial-time quantum verifier and an untrusted quantum prover that, given input an $n$-qubit state $\ket{\phi}$, constructs an approximation of $\ket{\psi_n} = C_n \ket{\phi}$. More precisely, the protocol has the following guarantees: for all $n \in \N$, 
	\begin{itemize}[leftmargin=*]
		\item \emph{(Completeness)} There exists an ``honest'' prover such that for all input states $\ket{\phi}$, the verifier accepts with probability $1$ and outputs a density matrix $\rho$ such that 
		\[
		\td(\rho,\kb{\psi_n}) \leq 1/\poly(n)~.
		\]
		\item \emph{(Soundness)} For all input states $\ket{\phi}$, for all prover strategies, 
		\[
		\Pr \left\{ \text{verifier accepts and outputs $\rho$ such that $\td(\rho,\kb{\psi_n}) \ge \frac{1}{\poly(n)}$} \right\} \leq \exp(-\Omega(n))~.
		\]
	\end{itemize}
\end{thm}

(Again, a similar statement holds with multiple entangled provers if we write $\cc{unitaryR}$ in place of $\cc{unitaryPSPACE}$.) Extending \Cref{thm:poly-action-intro} to all unitary families in \unitaryPSPACE, i.e.\ not just polynomial-action ones, is an open problem that we leave for future work. (It would even be interesting to do so using \emph{multiple} provers, either entangled or unentangled.) This question seems closely related to an open problem posed by Aaronson and Kuperberg~\cite{aaronson2007quantum}:

\begin{prb}[Unitary Synthesis Problem]
	\label{prb:unitary-synthesis}
	Is it true that for every $n$-qubit unitary $U$ there exists a classical oracle $M$ such that a $\poly(n)$-time quantum algorithm with query access to $M$ can approximately implement $U$?
\end{prb}

On the other hand, the corresponding oracle result for state synthesis \emph{is} known: for every $n$-qubit state $\ket{\psi}$, there is a classical oracle $M$ that can be queried by a polynomial-time quantum algorithm to synthesize $\ket{\psi}$  (see \Cref{prop:oracle-state-synthesis}). As explained in \Cref{sec:state-synthesis-overview}, this is the starting point for our state synthesis protocol of \Cref{thm:main-intro}. Furthermore, by reducing to the case of state synthesis, \Cref{prb:unitary-synthesis} \emph{does} admit a solution when the unitary $U$ to be simulated only acts nontrivially on a polynomial-dimensional subspace --- and \Cref{thm:poly-action-intro} is the analogue of this fact in the setting of interactive unitary synthesis.

We also show how to generalize \Cref{thm:poly-action-intro} beyond polynomial-action unitary families, provided that the verifier also receives a succinct description of a polynomial-dimensional subspace $S$ which contains the input state $\ket{\phi}$. By succinct description, we mean a polynomial-space Turing machine $M$ that on input $j$, outputs the $j$'th gate of a quantum circuit $P$ that implements the reflection about $S$ (i.e.\ $P = \id - 2 \Pi_S$ where $\Pi_S$ is the projection onto $S$).

\begin{cor}[Interactive unitary synthesis with restricted inputs, informal]
	\label{cor:poly-input-intro}
	Let $( C_n )_{n \in \N}$ denote a family of unitaries in \unitaryPSPACE. Then there exists an interactive protocol between a polynomial-time quantum verifier and an untrusted quantum prover that, when the verifier is given a succinctly-described subspace $S$ of $n$-qubit states and a state $\ket{\phi} \in S$ as input, constructs an approximation of $\ket{\psi_n} = C_n \ket{\phi}$. More precisely, the protocol has the following guarantees: for all polynomials $p$, for all $n \in \N$,
	\begin{itemize}[leftmargin=*]
		\item \emph{(Completeness)} There exists an ``honest'' prover such that for all succinctly-described subspaces $S \subset (\C^2)^{\otimes n}$ of dimension at most $p(n)$ and input states $\ket{\phi} \in S$, the verifier accepts with probability $1$ and outputs a density matrix $\rho$ such that
		\[
		\td(\rho,\kb{\psi_n}) \leq 1/\poly(n)~.
		\]
		\item \emph{(Soundness)} For all succinctly described subspaces $S \subset (\C^2)^{\otimes n}$ of dimension at most $p(n)$ and input states $\ket{\phi} \in S$, for all prover strategies, 
		\[
		\Pr \left\{ \text{verifier accepts and outputs $\rho$ such that $\td(\rho,\kb{\psi_n}) \ge \frac{1}{\poly(n)}$} \right\} \leq \exp(-\Omega(n))~.
		\]
	\end{itemize}
\end{cor}

In other words, compared to \cref{thm:poly-action-intro}, this corollary trades the assumption on the family of unitaries for an assumption that the verifier has a classical description of a low-dimensional subspace which contains the input state.

Similarly to our state synthesis results, the formal statements of \Cref{thm:poly-action-intro,cor:poly-input-intro} are presented in terms of a \emph{unitary complexity class} $\cc{unitaryQIP}$, which consists of all families $(U_n)_n$ of unitaries that can be synthesized via an interactive protocol with completeness and soundness properties like those described in \Cref{thm:poly-action-intro,cor:poly-input-intro}.

\subsection{Open problems and future directions} \label{sec:opfd}

We list some open problems:
\begin{enumerate}[leftmargin=*]
	\item Does it hold that $\cc{stateQIP} \subseteq \cc{statePSPACE}$? In other words, if a family of states can be synthesized interactively via a $\cc{stateQIP}$ protocol, does that mean the family of states can be synthesized in polynomial space? Answering this question may be related to the next two open problems.
	
	\item The state synthesis protocol of \Cref{thm:main-intro} has \emph{state soundness error} that is $1/\poly(n)$. This means that, conditioned on the verifier accepting with probability at least $\exp(-\Omega(n))$, the output state of the verifier is guaranteed to be at least $1/\poly(n)$-close to the ideal target state $\ket{\psi_n}$. Is it possible to also get the closeness guarantee to be exponentially small, even at higher acceptance probabilities?
	
	\item A fundamental result by Watrous~\cite{watrous03pspace} is that all $\QIP$ protocols with a polynomial number of rounds can be parallelized so that they require three rounds (the prover sends the first message, the verifier sends a response, and then prover sends the final message). Does this also hold for $\cc{stateQIP}$ and $\cc{unitaryQIP}$ protocols? \label{item:rr}
	
	\item Does it hold that $\cc{unitaryPSPACE} \subseteq \cc{unitaryQIP}$? In other words, if a family of unitaries can be implemented via polynomial space quantum circuits, then is there a $\cc{unitaryQIP}$ protocol for it? Can one prove this inclusion under the complexity-theoretic assumption that $\cc{P} = \cc{PSPACE}$? While $\cc{P} = \cc{PSPACE}$ implies that states in \statePSPACE can be synthesized in quantum polynomial time (without the help of a prover!), it is not obvious that this helps with synthesizing unitaries in \unitaryPSPACE.
	
	\item A potentially easier goal than the previous item is to show other interesting family of unitaries that can be interactively synthesized. For example, for which classes of Hamiltonians $H$ does the evolution operator $e^{-iHt}$ admit an interactive synthesis protocol for exponentially long $t$? This may be related to questions about ``fast-forwarding of Hamiltonians'' that were raised by~\cite{atia2017fast}.
	
	\item Can our state synthesis results be extended to synthesize families of \emph{mixed states} like those in \statePSPACE?\footnote{It would suffice to show that any such mixed state has a purification in \statePSPACE. This does not immediately follow from the results of \cite{fefferman2021eliminating,girish2021quantum} about purifying space-bounded quantum computation, because these results seem to apply only to decision problems.)} Similarly, one can generalize the task of synthesizing unitaries to synthesizing \emph{channels}. Are there special classes of channels (going beyond unitaries) which can be interactively synthesized?
	
	\item Are there interesting cryptographic applications of interactive state synthesis or unitary synthesis? For example, is there a meaningful notion of \emph{zero-knowledge} interactive state synthesis?
	
	\item Are there interesting classes of states that can be interactively synthesized with an \emph{efficient} prover? Our state synthesis protocol of \Cref{thm:main-intro} requires the honest prover to be capable of solving $\PSPACE$-complete problems. In cryptography, however, interactive protocols require that the honest prover can run in polynomial time (perhaps after having received some advice) and be accepted with high probability. 
	
	\item What state or unitary families can be synthesized by a verifier interacting with multiple \emph{un}entangled quantum provers? As will be explained later, the obstacles appear to be soundness amplification and obtaining an appropriate analogue of the results from \cref{sec:tomography}.
\end{enumerate}

Regarding Question~\ref{item:rr}, we remark that our state and unitary synthesis protocols have polynomially many rounds. However, in \cref{sec:rrsus} we sketch constant-round variants of our state and unitary synthesis protocols, and we conjecture that these constant-round variants satisfy completeness and soundness as well.

The Interactive State and Unitary Synthesis problems can be seen as the confluence of two central research directions in quantum computing and quantum information. One is about investigating protocols for \emph{verification} of quantum computations, and the other is about exploring the \emph{complexity of quantum states and state transformations}. Over the past decade, the field has developed significant insight into how quantum computations can be verified, including the fact that a classical verifier can efficiently check, via an interactive proof, the results of polynomial-time quantum computations~\cite{mahadev2018classical}. The systematic study of the complexity of ``inherently quantum'' tasks is much more nascent, however. 

Questions about state and unitary complexity are often asked in the context of a \emph{specific} state to prepare or a \emph{specific} quantum operation to perform. There is a conspicuous lack, however, of a \emph{general} theory of quantum state and unitary complexity. We hope that our results contribute toward building such a theory.
	\section{Technical Overview}
\label{sec:overview}
\subsection{State synthesis} \label{sec:state-synthesis-overview}
\subsubsection{State synthesis with a \emph{trusted} prover}
The starting point for our proof of \Cref{thm:main-intro} is a solution to an \emph{oracle} version of the state synthesis problem, due to Aaronson~{\cite[Proposition 3.3.5]{aaronson2016complexity}}:\footnote{Aaronson's argument also appears implicitly in \cite{grover2002creating} in the case where $\ket{\psi} = \sum_{x \in \cube n} \alpha_x \ket{x}$ for nonnegative real numbers $\alpha_x$.}

\begin{prp}[\cite{aaronson2016complexity,grover2002creating}, informal]
	\label{prop:oracle-state-synthesis}
	For all $n$-qubit states $\ket{\psi}$ there exists a classical oracle $M$ and a $\poly(n)$-size quantum circuit $C$ with oracle access to $M$ such that $C^M \ket{0^n} \approx \ket{\psi}$.
\end{prp}

We can equivalently view \Cref{prop:oracle-state-synthesis} as stating that there is an efficient synthesis protocol with a \emph{trusted} prover (namely, the oracle). The proof is as follows. Write the state $\ket{\psi}$ as $\sum_{x \in \{0,1\}^n} \alpha_x \ket{x}$. For all prefixes $y \in \{0,1\}^{\leq n}$, define $p_y$ do be the probability that measuring the first $|y|$ qubits of $\ket{\psi}$ yields the string $y$. For all $1 \leq k \leq n$ define the intermediate states 
\begin{equation}
	\label{eq:intermediate-state}
	\ket{\psi^{(k)}} = \sum_{\mathclap{y \in \{0,1\}^k}} \sqrt{p_y} \,\, \ket{y}~.
\end{equation}
Note that $\ket{\psi^{(n)}}$ is a ``phase-less'' version of the original state $\ket{\psi}$, in that $\ket{\psi^{(n)}} = \sum_{x \in \{0,1\}^*} |\alpha_x| \, \ket{x}$. We now define the corresponding classical oracle $M$ as follows: for all $x \in \{0,1\}^{\leq n}$,
\begin{equation*}
	M(x) =
	\begin{cases}
		p(x0)/p(x) & \text{if } |x|<n, \\
		\alpha_x/|\alpha_x| &\text{if } |x|=n~.
	\end{cases}
\end{equation*}
(In the definition of $M$ we use the convention that if the denominator $p(x)$ or $|\alpha_x|$ is $0$, then the quotient is defined to be $1$.) In other words, on a prefix $x$ that is less than $n$ bits long, the oracle $M$ outputs the conditional probability that measuring the first $|x|+1$ qubits of $\ket{\psi}$ yields the string $x0$ (i.e.\ $x$ appended with $0$), \emph{conditioned} on the first $|x|$ qubits measuring to $x$. If $x$ is the empty string then $M(x) = p(0)$. When given an $n$-bit string $x$ as input, the oracle $M$ outputs the phase $\alpha_x/|\alpha_x|$ (which is a unit-norm complex number). 

In \Cref{alg:aaronson} we describe the circuit $C$ that, given oracle access to $M$, synthesizes the state $\ket{\psi}$. In the description, the notation $\reg{A}_{[k]}$ denotes the concatenation of registers $\reg A_1, \dotsc, \reg A_k$. Note that the oracle $M$ is queried \emph{in superposition} in each iteration. 

\begin{algorithm}
	\caption{Constructing $\ket{\psi}$ with access to a trusted oracle $M$}
	\label{alg:aaronson}
	\begin{algorithmic}[1]
		\State Initialize one-qubit registers $\reg A_1, \dotsc, \reg A_n$ each to $\ket0$, and initialize an ancilla register $\reg{D}$ to all zeros.
		\For{$k=0$ to $n$}
		\Ctrl{the state $\ket x$ of $\reg A_{[k]}$ where $x \in \cube k$,}
		\State Call the oracle $M$ on input $x$, and save the output $M(x)$ in register $\reg{D}$.
		\Ctrl{the state $\ket m$ of $\reg D$ where $m \in \cube{\poly(n)}$ encodes a number,}
		\If{$k < n$}  construct $\sqrt{m} \ket 0 + \sqrt{1-m} \ket 1$ in register $\reg A_{k+1}$.
		\ElsIf{$k=n$} apply the phase $m$.
		\EndIf
		\EndCtrl 
		\State Call the oracle $M$ to uncompute $M(x)$ in register $\reg{D}$. \hlabel{line:au}
		\EndCtrl
		\EndFor
		\State \Return $\reg A_{[n]}$.
	\end{algorithmic}
\end{algorithm}

For this discussion and for the rest of the overview, we ignore issues of precision for the sake of simplicity. By induction, it is easy to see that the state in $\reg A_{[k]}$ at the end of the $k$'th iteration is $\ket{\psi^{(k+1)}}$ if $0 \le k < n$, and $\ket{\psi}$ if $k=n$, which implies \cref{prop:oracle-state-synthesis}. It is important that the oracle output $M(x)$ is uncomputed in \Cref{line:au} because otherwise the register $\reg{A}_{[n]}$ (which is supposed to hold the state $\ket{\psi}$ at the end) will be entangled with the ancilla register $\reg{D}$.

\subsubsection{First attempt at state synthesis with an \emph{untrusted} prover} \label{sec:fassup}
\label{sec:flawed}

In this section we describe a first attempt at converting \Cref{alg:aaronson} into an interactive protocol for state synthesis, assuming that the state $\ket{\psi}$ is in $\cc{statePSPACE}$. Although this first attempt does not work, its shortcomings will illustrate the challenges that need to be overcome.

We first argue that (an approximation of) the function $M$ defined above can be computed via a classical interactive proof. Given an input $x \in \{0,1\}^{\leq n}$, a quantum polynomial-space algorithm can sample a $\mrm{Bernoulli}(p_x)$ random variable by constructing $\ket{\psi}$ (which can be synthesized in polynomial space by assumption), and performing the two-outcome projective measurement $\{\kb x, \id - \kb x\}$ on the first $|x|$ qubits. By repeating this procedure exponentially many times and averaging the results, a polynomial-space\footnote{We assume that we are working with \emph{general} quantum circuits, which allow for non-unitary operations such as tracing out registers and initializing new ones; thus the space used to construct and measure $\ket{\psi}$ can be reused.} quantum algorithm can estimate $p_x$ with high accuracy. Using that $\cc{BQPSPACE} = \PSPACE$~\cite{watrous03complexity}, it follows that $p_x$ can also be approximated by a polynomial-space (classical) Turing machine. This immediately implies that $M(x)$ can be approximated in $\PSPACE$ when $|x|<n$, and the case where $|x|=n$ can be handled similarly. Finally, by the function version of the $\IP = \PSPACE$ theorem, it follows that (an approximation of) $M$ is in the function version of $\IP$. By this, we mean that there is an interactive protocol between a verifier (who receives input $x$) and an untrusted prover, where at the end of the interaction the verifier either rejects with high probability, or outputs an approximation of $M(x)$.

This suggests the following approach to state synthesis with an untrusted prover: the verifier simulates \cref{alg:aaronson}, but instead of querying $M$, the verifier runs the $\IP = \PSPACE$ protocol \emph{in superposition} to compute the $M(x)$ values. This approach is described in \cref{alg:flawed}. 

\begin{algorithm}
	\caption{First attempt at interactive synthesis of $\ket{\psi}$ with an untrusted prover}
	\label{alg:flawed}
	\begin{algorithmic}[1]
		\State Initialize registers $\reg A_1, \dotsc, \reg A_n$ each to $\ket0$.
		\For{$k=0$ to $n$}
		\Ctrl{the state $\ket x$ of $\reg A_{[k]}$ where $x \in \cube k$,} \hlabel{line:ff1}
		\State Run the $\IP = \PSPACE$ protocol to compute $M(x)$, and write the output to a \hlabel{line:ff2}
		\Statex  \hspace{\algorithmicindent} $\ \ \ \, $ fresh register $\reg{D}_k$.
		\EndCtrl \hlabel{line:ff3}
		\State Perform the two-outcome projective measurement $\{\kb{R}, \id - \kb{R}\}$ on $\reg D_k$, and \hlabel{line:fm}
		\Statex \hspace{\algorithmicindent} reject if the outcome ``$R$" occurs.
		\Ctrl{the state $\ket m$ of $\reg D_k$ where $m \in \cube{\poly(n)}$ encodes a number,}
		\If{$k < n$} construct $\sqrt{m} \ket 0 + \sqrt{1-m} \ket 1$ in $\reg A_{k+1}$.
		\ElsIf{$k=n$} apply the phase $m$.
		\EndIf
		\EndCtrl
		\State Reverse the computations in \cref{line:ff1,line:ff2,line:ff3}. \hlabel{line:fr}
		\EndFor
		\State \Return $\reg A_{[n]}$.
	\end{algorithmic}
\end{algorithm}

We model the verifier in the $\IP = \PSPACE$ protocol for $M$ as a sequence of unitary circuits, and at the end of the interaction it outputs outputs ``R'' (for ``reject'') or a string that is supposed to encode the value $M(x)$ (if the input to the verifier was $x$). In \Cref{line:ff1,line:ff2,line:ff3} the $\IP = \PSPACE$ verifier is executed in superposition over different inputs $x \in \{0,1\}^k$. In \cref{line:fr}, we mean that if in \cref{line:ff1,line:ff2,line:ff3} the verifier applies a sequence of unitaries $U_1,\ldots,U_\ell$, then the reverse sequence $(\adj U_\ell, \dotsc, \adj U_1)$ is applied.

It is easy to show that \cref{alg:flawed} satisfies completeness: an honest prover simulates an honest prover for the $\IP = \PSPACE$ protocol in \cref{line:ff1,line:ff2,line:ff3} (i.e.\ ``the forward direction of $\IP = \PSPACE$"), and simulates the ``reverse" of that same honest prover in \cref{line:fr} (``the backward direction of $\IP = \PSPACE$"). Between the forward and backward directions of $\IP = \PSPACE$, controlled on the state $\ket x$ of $\reg A_{[k]}$ where $x \in \cube k$, the register $\reg D_k$ is in the state $\ket{M(x)}$ (since the prover is honest) and the rest of the argument is similar to that for \cref{alg:aaronson}.

However, \cref{alg:flawed} is not sound. Although the measurement in \cref{line:fm} guards against the prover performing \emph{classical} cheating attacks in the forward direction of the $\IP = \PSPACE$ protocol (by which we mean convincing the verifier on input $x$ to accept a value different from $M(x)$), there is nothing preventing the prover from mounting \emph{quantum} attacks on the protocol. More specifically, in what is supposed to be the reverse direction of $\IP = \PSPACE$, the prover need not actually reverse what it did in the forward direction. This allows the prover to entangle itself with the register $\reg A_{[n]}$, causing the verifier to output the first $n$ qubits of an entangled state of the form $\sum_{x \in \cube n} \alpha_x \ket{x} \otimes \ket{\phi_x}$ where the $\{ \ket{\phi_x} \}$ states are in the possession of the prover. If the states $\ket{\phi_x}$ are all orthogonal, for example, then the verifier's output is the density matrix $\sum_{x \in \cube n} |\alpha_x|^2 \, \kb{x}$. We call this an \emph{entanglement attack} by the prover.

Even if the prover does not keep any qubits at the end of the interaction, the prover can still introduce ``spurious phases'' into the state; for example during the forward direction of $\IP = \PSPACE$ in the final iteration, controlled on the input $x$, the prover can apply a phase $\gamma_x$, and do everything else honestly. The resulting state of the verifier will be $\sum_{x \in \cube n} \gamma_x \, \alpha_x \, \ket{x}$, which could even be orthogonal to $\ket{\psi}$. We call this a \emph{phase attack}.

\subsubsection{Defending against quantum attacks} \label{sec:daqa}

In order to defend against quantum attacks there needs to be an inherently quantum test. Our next protocol, described in \Cref{alg:good}, augments the protocol of \Cref{alg:flawed} with the simple but powerful \emph{swap test}. The swap test is a quantum algorithm that, given a bipartite state on two registers $\reg{A}, \reg{B}$ of the same size, performs the two-outcome measurement $\{S,\id - S\}$ where $S$ is the projector onto the symmetric subspace of the two registers. To get some intuition for why the swap test is useful, observe that if one of the registers (say $\reg{B}$) is in a pure state $\ket{\phi}$, but the other register is in a mixed state $\rho$ (because it may be entangled with an external system), then the swap test will reject with probability $\frac{1}{2} - \frac{1}{2} \bra{\phi} \rho \ket{\phi}$, which is large if $\rho$ is more mixed (i.e.\ more entangled with the external system). The swap test only accepts with high probability if $\rho$ is close to being the pure state $\ket{\phi}$.

\Cref{alg:good} is similar to \cref{alg:flawed}, but keeps \emph{two} copies of the state under construction and uses swap tests to guard against the entanglement and phase attacks. There are $3n$ rounds\footnote{The choice of $3$ is arbitrary; the protocol works as long as the number of rounds is $(2 + c)n$ for any constant $c > 0$.} in the protocol, and each round is randomly chosen to be either a ``test'' round or a ``grow'' round. What type the $h$-th round is supposed to be is indicated by the bit $b_h$. The ``grow counter'' $k$ specifies how many qubits have been grown so far by the protocol. The invariant that is (intended to be) maintained is that the registers $\reg{A}_{[k]}$ and $\reg{B}_{[k]}$ each contain a copy of the intermediate state $\ket{\psi\uppart{k}}$ defined in~\eqref{eq:intermediate-state}. 

\begin{algorithm}[t]
	\caption{Interactive synthesis of $\ket{\psi}$ with an untrusted prover (informal)}
	\label{alg:good}
	\begin{algorithmic}[1]
		\State Initialize registers $\reg A_1, \dotsc, \reg A_n$ and $\reg B_1, \dotsc, \reg B_n$ each to $\ket0$.
		\State $k \gets 0$.
		\For{$h=1$ to $3n$}
		\IIf{$k = n+1$} go to \Cref{line:g-end}.
		\EndIIf
		\Ctrl{the state $\ket x$ of $\reg A_{[k]}$ where $x \in \cube k$,} \hlabel{line:gf1}
		\State Run the $\IP = \PSPACE$ protocol to compute $M(x)$, and write the output to a \hlabel{line:gf2}
		\Statex \hspace{\algorithmicindent} $\ \ \ \, $ fresh register $\reg{D}_h$.
		\EndCtrl \hlabel{line:gf3}
		\State Sample a uniformly random $b_h \in \{0,1\}$.
		\If{$b_h = 1$}
		\State Perform the two-outcome projective measurement $\{\kb{R}, \id - \kb{R}\}$ on $\reg D_h$,
		\Statex \hspace{\algorithmicindent} $\ \ \ \, $ and reject if the outcome ``$R$" occurs.
		\Ctrl{the state $\ket m$ of $\reg D_h$ where $m \in \cube{\poly(n)}$ encodes a number,}
		\If{$k < n$} construct $\sqrt{m} \ket 0 + \sqrt{1-m} \ket 1$ in $\reg A_{k+1}$. \hlabel{line:grow-k-small}
		\ElsIf{$k=n$} apply the phase $m$.
		\EndIf
		\EndCtrl
		\EndIf
		\State Reverse the computations in \cref{line:gf1,line:gf2,line:gf3}. Then announce $b_h$ to the prover.  \hlabel{line:gr}
		\State $r \gets \min(k+b_h,n)$.
		\IIf{$b_h = 1$} send the register $\reg B_{[r]}$ to the prover, and receive it back. \EndIIf \hlabel{line:gift}
		\State Perform the swap test between $\reg A_{[r]}$ and $\reg B_{[r]}$, and reject if this fails. \hlabel{line:swap-ov}
		\State $k \gets k + b_h$.
		\EndFor
		\State \Return $\reg A_{[n]}$. \hlabel{line:g-end}
	\end{algorithmic}
\end{algorithm}

This protocol is complete: if the prover is honest, then at the beginning of the $h$'th iteration with $k = b_1 + \dotsb + b_{h-1} \le n$, the registers $\reg A_{[k]}$ and $\reg B_{[k]}$ each hold copies of the intermediate state $\ket{\psi\uppart{k}}$, and similarly once $k$ reaches $n+1$ the registers $\reg A_{[n]}$ and $\reg B_{[n]}$ each hold copies of $\ket{\psi}$. The honest prover behaves similarly to the honest prover from \cref{sec:flawed}: if $h$ is a grow round (i.e.\ $b_h=1$) then the verifier grows the state in register $\reg{A}_{[k+1]}$ by one qubit to get $\ket{\psi\uppart{k+1}}$ (or if $k=n$, applies the phases to the state in $\reg{A}_{[n]}$ to get $\ket{\psi}$). In a grow round, in \cref{line:gift} the prover receives the register $\reg{B}_{[k+1]}$, replaces its contents with $\ket{\psi\uppart{k+1}}$, and sends the register back (or if $k=n$, replaces $\reg{B}_{[n]}$ with $\ket{\psi}$). Thus both registers $\reg{A}_{[r]}$ and $\reg{B}_{[r]}$ with $r = \min(k+b_h,n)$ have the same pure state, and thus pass the swap test. On the other hand, if it is a test round (i.e.\ $b_h = 0$), then the $h$'th iteration has no effect -- the honest prover performs the $\IP = \PSPACE$ protocol both forwards and then reverses its actions. The swap test still passes because the registers $\reg{A}_{[r]}$ and $\reg{B}_{[r]}$ still have the same state.

By a Chernoff bound, the counter $k$ reaches $n+1$ after $3n$ rounds with high probability, which establishes completeness.

Now we give some intuition for why soundness holds. A malicious prover cannot perform either entanglement or phase attacks in the $\IP = \PSPACE$ protocols without being caught with some probability. This is because the prover does not know whether it is a test or grow round. If it is a test round (which happens with probability $1/2$), then an entanglement (or phase) attack would introduce an asymmetry between registers $\reg A_{[k]}$ and $\reg B_{[k]}$, which would be exposed by the swap test. 

Similarly, a malicious prover cannot cheat in the grow round in \Cref{line:gift} either, because if it replaces the register $\reg{B}_{[r]}$ with anything other than a copy of the state in $\reg{A}_{[r]}$, then it would again be detected by the swap test.

We obtain our final $\cc{stateQIP}$ protocol for synthesizing $\ket{\psi}$ by \emph{amplifying} the protocol given in \Cref{alg:good}. In the amplified protocol, polynomially many instances of \Cref{alg:good} are executed, and the verifier accepts only if all instances accept. In that case, the output state of a random instance is returned as output.

A simpler candidate approach (that does not work) would be for the prover to send the verifier a register $\reg A$ which allegedly holds a copy of $\ket\psi$, and then for the verifier to perform \cref{alg:flawed} with output register $\reg B$ and afterwards perform the swap test on registers $\reg A$ and $\reg B$. This would guard against some entanglement attacks, but would not guard against phase attacks, because the prover could ensure that $\reg A$ and $\reg B$ both hold copies of the same state $\sum_x \gamma_x \alpha_x \ket x$ for arbitrary phases $(\gamma_x)_x$.

\subsection{Unitary synthesis}
We now discus the Interactive Unitary Synthesis problem in the special case of unitary families with polynomial action. The proof of \Cref{thm:poly-action-intro} proceeds via reduction to the Interactive State Synthesis problem. We observe that every unitary $U$ that acts nontrivially on a subspace of dimension $d$ can be expressed as the \emph{Hamiltonian evolution} $e^{i \rho t}$ of a mixed state $\rho$ for time $t \le O(d)$. Furthermore, using the Hamiltonian simulation algorithm of Lloyd, Mohseni and Rebentrost~\cite{lloyd2014quantum}, the evolution operator $e^{i \rho t}$ can be efficiently implemented using $\poly(t)$ copies of the state $\rho$. We call the states $\rho$ \emph{program states}, as they represent the ``quantum program'' $U$ in quantum state form. 

We then show that a family $(C_n)_{n \in \N} \in \unitaryPSPACE$ of circuits with polynomial action induces a family of states $(\ket{\rho_n})_{n \in \N} \in \statePSPACE$, where $\ket{\rho_n}$ has $\poly(n)$ qubits\footnote{This is a slight abuse of notation, since the definition of \statePSPACE requires that the $n$'th state in each sequence have exactly $n$ qubits.} and the reduced state on the first $n$ qubits is the program state $\rho_n$ associated with the circuit $C_n$ (i.e.\ $\ket{\rho_n}$ is a \emph{purification} of $\rho_n$). Furthermore, the time $t_n$ associated with $\rho_n$ can be computed in space $\poly(n)$ as well.

Thus, to implement the circuit $C_n$ on an input state $\ket{\phi}$, the verifier simply needs to run the state synthesis protocol of \Cref{thm:main-intro} for the state family $(\ket{\rho_n})_n$ to synthesize $\poly(n)$ copies of $\ket{\rho_n}$ and compute $t_n$, and then use the Hamiltonian simulation algorithm of~\cite{lloyd2014quantum} to approximately implement $e^{i \rho_n t_n}$ on $\ket{\phi}$.

\subsection{Candidate constant-round protocols for state and unitary synthesis} \label{sec:rrsus}

As promised in \cref{sec:opfd}, here we sketch constant-round variants of our state and unitary synthesis protocols, which we conjecture to satisfy completeness and soundness as well. In fact, a constant-round unitary synthesis protocol follows immediately from a constant-round state synthesis protocol, because all interaction between the verifier and prover in our unitary synthesis protocol occurs during the reduction to state synthesis. Thus, all that remains is to sketch a constant-round state synthesis protocol.

Recall the intermediate states $\ket{\psi \uppart 1}, \dotsc, \ket{\psi \uppart n}$ defined in \cref{sec:state-synthesis-overview}. First the prover sends registers $\reg R_1, \dotsc, \reg R_n, \reg S_1, \dotsc, \reg S_n, \reg S_{n+1}$ to the verifier, where each register $\reg R_j$ or $\reg S_j$ is on $j$ qubits, except for $\reg S_{n+1}$ which is on $n$ qubits. If the prover is honest, then each register $\reg R_j$ or $\reg S_j$ holds a copy of $\ket{\psi \uppart j}$, except for $\reg S_{n+1}$ which holds a copy of $\ket\psi$.

Then the verifier proceeds as follows: Sample $k \in [n]$ uniformly at random, and perform the swap test on $\reg R_k$ and $\reg S_k$. If the swap test fails then reject, otherwise perform a single iteration of the for loop from \cref{alg:good} with this parameter $k$, with the following two modifications: First, substitute $\reg R_k$ for $\reg A_{[k]}$ on \cref{line:gf1}. Second, in place of \cref{line:gift,line:swap-ov}, if this is a ``test round" ($b_h = 0$) then perform the swap test on $\reg R_k$ and $\reg S_k$, and if this is a ``grow round" ($b_h = 1$) then perform the swap test on $\reg R_k \reg A_{k+1}$ and $\reg S_{k+1}$, where $\reg R_k \reg A_{k+1}$ denotes the concatenation of registers $\reg R_k$ and $\reg A_{k+1}$. Again, reject if this swap test fails.

If the verifier has not yet rejected, then the verifier accepts and outputs $\reg S_{n+1}$. Finally, our soundness amplification procedure can be parallelized in a natural way as well.

	\section{Preliminaries}
\label{sec:prelims}
Let $\cube{\le n} = \bigcup_{0 \le k \le n} \cube k$ and $\cube{<n} \bigcup_{0 \le k < n} \cube k$ (where $\cube0$ contains the empty string), and let $\cube* = \cube{<\infty}$. For a string $x = (x_1, \dotsc, x_n) \in \cube n$, let $|x| = n$ denote its length, let $\hw(x) = \sum_j x_j$ denote its Hamming weight, let $x_{\le k} = (x_1, \dotsc, x_k)$ (for $0 \le k \le n$), and let $x_{<k} = (x_1, \dotsc, x_{k-1})$ (for $1 \le k \le n+1$). For strings $x,y \in \cube*$ let $xy$ denote their concatenation. For a finite set $S$, we write $x \unif S$ to denote that $x$ is sampled uniformly at random from $S$.

When we refer to polynomials, we mean time-constructible functions $p : \N \to \N$ (with inputs represented in unary) for which there exists a constant $c$ such that for all $n \in \N$ it holds that $p(n) \le cn^c$. For polynomials $p,q$ we write $p \leq q$ if $p(n) \leq q(n)$ for all $n \in \N$.

Throughout this paper, our algorithms represent real numbers using \emph{dyadic rationals}. For $m \in \N$ let
\begin{equation*}
	\D{m} = \{k2^{-m} : k \in \N, 0 \le k < 2^m\} \subset \Q \cap [0,1)~.
\end{equation*}
A number $k2^{-m} \in \D{m}$ can be encoded as a string in $\cube m$ via the standard binary representation of $k$. We will use the fact that for all $x \in [0,1]$, there exists $r \in \D{m}$ such that $|x - r| \le 2^{-m}$. Also for $m \in \N$ let
\newcommand*{\U}[1]{\mathbb U_{#1}}
\begin{equation*}
	\U m = \{\exp(2\pi i r) : r \in \D m\},
\end{equation*}
i.e.\ $\U m$ is a discretization of the set of complex numbers with magnitude 1. A number $\exp(2\pi i r) \in \U m$ can be encoded as a string in $\cube m$ via the aforementioned encoding of $r$.

We use the fact that for $a,b \in \R$,
\begin{equation} \label{eq:integral}
	\left| e^{ib} - e^{ia} \right|
	\le |b-a|,
\end{equation}
which holds because if $a \le b$ then by the triangle inequality
\begin{equation*}
	\left| e^{ib} - e^{ia} \right|
	= \left| \int_a^b i e^{ix} dx \right|
	\le \int_a^b \left| i e^{ix} \right| dx
	= b-a,
\end{equation*}
and the case $b \ge a$ is symmetric. We also use the following fact:

\begin{lem} \label{lem:4p}
	If $a,b,c$ are real numbers satisfying
	\begin{align*}
		&0 \le a \le 1,
		&b \ge 0,
		&&c \ge a-b,
	\end{align*}
	then $c^4 \ge a^4 - 4b$.
\end{lem}
\begin{proof}
	If $a \ge b$ then
	\begin{equation*}
		c^4 \ge (a-b)^4 = a^4 \Paren{1 - \frac{b}{a}}^4 \ge a^4 \Paren{1 - 4\frac{b}{a}} = a^4 - 4ba^3 \ge a^4 - 4b~.
	\end{equation*}
	Alternatively, if $a \le b$ then
	\begin{equation*}
		c^4 \ge 0 \ge a-b \ge a^4 - 4b~.
	\end{equation*}
	Either way the claim holds.
\end{proof}

\paragraph{Turing Machines}

In this paper all Turing machines have an input tape, a work tape, and an output tape, where the output tape head is not allowed to move left. We say that a Turing machine on an input $x$ uses space $s$ if the number of non-blank squares on its work tape at any point in the computation is at most $s$ (the input and output tapes do not count toward its space usage). We say that $M$ is a polynomial-space Turing machine if there exists a polynomial $p$ such that on every input $x \in \cube*$, the machine $M$ uses at most $p(|x|)$ space.

\paragraph{Quantum Information Theory}
A \emph{register} $\reg{R}$ is a named finite-dimensional complex Hilbert space. If $\reg{A}, \reg{B}, \reg{C}$ are registers, for example, then the concatenation $\reg{A} \reg{B} \reg{C}$ denotes the tensor product of the associated Hilbert spaces. For a linear transformation $L$ and register $\reg R$, we sometimes write $L_{\reg R}$ to indicate that $L$ acts on $\reg R$, and similarly we sometimes write $\rho_{\reg R}$ to indicate that a state $\rho$ is in the register $\reg R$. We write $\tr(\cdot)$ to denote trace, and $\tr_{\reg R}(\cdot)$ to denote the partial trace over a register $\reg R$.

For a pure state $\ket\varphi$, we write $\varphi$ to denote the density matrix $\kb\varphi$. Let $\id$ denote the identity transformation, and let $\ket+, \ket-$ denote the Hadamard basis states. Let $\norm\cdot$ denote the vector 2-norm, let $\norm\cdot_1$ denote the vector 1-norm, and let $\norm\cdot_{\mrm{op}}$ denote the operator 2-norm (i.e.\ the maximum singular value of the matrix).

Let $\td(\rho,\sigma)$ denote the trace distance between two density matrices $\rho,\sigma$, i.e.\ $\td(\rho, \sigma) = \max_{0 \le P \le \id} \tr((\rho - \sigma)P)$. Trace distance is symmetric (i.e., $\td(\rho,\sigma) = \td(\sigma,\rho)$), convex (i.e., $\td(\alpha_1 \rho_1 + \alpha_2 \rho_2, \sigma) \leq \alpha_1 \td(\rho_1, \sigma) + \alpha_2 \td(\rho_2, \sigma)$ where $\alpha_1,\alpha_2$ are probabilities summing to $1$), and satisfies the triangle inequality (i.e., $\td(\rho,\sigma) \leq \td(\rho,\tau) + \td(\tau,\sigma)$). For a pure state $\ket\varphi$ and mixed state $\sigma$, it holds that $\td(\varphi, \sigma) \le \sqrt{1 - \bra\varphi \sigma \ket\varphi}$~\cite[Chapter 9]{nielsen2000quantum}, and therefore
\begin{equation} \label{eq:td-fid}
	\td(\varphi, \sigma) \le \sqrt{2\Paren{1 - \sqrt{\bra\varphi \sigma \ket\varphi}}}~.
\end{equation}
As a corollary, if $\sigma = \kb{\sigma}$ is a pure state, then
\begin{equation} \label{eq:td-fid2}
	\td(\varphi, \sigma) \leq \norm{\ket{\varphi} - \ket{\sigma}}~.
\end{equation}

\paragraph{Families of Quantum Circuits and States}

For convenience we assume that all quantum circuits use gates from the universal gate set $\{ H, \mathit{CNOT}, T \}$~\cite[Chapter 4]{nielsen2000quantum} (although our results hold for any universal gate set consisting of gates with algebraic entries). A \emph{unitary quantum circuit} is one that consists only of gates from this gate set. A \emph{general quantum circuit} is a quantum circuit that can additionally have non-unitary gates that (a) introduce new qubits initialized in the zero state, (b) trace them out, or (c) measure them in the standard basis. We say that a general quantum circuit uses space $s$ if the total number of qubits involved at any time step of the computation is at most $s$. The description of a general quantum circuit is a sequence of gates (unitary or non-unitary) along with a specification of which qubits they act on.

\begin{dfn}[Polynomial size and space circuit families]
	We say that $(C_n)_{n \in \N}$ is a family of \emph{polynomial-size general quantum circuits} if there exists a polynomial $p$ such that $C_n$ has size at most $p(n)$. We say that $(C_n)_{n \in \N}$ is a family of \emph{polynomial-space general quantum circuits} if there exists a polynomial $p$ such that $C_n$ uses at most $p(n)$ space.
\end{dfn}

\begin{dfn}[Uniform circuit families]
	We say that a family $(C_n)_{n \in \N}$ of polynomial-size general quantum circuits is \emph{time-uniform} (or simply \emph{uniform}) if there exists a polynomial-time Turing machine that on input $1^n$ outputs the description of $C_n$. Similarly, we say that a family $(C_n)_{n \in \N}$ of polynomial-space general quantum circuits is \emph{space-uniform} if there exists a polynomial-space Turing machine that on input $1^n$ outputs the description of $C_n$.
\end{dfn}

We define the following ``state complexity class" and ``unitary complexity class":

\begin{dfn}[\statePSPACE]
	\statePSPACE is the class of all sequences $(\ket{\psi_n})_{n \in \N}$ such that each $\ket{\psi_n}$ is a state on $n$ qubits, and for every polynomial $q$ there exists a space-uniform family of general quantum circuits $(C_n)_{n \in \N}$ such that for all $n \in \N$, the circuit $C_n$ takes no inputs and $C_n$ outputs a density matrix $\rho$ such that $\td(\rho, \psi_n) \leq \exp(-q(n))$.
\end{dfn}

\begin{dfn}[\unitaryPSPACE]
	\unitaryPSPACE is the class of all space-uniform sequences $(U_n)_{n \in \N}$ of unitary quantum circuits such that each $U_n$ acts on $n$ qubits.
\end{dfn}

The error term $\exp(-q(n))$ in the definition of \statePSPACE is necessary for our reduction from unitary synthesis to state synthesis.

	\section{Quantum Interactive Protocols and Related Complexity Classes}
\label{sec:protocols}
\subsection{Quantum interactive protocols}

Since in quantum computing the standard model of computation is the quantum circuit model (rather than quantum Turing machines), we model the verifier in a quantum interactive protocol as a sequence of \emph{verifier circuits}, one for each input length. A verifier circuit is itself a tuple of quantum circuits that correspond to the operations performed by the verifier in each round of the protocol.

More formally, a \emph{$k$-round quantum verifier circuit} $C = (C_j)_{j \in [k]}$ is a tuple of general quantum circuits that each act on a pair of registers $(\reg{V},\reg{M})$. The register $\reg{V}$ is further divided into disjoint sub-registers $(\reg{V}_\work, \reg{V}_\flag, \reg{V}_\out)$. The register $\reg V_\work$ is the verifier circuit's ``workspace", the register $\reg V_\flag$ is a single qubit indicating whether the verifier accepts or rejects, and the register $\reg V_\out$ holds the verifier's output (if applicable). The register $\reg{M}$ is the message register. The size of a verifier circuit $C$ is the sum of the circuit sizes of the $C_j$'s.

A \emph{quantum prover} $P$ for a verifier circuit $C$ is a unitary that acts on $\reg{M}$ as well as a disjoint register $\reg{P}$. 

Let $x \in \{0,1\}^*$ denote a string whose length is at most the number of qubits in $\reg{V}_\work$. We write $C(x) \interact P$ to denote the interaction between the verifier circuit $C$ and the prover $P$ on input $x$, which is defined according to the following process. The initial state of the system is $\ket{\phi_0} = \ket{x,\zs}_{\reg V_\work} \ket\zs_{\reg{V}_\flag \reg{V}_\out \reg{M} \reg{P}}$. Inductively define $\ket{\phi_i} = P \ket{\phi_{i-1}}$ for odd $i \leq 2k$, and $\ket{\phi_i} = C_{i/2} \ket{\phi_{i-1}}$ for even $i \leq 2k$. We say that $C(x) \interact P$ accepts (resp. rejects) if measuring the register $\reg{V}_\flag$ in the standard basis yields the outcome $1$ (resp. $0$). We say that the \emph{output of $C(x) \interact P$ conditioned on accepting} is the density matrix
\[
\frac{\tr_{\reg{V} \reg{M} \reg{P} \setminus \reg{V}_\out} \left( \ketbra{1}{1}_{\reg{V}_\flag} \cdot \phi_{2k} \right)}{\tr \left( \ketbra{1}{1}_{\reg{V}_\flag} \cdot \phi_{2k} \right)}~;
\]
in other words, it is the reduced density matrix of $\ket{\phi_{2k}}$ on register $\reg{V}_\out$, conditioned on $C(x) \interact P$ accepting. (If the probability of accepting is $0$, then we leave the output undefined.) 

A \emph{quantum verifier} $V = (V_n)_{n \in \N}$ is a uniform sequence of polynomial-size and polynomial-round quantum verifier circuits.

\subsection{The class \texorpdfstring{$\cc{QIP}$}{QIP}}

The following is the standard quantum analogue of the complexity class $\IP$:

\begin{dfn}[$\QIP$]
	Let $L \subseteq \{0,1\}^*$ be a language, let $s : \N \to \R_{\geq 0}$ be a function, and let $V = (V_n)_{n \in \N}$ be a quantum verifier. Then $V$ is a \emph{$\QIP[s]$ verifier for $L$} if and only if
	\begin{itemize}[leftmargin=*]
		\renewcommand\labelitemi{--}
		\item \emph{Completeness:} For all $x \in L$, there exists a quantum prover $P$ (called an \emph{honest prover}) such that
		\begin{equation*}
			\pr {V_{|x|}(x)  \interact P \text{ accepts}} = 1.
		\end{equation*}
		\item \emph{Soundness:} For all quantum provers $P$ and all $x \notin L$,
		\begin{equation*}
			\pr {V_{|x|}(x)  \interact P \text{ accepts}} \leq s(|x|).
		\end{equation*}
	\end{itemize}
	Here, the probability is over the randomness of the interaction. 
	
	Finally, let $\QIP[s]$ be the set of languages $L$ such that there exists a $\QIP[s]$ verifier for $L$. Let $\QIP = \QIP[1/2]$.\footnote{The reader may wonder whether the definition of $\QIP$ here is sensitive to the assumption of perfect completeness; it is known that if we use the universal gate set $\{ H, T, \mathit{CNOT} \}$, then we can assume perfect completeness without loss of generality~\cite[Section 4.2]{vidick2016quantum}.}
\end{dfn}

It is known that the soundness error of $\QIP$ protocols can be reduced to be exponentially small, at the cost of repeating the protocol polynomially many times. The next theorem characterizes the class $\QIP$, and is a quantum analogue of the $\IP = \PSPACE$~\cite{lund1992algebraic,shamir1992ip} theorem:

\begin{thm}[\cite{jain2011qip}] \label{thm:qip=pspace}
	$\QIP = \PSPACE$.
\end{thm}

\subsection{The class \texorpdfstring{$\cc{stateQIP}$}{stateQIP}}

Next, we define quantum interactive proofs for state synthesis and their associated complexity class.

\begin{dfn}[{$\stateQIP{c,s}$}]
	\label{def:stateQIP}
	Let $\Psi = (\ket{\psi_n})_{n \in \N}$ be a family of states where $\ket{\psi_n}$ is on $n$ qubits\footnote{The assumption that $\ket{\psi_n}$ is on $n$ qubits is for convenience; this definition and our results can easily be generalized to the case where each $\ket{\psi_n}$ is a state on $\poly(n)$ qubits.}, let $c: \N \to \R_{\geq 0}$ and $s:\N \times [0,1] \to \R_{\geq 0}$ be functions, and let $V = (V_n)_{n \in \N}$ be a quantum verifier. Then $V$ is a \emph{$\cc{stateQIP}[c,s]$ verifier for $\Psi$} if and only if
	\begin{itemize}[leftmargin=*]
		\renewcommand\labelitemi{--}
		\item \emph{Completeness:} For all $n \in \N$, there exists a quantum prover $P$ (called an \emph{honest prover}) such that
		\begin{equation*}
			\pr {V_n \interact P \text{ accepts}} = 1 \qquad \text{and} \qquad \td(\rho, \psi_n) \leq c(n),
		\end{equation*}
		where $\rho$ is the output of $V_n \interact P$ conditioned on accepting.
		\item \emph{Soundness:} For all $n \in \N$ and for all quantum provers $P$, 
		\begin{equation*}
			\pr {V_n \interact P \text{ accepts}} \leq s(n,\td(\rho, \psi_n)),
		\end{equation*}
		where $\rho$ is the output of $V_n \interact P$ conditioned on accepting.
	\end{itemize}
	Here, the probability is over the randomness of the interaction. We call $c$ the \emph{state completeness} and $s$ the \emph{soundness} of $V$.
	
	Finally, let $\stateQIP{c,s}$ be the set of state sequences $\Psi = (\ket{\psi_n})_{n \in \N}$ such that there exists a $\stateQIP{c,s}$ verifier for $\Psi$.
\end{dfn}

As discussed in the introduction, the class $\stateQIP{c,s}$ is a set of \emph{sequences of quantum states}, rather than languages as is typical in complexity theory. One could define $\stateQIP{c,s}$ more generally to include families of quantum states $(\ket{\psi_x})_{x \in \{0,1\}^*}$ indexed by binary strings, but for simplicity we stick with sequences of states. We call $\stateQIP{c,s}$ a \emph{state complexity class}.

To keep consistent with the definition of $\QIP$, we assume that $\cc{stateQIP}$ protocols have perfect protocol completeness, i.e.\ the honest prover is accepted with probability 1. However, we do not know whether this holds without loss of generality (like it does for $\QIP$); we leave this as an open problem.

The soundness function $s$ in \Cref{def:stateQIP} expresses how the probability of acceptance depends on the closeness of the output state (conditioned on accepting) to the target state $\ket{\psi_n}$. The next lemma shows that soundness amplification holds for $\stateQIP{c,s}$ for soundness functions $s(n,\delta)$ that for all $n \in \N$ are log-concave\footnote{A function $f:[0,1] \to \R_{> 0}$ is \emph{log-concave} if $\log f$ is concave, i.e.\ if $\log f(\alpha \delta_1 + (1 - \alpha) \delta_2) \geq \alpha \log f(\delta_1) + (1 - \alpha) \log f(\delta_2)$ for $0 \le \alpha, \delta_1, \delta_2 \le 1$.}, nonincreasing functions of $\delta$:

\begin{lem}[Soundness amplification for $\cc{stateQIP}$]
	\label{lem:stateQIP-amplification}
	Let $c: \N \to \R_{\geq 0}$ and $s: \N \times [0,1] \to \R_{\geq 0}$ be functions, where for every $n \in \N$ the function $\delta \mapsto s(n, \delta)$ is log-concave and nonincreasing. Then for all polynomials $m \ge 1$, it holds that $\stateQIP{c,s} = \stateQIP{c,s^m}$, where $s^m$ denotes the function $(n,\delta) \mapsto s(n,\delta)^{m(n)}$.
\end{lem}

\begin{proof}
	First we argue that $\stateQIP{c,s^{m}} \subseteq \stateQIP{c,s}$. Fix a $\stateQIP{c,s^m}$ verifier $V$ for a state family $\Psi = (\psi_n)_{n \in \N}$; it suffices to prove that $V$ is also a $\stateQIP{c,s}$ verifier for $\Psi$. To this end, for an arbitrary $n \in \N$ and quantum prover $P$, it suffices to prove that $V_n \interact P$ accepts with probability at most $q = s(n, \td(\rho, \psi_n))$ where $\rho$ is the output of $V_n \interact P$ conditioned on accepting. If $q \ge 1$ then this holds vacuously (because probabilities are at most 1), and otherwise it holds because $\pr{V_n \interact P \text{ accepts}} \le q^m \le q$ by the definition of $V$ and the fact that $m \ge 1$.
	
	We now argue the reverse containment $\cc{stateQIP}[c,s] \subseteq \cc{stateQIP}[c,s^{m}]$. Let $V = (V_n)_{n \in \N}$ be a $\stateQIP{c,s}$ verifier for a state family $\Psi = (\ket{\psi_n})_{n \in \N}$. Consider the following verifier $V' = (V_n')_{n \in \N}$. The verifier circuit $V_n'$ runs $m(n)$ independent instances of $V_n$ in sequence, where the $j$'th instance of $V_n$ uses output register $\reg{O}_j$ for $j \in [m(n)]$. If any instance of $V_n$ rejects, then $V'_n$ rejects. Otherwise, $V'_n$ samples a uniformly random $j \in [m(n)]$ and outputs the content of $\reg{O}_j$. This is clearly a polynomial-round, polynomial-size $\cc{stateQIP}$ verifier for $\Psi$.
	
	Completeness is straightforward: an honest prover for $V'_n$ can simulate $m(n)$ independent instances of an honest prover for $V_n$. We now prove soundness. Let $P$ denote an arbitrary prover for the verifier circuit $V_n'$. Let $\rho_j$ denote the reduced state in $\reg{O}_j$ at the end of the execution of the $j$'th instance, conditioned on the first $j-1$ instances accepting. Then we have that 
	\begin{align}
		\pr {V_n \interact P \text{ accepts}} &= \prod_{j=1}^{m(n)} \pr{\text{the $j$'th instance accepts} \mid \text{first $j-1$ instances accepted}}~ \notag \\
		&\leq \prod_{j=1}^{m(n)} s(n,\td(\rho_j,\psi_n)) \label{eq:stateQIP-amp-1},
		\intertext{where the first equality is by the definition of conditional probability. The second line is due to the fact that the $j$'th instance of $V_n$ was run independently of all previous instances, so we can consider the output of the previous instances (conditioned on having accepted) and the state of the prover $P$ at the beginning of the $j$'th instance of $V_n$ to form the state of a prover $P_j$ for $V_n$. So, by the soundness guarantee of $V$, the $j$'th probability of acceptance is at most $s(n,\td(\rho_j,\psi_n) )$. Using the log-concavity of $s(n,\delta)$ as a function of $\delta$, we get that~\eqref{eq:stateQIP-amp-1} can be upper-bounded by}
		&\leq s \left( n , \E_{j \sim [m(n)]} \td(\rho_j,\psi_n) \right)^{m(n)} \notag \\
		\intertext{which, since trace distance is convex and $s(n,\delta)$ is a nonincreasing function of $\delta$, is at most}
		&\leq s \left( n , \td \left(\E_{j \sim [m(n)]}  \rho_j,\psi_n \right) \right)^{m(n)}. \notag
	\end{align}
	The result follows, since $\E_{j \sim [m(n)]}  \rho_j$ is the density matrix that is output by $V_n'$ conditioned on accepting all instances. 
\end{proof}

\begin{rmk*}
	For functions $s: \N \times [0,1] \to \R_{\geq 0}$ that do not satisfy the preconditions of \cref{lem:stateQIP-amplification}, soundness amplification can still be achieved as follows: Find a function $s'$ that is pointwise at least $s$ and that \emph{does} satisfy the preconditions of \cref{lem:stateQIP-amplification}, and then apply soundness amplification with $s'$.
	
	For example, consider a function of the form
	\begin{equation*}
		s(n,\delta) = \left\{
		\begin{array}{ll}
			1  & \mbox{if } \delta \leq a(n) \\
			p(n) & \mbox{otherwise}
		\end{array}
		\right.,
	\end{equation*}
	for which the soundness condition has the following natural interpretation: If the verifier accepts with probability greater than $p(n)$, then the output state conditioned on accepting must be $a(n)$-close to the ideal state $\psi_n$. Then $s \le s'$ pointwise for
	\begin{equation*}
		s'(n,\delta) = \left\{
		\begin{array}{ll}
			1  & \mbox{if } \delta \leq a(n) \\
			p(n)^{(\delta - a(n)) / (1-a(n))} & \mbox{otherwise}
		\end{array}
		\right.,
	\end{equation*}
	and this function $s'$ satisfies the preconditions of \cref{lem:stateQIP-amplification}.
\end{rmk*}

The main result of our paper is the following theorem, which states that every family of states in \statePSPACE has a \cc{stateQIP} protocol. This was informally presented in the introduction as \Cref{thm:main-intro}.

\begin{thm} \label{thm:prim-result}
	For every polynomial $q$, it holds that $\statePSPACE \subseteq \stateQIP{c,s}$ for
	\begin{align*}
		&c(n) = \exp(-q(n)),
		&s(n,\delta) = \exp \Paren{e^{-q(n)} - q(n) \cdot \delta^4}~.
	\end{align*}
\end{thm}

We prove \Cref{thm:prim-result} in \Cref{sec:synthesis}. \Cref{thm:prim-result} implies that for all polynomials $p$ and sequences $(\ket{\psi_n})_{n \in \N} \in \statePSPACE$, there exists a $\cc{stateQIP}$ verifier such that for all $n$, if the output state conditioned on accepting is further than $1/p(n)$ in trace distance from $\psi_n$, then the verifier accepts with probability less than $1/2$. We leave it as an open problem to improve this statement to negligible functions $p$:
\begin{prb}
	Prove that $\statePSPACE \subseteq \stateQIP{c,s}$ for some functions $c$ and $s$ such that, for some function $\eps(n) = n^{-\omega(1)}$,
	\begin{equation*}
		\forall n \in \N, \delta \in [\eps(n), 1]: \ s(n,\delta) \le 1/2~.
	\end{equation*}
\end{prb}

We also show a partial converse (which was informally presented in the introduction as \Cref{thm:converse-intro}) to \Cref{thm:prim-result}:

\begin{thm} \label{thm:sec-result}
	Let $c,e: \N \to \R_{\geq 0}$ and $s: \N \times [0,1] \to \R_{\geq 0}$ be functions such that $c \geq e$ and $s(n,\delta)$ is a nonincreasing function of $\delta$ for all $n \in \N$. Let $(\ket{\psi_n})_{n \in \N}$ be a state family in $\cc{stateQIP}[c,s]$. Then for all polynomials $m$ there exists a polynomial $q$ and a uniform family of quantum circuits $(C_n)_{n \in \N}$ such that each circuit $C_n$ takes in no inputs and has size at most $\exp(q(n))$, and furthermore the output $\rho_n$ of $C_n$ satisfies
	\[
	\td(\rho_n,\psi_n) \leq \delta(n) + \exp(-m(n))
	\]
	where $\delta: \N \to \R_{\ge 0}$ is a function such that $s(n,\delta(n)) \leq e(n)$.	
\end{thm}

The theorem is phrased abstractly in terms of the completeness and state soundness functions $c$ and $s$; for intuition, suppose that $s(n,\delta)$ is a function such that for all $\delta \geq 1/\poly(n)$, $s(n,\delta) \leq \exp(-n)$ (which is the parameter regime of \Cref{thm:prim-result}). Then by setting $e(n)$ to be $\exp(-n)$, we get that the circuit $C_n$ synthesizes $\ket{\psi_n}$ up to $1/\poly(n)$ error.

\begin{proof}
	The main idea of the proof follows the argument that $\QIP \subseteq \mathsf{EXP}$. Consider the $\cc{stateQIP}[c,s]$ verifier for the state family $(\ket{\psi_n})_{n \in \N}$. Fix an index $n \in \N$, and let $k$ denote the number of rounds of the verifier on input $1^n$. There is a natural semidefinite program that captures the behavior of the interactive protocol:
	the program optimizes over a set of $2k$ density matrices that encode the state of the verifier and message registers at the beginning and end of each round, and checks that these density matrices are all consistent with each other. For example, in a $k=1$ round protocol, there is a density matrix $\sigma_0$ corresponding to the initial state of the verifier right before it sends its message, and there is a density matrix $\sigma_1$ that denotes the state of the verifier after the prover sends back its message, and the verifier measures $\sigma_1$ to determine whether to accept or reject. It turns out that any set of density matrices satisfying the initialization and consistency conditions corresponds to a valid interaction between the verifier and some prover (see~\cite[Section 4.3]{vidick2016quantum} for a detailed proof). The objective of the semidefinite program is to maximize the acceptance probability of the verifier.
	
	A feasible solution with value $\alpha$ is a sequence of density matrices $(\sigma_0,\sigma_1,\ldots,\sigma_{2k-1})$ where $\sigma_{2k-1}$ is the final state of the verifier, and measuring this state yields the accept state with probability $\alpha$. Suppose that $\alpha > e(n)$. Then by the soundness of the $\cc{stateQIP}$ verifier, the state $\sigma_{2k-1}$ (after being projected into the accept state) must contain in an output register a state $\sigma_{\mrm{out}}$ that is $\delta(n)$-close in trace distance to the target state $\ket{\psi_n}$. This is because $s(n, \delta(n)) \le e(n) < \alpha \le s(n, \td(\sigma_{\mrm{out}}, \psi_n))$, and because the soundness function $s(n, \delta)$ is non-increasing in $\delta$. Furthermore, since $c(n) \geq e(n)$, there must exist a feasible solution with value at least $e(n)$.
	
	Since these density matrices are on $\poly(n)$ qubits, the size of the semidefinite program is $\exp(\poly(n))$. The semidefinite program can be solved up to precision $\exp(-m(n))$ in $\exp(\poly(n,m(n)))$ time to obtain a feasible solution with value at least $e(n)$. Given the last density matrix $\sigma_{2k-1}$, a classical description of a $\left(\delta(n)+\exp(-m(n)) \right)$-approximation of $\ket{\psi_n}$ can be computed, and then using an algorithm like that of \Cref{alg:aaronson}, the state $\ket{\psi_n}$ can be synthesized up to $\delta(n) + \exp(-m(n))$ error.
\end{proof}

This leaves open the intriguing question of whether a full converse to \Cref{thm:prim-result} can be proved --- in other words, whether $\cc{stateQIP} = \cc{statePSPACE}$.

\begin{prb}
	Show that $\cc{stateQIP}[c,s] \subseteq \statePSPACE$ for some nontrivial completeness and soundness functions $c,s$.
\end{prb}

\subsection{The class \texorpdfstring{$\cc{unitaryQIP}$}{unitaryQIP}}

Next, we define quantum interactive proofs for unitary synthesis and their associated complexity class. In the following definition we write $V_n(\ket{\psi})$ to indicate that the input to the verifier circuit $V_n$ is an $n$-qubit state $\ket{\psi}$ rather than an $n$-bit string (as in the case of $\QIP$).

\begin{dfn}[{$\cc{unitaryQIP}[c,s]$}]
	Let $U = (U_n)_{n \in \N}$ be a sequence of unitary transformations where $U_n$ acts on $n$ qubits, let $c: \N \to \R_{\geq 0}$ and $s:\N \times [0,1] \to \R_{\geq 0}$ be functions, and let $V = (V_n)_{n \in \N}$ be a quantum verifier. Then $V$ is a \emph{$\cc{unitaryQIP}[c,s]$ verifier for $U$} if and only if	
	\begin{itemize}[leftmargin=*]
		\renewcommand\labelitemi{--}
		\item \emph{Completeness:} For all $n \in \N$, there exists a quantum prover $P$ (called an \emph{honest prover}) such that for all $n$-qubit states $\ket{\psi}$,
		\begin{equation*}
			\pr {V_n(\ket{\psi}) \interact P \text{ accepts}} = 1 \qquad \text{and} \qquad \td\Paren{\rho, U \psi U^\dagger} \leq c(n),
		\end{equation*}
		where $\rho$ is the output of $V_n (\ket{\psi})\interact P$ conditioned on accepting.
		\item \emph{Soundness:} For all $n \in \N$, for all $n$-qubit states $\ket{\psi}$, and for all quantum provers $P$, 
		\begin{equation*}
			\pr {V_n(\ket{\psi}) \interact P \text{ accepts}} \leq s\Paren{n, \td \Paren{\rho, U \psi U^\dagger}},
		\end{equation*}
		where $\rho$ is the output of $V_n (\ket{\psi}) \interact P$ conditioned on accepting.
	\end{itemize}
	Here, the probability is over the randomness of the interaction. We call $c$ the \emph{unitary completeness} and $s$ the \emph{soundness} of $V$.
	
	Finally, let $\cc{unitaryQIP}[c,s]$ be the set of unitary sequences $U = (U_n)_{n \in \N}$ such that there exists a $\cc{unitaryQIP}[c,s]$ verifier for $U$.
\end{dfn}

We call $\cc{unitaryQIP}$ a \emph{unitary complexity class}. 

One can consider variants of this definition; for example, the verifier could get multiple copies of an input state $\ket{\psi}$, which may make it easier to verifiably perform the unitary transformation on $\ket{\psi}$. One could also consider a variant where the input state $\ket{\psi}$ is sampled from some distribution; thus even an adversarial prover may not have full knowledge of $\ket{\psi}$. Each of these variations opens interesting and well-motivated questions that we leave to future work.

Unlike with $\cc{stateQIP}$, we do not know analogous basic properties of $\cc{unitaryQIP}$ such as whether soundness amplification is possible, or whether the unitary families in $\cc{unitaryQIP}$ are (approximately) implementable by a family of circuits even of exponential size. As alluded to in the introduction, it appears more challenging to establish such statements because of the fragility of the quantum input to the unitary synthesis problem.

For example, consider the problem of amplifying the soundness of a $\cc{unitaryQIP}$ protocol for a unitary family $(U_n)_{n \in \N}$. Typical approaches to soundness amplification, such as repeating the protocol a number of times, don't seem to apply here: the verifier only receives one copy of the input state $\ket{\phi}$ that it has to apply the unitary $U_n$ to, and due to the No-Cloning Theorem the verifier cannot obtain multiple copies of $\ket{\phi}$ to run independent instances of the protocol on.

We thus leave the following as open questions:

\begin{prb}
	Obtain soundness amplification for $\cc{unitaryQIP}$ (or give evidence that it is impossible).
\end{prb}

We note that for the aforementioned variants of the definition of $\cc{unitaryQIP}$ where the verifier has access to multiple copies of the input state, or the input state is sampled from some efficiently sampleable distribution, soundness amplification is possible.

\begin{prb}
	Show that if a family $(U_n)_{n \in \N}$ of unitaries is in $\cc{unitaryQIP}$, then there exists a space-uniform family $(C_n)_{n \in \N}$ of quantum circuits that approximately implements $(U_n)_n$, in the sense that $\norm{C_n - U_n}_{\mrm{op}} \le \delta(n)$ for some function $\delta$ related to the completeness and soundness of the $\cc{unitaryQIP}$ protocol.
\end{prb}

\begin{prb}
	\label{prb:unitaryQIP}
	Show that $\unitaryPSPACE \subseteq \cc{unitaryQIP}[c,s]$ for nontrivial completeness and soundness functions $c,s$. 
\end{prb}

In this paper we solve a special case of \Cref{prb:unitaryQIP}. Let $U = (U_n)_{n \in \N}$ be a family of unitaries where $U_n$ acts on $n$ qubits. We say that $U$ has \emph{polynomial action} if for all $n \in \N$, the unitary $U_n$ acts nontrivially on a subspace $S_n \subset (\C^2)^{\otimes n}$ of dimension at most $\poly(n)$, and acts as the identity on the orthogonal complement of $S_n$. We show that unitary sequences in \unitaryPSPACE with polynomial action admit $\cc{unitaryQIP}$ protocols. The following theorem was informally presented as \Cref{thm:poly-action-intro} in the introduction.

\begin{restatable}{thm}{polyaction}
	\label{thm:poly-action}
	Let $U = (U_n)_{n \in \N}$ be a family in \unitaryPSPACE with polynomial action, and let $q$ be a polynomial. Then $U \in \cc{unitaryQIP}[c,s]$ for
	\begin{align*}
		&c(n) = \frac{1}{q(n)},
		&s(n,\delta) = \exp \Paren{\frac{1}{q(n)} - q(n) \cdot \delta^4}~.
	\end{align*}
\end{restatable}

Note that the completeness error is larger than that of \Cref{thm:prim-result}; with the honest prover, the verifier can only guarantee that its output is $1/\poly(n)$-close to the target state. We prove \Cref{thm:poly-action} in \Cref{sec:transformations}.

	\section{\texorpdfstring{$\PSPACE$}{PSPACE} Algorithms for Quantum State Tomography}
\label{sec:tomography}
In this section we give polynomial-space algorithms for tomography of states in \statePSPACE. This plays a role in our interactive state synthesis protocol, as discussed in \cref{sec:fassup}. The results needed for our state synthesis protocol are proved in \cref{subsec:sa}, as corollaries of a more general result proved in \cref{subsec:gf}.
\subsection{General framework} \label{subsec:gf}

\begin{thm}\label{tomog-0}
	Let $\ell, m$ be polynomials, and let $(C_n)_{n \in \N}$ be a space-uniform family of general quantum circuits such that each $C_n$ has $\ell(n)$ input qubits and one output qubit. Then there exists a $\PSPACE$-computable function $f$ such that for all $n \in \N, x \in \cube{\ell(n)}$ it holds that $f(1^n, x) \in \D{m(n)}$ and
	\begin{equation*}
		\left|f(1^n, x) - \bra1 C_n(x) \ket1\right| \le 2 \cdot 2^{-m(n)}~.
	\end{equation*}
\end{thm}

The proof of \Cref{tomog-0} uses the following special case of a result of Watrous~\cite{watrous03complexity}:\footnote{Watrous's result applies only to circuit families over a finite gate set in which the entries of all gates are algebraic numbers. This is why we imposed a similar requirement in our definition of space-uniform circuit families.}

\begin{thm}[{\cite{watrous03complexity}}] \label{thm:watrous03}
	Let $\ell$ be a polynomial, and let $(C_n)_{n \in \N}$ be a space-uniform family of general quantum circuits such that each $C_n$ has $\ell(n)$ input qubits and one output qubit. Then the language
	\begin{equation*}
		\left\{(1^n, x) : n \in \N, x \in \{0,1\}^{\ell(n)}, \, \bra{1} C_n (x) \ket{1} > 1/2 \right\}
	\end{equation*}
	is in $\PSPACE$.
\end{thm}

\begin{proof}[Proof of \cref{tomog-0}]	
	The circuit $A_n$ defined in \cref{alg:An} uses $\poly(n)$ space, because it takes $\poly(n)$ space to store the value of $k$ and (by assumption) to apply $C_n$, and because all qubits introduced in an iteration of the for loop are traced out by the end of that iteration. By similar reasoning, the description of $A_n$ can be computed in space $\poly(n)$. Therefore $(A_n)_{n \in \N}$ is space-uniform, so by \cref{thm:watrous03} the language
	\begin{equation*}
		L = \left\{(1^n, x, r): n \in \N, x \in \cube{\ell(n)}, r \in \D{m(n)}, \, \bra1 A_n(x,r) \ket1 > 1/2\right\}
	\end{equation*}
	is in $\PSPACE$. Let $k_n(x)$ be a binomial random variable with parameters $t_n$ and $p_n(x) = \bra1 C_n(x) \ket1$, and observe that
	\begin{equation*}
		\bra1 A_n(x,r) \ket1 = \PR{|k_n(x)/t_n - r| \le 3/2 \cdot 2^{-m(n)}}~.
	\end{equation*}

	\begin{algorithm}
		\caption{The circuit $A_n$} \label{alg:An}
		\begin{algorithmic}[1]
			\Require $x \in \cube{\ell(n)}, r \in \D{m(n)}$
			\State $k \gets 0$.
			\For{$t_n = 10 \cdot 4^{m(n)}$ times}
				\State Construct $C_n(x)$, and measure it in the standard basis.
				\IIf{the measurement outcome is 1} $k \gets k+1$. \EndIIf
				\State Trace out the qubit holding the measurement outcome.
			\EndFor
			\If{$|k/t_n - r| \le \frac32 \cdot 2^{-m(n)}$} \Return 1.
			\Else{} \Return 0.
			\EndIf
		\end{algorithmic}
	\end{algorithm}
	
	Let $f(1^n, x)$ be the lexicographically first number $r \in \D{m(n)}$ such that $(1^n, x, r)$ is in $L$. We first argue that $f$ is well defined, i.e.\ that for all $n,x$ there exists $r$ such that $(1^n, x, r)$ is in $L$. Given $n,x$ let $r \in \D{m(n)}$ be such that $|p_n(x) - r| \le 2^{-m(n)}$. Then, abbreviating $k_n(x), t_n, m(n), p_n(x)$ as $k,t,m,p$ respectively, by a Chernoff bound we have that
	\begin{align*}
		\pr{|k/t - r| > (3/2) \cdot 2^{-m}}
		&\le \pr{|k/t - p| + |p - r| > 3/2 \cdot 2^{-m}} \\
		&\le \pr{|k/t - p| \cdot 2^{-m} > 3/2 \cdot 2^{-m}} \\
		&= \pr{|k/t - p| > 1/2 \cdot 2^{-m}} \\
		&\le 2\exp\Paren{-2t\Paren{1/2 \cdot 2^{-m}}^2} \\
		&= 2\exp\Paren{-5} < 1/2,
	\end{align*}
	so $(1^n, x, r)$ is in $L$. Therefore $f$ is well defined, and since $L$ is in $\PSPACE$ clearly $f$ is computable in $\PSPACE$ as well.
	
	Finally, consider $n,x,r$ with $|p_n(x) - r| \ge 2 \cdot 2^{-m(n)}$. Using the same abbreviations as above, we have that
	\begin{align*}
		\pr{|k/t - r| \le 3/2 \cdot 2^{-m}}
		&\le \pr{|p - r| - |k/t - p| \le 3/2 \cdot 2^{-m}} \\
		&\le \pr{2 \cdot 2^{-m} - |k/t - p| \le 3/2 \cdot 2^{-m}} \\
		&= \pr{|k/t - p| \ge 1/2 \cdot 2^{-m}} \\
		&< 1/2,
	\end{align*}
	where the last inequality holds by the same reasoning as above Therefore $(1^n, x, r)$ is not in $L$, so $r \neq f(1^n, x)$, implying that $|p_n(x) - f(1^n,x)| \le 2 \cdot 2^{-m(n)}$.	
\end{proof}

The following is an easy consequence of \cref{tomog-0}:

\begin{cor} \label{tomog}
	Let $\ell, m$ be polynomials, let $(C_n)_{n \in \N}$ be a space-uniform family of general quantum circuits such that each $C_n$ has $\ell(n) + n$ input qubits and one output qubit, and let $(\ket{\psi_n})_{n \in \N} \in \statePSPACE$. Then there exists a $\PSPACE$-computable function $f$ such that for all $n \in \N, x \in \cube{\ell(n)}$ it holds that $f(1^n, x) \in \D{m(n)}$ and
	\begin{equation*}
		|f(1^n, x) - \bra1 C_n(x, \psi_n) \ket1| \le 3 \cdot 2^{-m(n)}~.
	\end{equation*}
\end{cor}
\begin{proof}
	Since $(\ket{\psi_n})_{n \in \N}$ is in \statePSPACE, there exists a space-uniform family of general quantum circuits $(D_n)_{n \in \N}$ such that for all $n \in \N$, the circuit $D_n$ takes no inputs and $D_n$ outputs a density matrix $\rho_n$ such that $\td(\rho_n, \psi_n) \le 2^{-m(n)}$. Let $B_n$ be the general quantum circuit that on input $x \in \cube{\ell(n)}$, first applies $D_n$ to construct $\rho_n$, and then applies $C_n$ on input $\kb{x} \otimes \rho_n$. The family $(B_n)_{n \in \N}$ is clearly space-uniform, so by \cref{tomog-0} there exists a $\PSPACE$-computable function $f$ such that for all $n \in \N, x \in \cube{\ell(n)}$ it holds that $f(1^n,x) \in \D{m(n)}$ and
	\begin{equation*}
		|f(1^n,x) - \bra1 B_n(x) \ket1| \le 2 \cdot 2^{-m(n)}~.
	\end{equation*}
	For all $n \in \N, x \in \cube{\ell(n)}$, by the definition of trace distance
	\begin{equation*}
		|\bra1 \Paren{B_n(x) - C_n(x, \psi_n)} \ket1|
		= |\bra1 \Paren{C_n(x, \rho_n) - C_n(x, \psi_n)} \ket1|
		\le \td(\rho_n, \psi_n)
		\le 2^{-m(n)},
	\end{equation*}
	so by the triangle inequality
	\begin{align*}
		|f(1^n, x) - \bra1 C_n(x, \psi_n) \ket1|
		&\le |f(1^n,x) - \bra1 B_n(x) \ket1| + |\bra1 \Paren{B_n(x) - C_n(x, \psi_n)} \ket1| \\
		&\le 3 \cdot 2^{-m(n)}. \qedhere
	\end{align*}
\end{proof}

\subsection{Specific applications} \label{subsec:sa}

\newcommand{\cp}{\mrm{cp}}
\newcommand{\ph}{\mrm{ph}}
\newcommand{\round}{\mai{round}}

\begin{cor}\label{cor:tomog-relative-weight}
	Let $m$ be a polynomial, let $(\ket{\psi_n})_{n \in \N} \in \statePSPACE$, and for $n \in \N, x \in \cube{\le n}$ let $p_n(x) = \norm{(\bra{x} \otimes \id) \ket{\psi_n}}^2$. Then there exists a $\PSPACE$-computable function $\cp$ (for ``conditional probability") and a polynomial $\ell$ such that for all $n \in \N, x \in \cube{<n}$ it holds that $\cp(1^n, x) \in \D{\ell(n)}$ and
	\begin{equation*}
		|p_n(x) \cdot \cp(1^n,x) - p_n(x0)| \le 2^{-m(n)}~.
	\end{equation*}
\end{cor}

The intuition behind the definitions of $p_n$ and cp is that if $\ket{\psi_n}$ is measured in the standard basis, then $p_n(x)$ is the probability that the first $|x|$ qubits measure to $x$, and $\cp(1^n, x) \approx p_n(x0)/p_n(x)$ is approximately the probability that the $(|x|+1)$'st qubit measures to 0 conditioned on the first $|x|$ qubits measuring to $x$.

\begin{proof}
	Let $C_n$ be the general quantum circuit that on input $\kb{x} \otimes \rho$, where $x \in \cube{\le n}$ (encoded in some reasonable way as a string of length depending only on $n$) and $\rho$ is an $n$-qubit state, does the following: measure the first $|x|$ qubits of $\rho$ in the standard basis, output 1 if the outcome $x$ occurs, and otherwise output 0. Clearly $\bra1 C_n(x, \psi_n) \ket1 = p_n(x)$ and $(C_n)_{n \in \N}$ is space-uniform, so by \cref{tomog} there exists a $\PSPACE$-computable function $f$ such that for all $n \in \N, x \in \cube{\le n}$ it holds that $f(1^n, x) \in \D{\ell(n)}$ (we will specify the polynomial $\ell$ later) and
	\begin{equation*}
		|f(1^n, x) - p_n(x)| \le 3 \cdot 2^{-\ell(n)}~.
	\end{equation*}
	Let
	\begin{equation*}
		\cp(1^n,x) = \round\Paren{\frac{f(1^n,x0)}{f(1^n,x)}},
	\end{equation*}
	where the function $\round(\cdot)$ maps its argument to the nearest element of $\D{\ell(n)}$. (If $f(1^n,x) = 0$ then define $\cp(1^n,x)$ arbitrarily.) Clearly $\cp(\cdot)$ is computable in $\PSPACE$.
	
	Fix $n \in \N, x \in \cube{<n}$ and let
	\begin{align*}
		a &= p_n(x0), &b&= p_n(x),& \delta &=  3 \cdot 2^{-\ell(n)},\\
		\wt a &= f(1^n, x0), &\wt b&= f(1^n,x),& \mu &= \min\Paren{\wt a/\wt b, \, 1}~.
	\end{align*}
	It follows from definitions that $|\wt a - a|, |\wt b - b| \le \delta$, that $0 \le a \le b \le 1$, that $0 \le \wt a, \wt b \le 1$, and that $|\cp(1^n,x) - \mu| \le 2^{-\ell(n)}$. By the triangle inequality,
	\begin{equation*}
		|b \cdot \cp(1^n,x) - a|
		\le b \cdot |\cp(1^n,x) - \mu| + |b\mu - a|
		\le 2^{-\ell(n)} + |b\mu - a|~.
	\end{equation*}
	If $b \le 3\delta$ then
	\begin{equation*}
		|b\mu - a|
		\le \max(b\mu, a)
		\le b
		\le 3\delta,
	\end{equation*}
	and if $b > 3\delta$ then
	\begin{align*}
		|b\mu - a|
		&\le \left|b \cdot \frac{\wt a}{\wt b} - a\right|
		= \left|\frac{b \cdot \wt a - a \cdot \wt b}{\wt b}\right|
		= \left|\frac{b (\wt a - a) + a (b - \wt b)}{\wt b}\right| \\
		&\le \left|\frac{b (\wt a - a)}{\wt b}\right| + \left|\frac{a (b - \wt b)}{\wt b}\right|
		\le \frac{b\delta}{b - \delta} + \frac{a\delta}{b - \delta}
		\le \frac{2b\delta}{(2/3) b}
		= 3\delta,
	\end{align*}
	so either way it holds that $\left|b\mu - a\right| \le 3\delta$ and therefore
	\begin{equation*}
		|b \cdot \cp(1^n,x) - a|
		\le 2^{-\ell(n)} + 3\delta
		= 10 \cdot 2^{-\ell(n)}~.
	\end{equation*}
	The result follows by taking $\ell(n) = m(n) + 4$.
\end{proof}

For $\ell \in \N$ recall that we define $\U\ell = \{\exp(2\pi i r) : r \in \D\ell\}$.

\begin{cor}\label{cor:tomog-phase}
	Let $m$ be a polynomial, let $(\ket{\psi_n})_{n \in \N} \in \statePSPACE$, and for $n \in \N, x \in \cube{n}$ let $\alpha_x = \ip{x}{\psi_n}$. Then there exists a polynomial $\ell$, a sequence of unit-magnitude complex numbers $(\gamma_n)_{n \in \N}$, and a $\PSPACE$-computable function $\ph$ (for ``phase") such that for all $n \in \N, x \in \cube n$ it holds that $\ph(x) \in \U{\ell(n)}$ and
	\begin{equation*}
		\big| \ph(x) \cdot \gamma_n \cdot |\alpha_x| - \alpha_x \big| \le  2^{-m(n)}~.
	\end{equation*}
\end{cor}

At a high level we compute $\ph(x)$ as follows: first estimate $\alpha_x \conj \alpha_{y_n}$ for an appropriate string $y_n$ (where $y_n$ depends only on $n$) in a manner similar to one-qubit tomography, and then normalize to unit magnitude.

\begin{proof}
	Let $p(x) = |\alpha_x|^2$ for $x \in \cube*$. Below we define a sequence of strings $(y_n)_{n \in \N}$, with $y_n \in \cube n, p(y_n) \ge \frac12 \cdot 2^{-n}$ and the function $1^n \mapsto y_n$ computable in $\PSPACE$. Let $\gamma_n = \alpha_{y_n} / |\alpha_{y_n}|$.
	
	Let $A_n$ be the general quantum circuit that on input $\kb{x} \otimes \rho$, where $x \in \cube n$ and $\rho$ is an $n$-qubit state, measures $\rho$ in the standard basis and then outputs 1 if the outcome $x$ occurs and 0 otherwise. Clearly $\bra1 A_n(x, \psi) \ket1 = p(x)$ and $(A_n)_{n \in \N}$ is space-uniform, so by \cref{tomog} there exists a $\PSPACE$-computable function $f$ such that for all for all $n \in \N, x \in \cube n$ it holds that $f(x) \in \D{n+4}$ and $|f(x) - p(x)| \le \frac14 \cdot 2^{-n}$. Let $y_n$ be the lexicographically first string in $\cube n$ satisfying $f(y_n) \ge \frac34 \cdot 2^{-n}$. The string $y_n$ is well defined, because there exists $x \in \cube n$ such that $p(x) \ge 2^{-n}$, and for this $x$ it holds that
	\begin{equation*}
		f(x) = p(x) + (f(x) - p(x)) \ge 2^{-n} - 1/4 \cdot 2^{-n} = 3/4 \cdot 2^{-n}~.
	\end{equation*}
	By a similar calculation,
	\begin{equation*}
		p(y_n) = f(y_n) + (p(y_n) - f(y_n)) \ge 3/4 \cdot 2^{-n} - 1/4 \cdot 2^{-n} = 1/2 \cdot 2^{-n},
	\end{equation*}
	and the function $1^n \mapsto y_n$ is computable in $\PSPACE$ because $f$ is computable in $\PSPACE$.
	
	\newcommand*{\uu}{U}
	Let $\uu$ denote the set $\{1,i,-1,-i\}$. For $x,z \in \cube n$,
	\begin{align}
		2 \alpha_x \conj \alpha_z
		&= 2\ip{x}{\psi_n} \ip{\psi_n}z  \nonumber \\
		&= \frac12\sum_{u \in \uu} \Paren{
			u \ip{x}{\psi_n} \ip{\psi_n}{x} +
			\ip{x}{\psi_n} \ip{\psi_n}{z} +
			u^2 \ip{z}{\psi_n} \ip{\psi_n}{x} +
			u \ip{z}{\psi_n} \ip{\psi_n}{z}} \nonumber \\
		&= \sum_{u \in \uu} u \frac{\bra{x} + u \bra{z}}{\sqrt 2} \kb{\psi_n} \frac{\ket{x} + \conj{u} \ket{z}}{\sqrt 2} \nonumber \\
		&= \sum_{u \in \uu} u \left|\frac{\bra{x} + u \bra{z}}{\sqrt 2} \ket{\psi_n}\right|^2. \label{U4sum}
	\end{align}
	
	Let $W_n$ be the set of tuples $(x,z,u)$ with $x,z \in \cube n, u \in \uu$ such that $x \neq z$. We now prove that there is a $\PSPACE$-computable function $g$ with $g(x, z, u) \in \D{\ell(n)}$ for all $(x,z,u) \in W_n$ (we will specify the polynomial $\ell$ later) such that
	\begin{equation} \label{gapprox}
		\forall (x,z,u) \in W_n: \ \left|g(x,z,u) - \left|\frac{\bra{x} + u \bra{z}}{\sqrt 2} \ket{\psi_n}\right|^2 \right| \le 3 \cdot 2^{-\ell(n)}~.
	\end{equation}
	By \cref{tomog}, it suffices to give a space-uniform family of general quantum circuits $(B_n)_{n \in \N}$, where $B_n$ has $3n+2$ input qubits and one output qubit, such that
	\begin{equation*}
		\forall (x,z,u) \in W_n: \ \bra1 B_n(x,z,u, \psi_n) \ket1 = \left|\frac{\bra{x} + u \bra{z}}{\sqrt 2} \ket{\psi_n}\right|^2~.
	\end{equation*}
	Our construction of $B_n$ involves the following unitary quantum circuit $C_n$, which satisfies
	\begin{equation*}
		\forall (x,z,u) \in W_n: \ C_n \ket{x,z,u, 0^{\poly(n)}, 0^n} = \ket{x,z,u, 0^{\poly(n)}} \otimes \frac{\ket x + \conj u \ket z}{\sqrt 2},
	\end{equation*}
	and which acts on input $\ket{x,z,u, 0^{\poly(n)}, 0^n}$ as follows: Construct $\ket+$ in a one-qubit ancilla register $\reg R$; controlled on 0 in $\reg R$, XOR $x$ into the last $n$ qubits; controlled on 1 in $\reg R$, XOR $z$ into the last $n$ qubits and apply a phase of $\conj u$; and controlled on the last $n$ qubits equaling $z$, XOR 1 into $\reg R$. Finally, let $B_n$ act as follows on input $(x,z,u,\phi)$ where $\ket\phi$ is an $n$-qubit state: Construct $\adj C_n \ket{x, z, u, \zs, \phi}$, measure the last $n$ qubits in the standard basis, output 1 if the outcome is all-zeros, and otherwise output 0.
	
	For $n \in \N$, let $\ph(y_n) = 1$, and for $x \in \cube{n} \backslash \{y_n\}$ let
	\begin{equation*}
		\ph(x) = \round \Paren{\frac{\sum_{u \in \uu} u \cdot g(x,y_n,u)} {|\sum_{u \in \uu} u \cdot g(x,y_n,u)|}},
	\end{equation*}
	where the function $\round(\cdot)$ maps its argument to the nearest element of $\U{\ell(n)}$ (and $\round(0/0)$ may be defined arbitrarily). The function $\ph(\cdot)$ is computable in $\PSPACE$ because the functions $1^n \mapsto y_n$ and $g$ are computable in $\PSPACE$.
	
	Fix $n \in \N$, and note that
	\begin{equation*}
		\big| \ph(y_n) \, \gamma_n \, |\alpha_{y_n}| - \alpha_{y_n} \big|
		= |1 \cdot \alpha_{y_n} - \alpha_{y_n}|
		= 0
		<  2^{-m(n)}~.
	\end{equation*}
	Now consider a string $x \in \cube{n} \backslash \{y_n\}$, and let
	\begin{align*}
		&\lambda = 2 \alpha_x \conj \alpha_{y_n} = 2 \alpha_x |\alpha_{y_n}| \conj \gamma_n,
		&\mu = \sum_{u \in \uu} u \cdot g(x,y_n,u).
	\end{align*}
	By \eqref{U4sum}, \eqref{gapprox} and the triangle inequality,
	\begin{equation*}
		|\lambda - \mu|
		= \left|\sum_{u \in \uu} u \cdot \Paren{\left|\frac{\bra{x} + u \bra{y_n}}{\sqrt 2} \ket{\psi_n}\right|^2 - g(x,y_n,u)}\right|
		\le 12 \cdot 2^{-\ell(n)}~.
	\end{equation*}
	Therefore, by the triangle inequality and \eqref{eq:integral} and recalling that $|\alpha_{y_n}| \ge \frac1{\sqrt 2} \cdot 2^{-n/2} \ge \frac12 \cdot 2^{-n/2}$,
	\begin{align*}
		\big| \ph(x) \, \gamma_n \, |\alpha_x| - \alpha_x \big|
		&= \big|\ph(x) \, |\alpha_x|  - \alpha_x \, \conj \gamma_n \big| \\
		&\le \left| \ph(x) \, |\alpha_x| - \frac{\mu}{|\mu|} |\alpha_x| \right| + \left| \frac{\mu}{|\mu|} |\alpha_x| - \alpha_x \, \conj \gamma_n \right| \\
		&\le \left|\ph(x) - \frac\mu{|\mu|}\right| + \frac1{2 |\alpha_{y_n}|} \cdot \left|\frac\mu{|\mu|} |\lambda| - \lambda \right| \\
		&\le 2\pi \cdot 2^{-\ell(n)} + 2^{n/2} \cdot \left|\frac\mu{|\mu|} (|\lambda| - |\mu|) + \mu - \lambda\right| \\
		&\le 7 \cdot 2^{-\ell(n)} + 2^{n/2} \cdot \Paren{\big||\lambda|-|\mu|\big| + |\lambda-\mu|} \\
		&\le 7 \cdot 2^{-\ell(n)} + 2^{n/2} \cdot 2|\lambda-\mu| \\
		&\le 31 \cdot 2^{n/2 - \ell(n)},
	\end{align*}
	so it suffices to take $\ell(n) = m(n) + \ceil{n/2} + 5$.
\end{proof}
	\section{Interactive State Synthesis}
\label{sec:synthesis}
Let $\Psi = (\ket{\psi_n})_{n \in \N}$ be a family of quantum states in \statePSPACE, and let $q$ be a polynomial. In this section we prove that $\Psi \in \stateQIP{c,s}$ with completeness $c(n) = \exp(-q(n))$ and soundness 
\[
s(n,\delta) = \exp \Paren{-\frac{ \delta^4 - \exp(-q(n))}{\Delta(n)}}
\]
for some polynomial $\Delta$. For every fixed $n$, this soundness function is log-concave and nonincreasing as a function of $\delta$. So by \cref{lem:stateQIP-amplification} applied with the polynomial $m(n) = \Delta(n) \cdot q(n)$, it follows that $\Psi \in \stateQIP{c,s'}$ with the same completeness $c$ and with soundness
\[
s'(n,\delta)
= \exp \Paren{-q(n) \cdot (\delta^4 - \exp(-q(n)))}
= \exp \Paren{q(n) e^{-q(n)} - q(n) \cdot \delta^4}~.
\]
This establishes \Cref{thm:prim-result}, because for every polynomial $p$ there exists a polynomial $q \ge p$ such that $q(n) e^{-q(n)} \le e^{-p(n)}$ for all sufficiently large $n$.

\subsection{Description of the protocol} \label{subsec:desc-prot}
Before describing the $\cc{stateQIP}$ protocol for the state family $\Psi$, we first describe some of its subroutines:
\paragraph{Interactive protocols for quantum state tomography.}
For $n \in \N, x \in \cube{\le n}$ let
\begin{equation*}
	p_{n}(x) = \norm{(\bra{x} \otimes \id) \ket{\psi_n}}^2,
\end{equation*}
i.e.\ $p_{n}(x)$ is the probability that measuring the first $|x|$ qubits of $\ket{\psi_n}$ yields the string $x$. And for $n \in \N, x \in \cube{n}$ let
\begin{equation*}
	\alpha_x = \ip{x}{\psi_n}~.
\end{equation*}
Let $m$ be a sufficiently large polynomial, to be specified later.

By \cref{cor:tomog-relative-weight} there exists a $\PSPACE$-computable function $\cp$ and a polynomial $\ell_\cp$ such that for all $n \in \N, x \in \cube{<n}$ it holds that $\cp(1^n,x) \in \D{\ell_\cp(n)}$ and
\begin{equation} \label{eq:cp1}
	\left|p_n(x) \cdot \cp(1^n,x) - p_n(x0) \right| \le 2^{-m(n)},
\end{equation}
or equivalently
\begin{equation} \label{eq:cp2}
	\left|p_n(x) \cdot (1-\cp(1^n,x)) - p_n(x1) \right| \le 2^{-m(n)}~.
\end{equation}
Similarly, by \cref{cor:tomog-phase} there exists a polynomial $\ell_\ph$, a sequence of unit-magnitude complex numbers $(\gamma_n)_{n \in \N}$, and a $\PSPACE$-computable function $\ph$ such that for all $n \in \N, x \in \cube n$ it holds that $\ph(x) \in \U{\ell_\ph (n)}$ and
\begin{equation} \label{eq:ph}
	\left| \ph(x) \cdot \gamma_n \cdot \sqrt{p_n(x)} - \alpha_x \right| \le  2^{-m(n)}
\end{equation}
where we used that $\sqrt{p_n(x)} = |\alpha_x|$.  Assume without loss of generality that $\ell_\ph = \ell_\cp = \ell$ for some polynomial $\ell$. Define languages
\begin{align*}
	&L_\cp = \{(1^n, x, \cp(1^n, x)) : n \in \N, x \in \cube{<n}\},
	&L_\ph = \{(x,\ph(x)) : x \in \cube*\}~.
\end{align*}
Since the functions $\cp$ and $\ph$ are computable in $\PSPACE$, the languages $L_\cp$ and $L_\ph$ are in $\PSPACE$, so by \cref{thm:qip=pspace} there exist $\QIP[1/2]$ verifiers $V_\cp$ and $V_\ph$ for $L_\cp$ and $L_\ph$ respectively. Without loss of generality, these verifiers may be assumed to be unitary except for the measurement of the accept/reject qubit at the end.

\paragraph{The swap test.}

The \emph{swap test}~\cite{buhrman2001quantum} is a procedure that takes in as input two registers $\reg{A}$ and $\reg{B}$ that have the same number of qubits, and performs the two-outcome projective measurement $\{ S_{\reg{AB}}, A_{\reg{AB}} \}$ where $S_{\reg{AB}}$ is the projector onto the \emph{symmetric subspace} and $A_{\reg{AB}} = \id - S_{\reg{AB}}$ is the projector onto the \emph{antisymmetric subspace} of registers $\reg{A}, \reg{B}$. The projectors $S_{\reg{AB}}, A_{\reg{AB}}$ can alternatively be expressed as 
\begin{align*}
	&S_{\reg{AB}} = \frac{\id + \mai{Swap}_{\reg{AB}}}2,
	&A_{\reg{AB}} = \frac{\id - \mai{Swap}_{\reg{AB}}}2,
\end{align*}
where $\mai{Swap}_{\reg{AB}}$ is the Hermitian unitary that swaps the contents of registers $\reg{A}$ and $\reg{B}$.

\begin{clm} \label{clm:swap-test}
	There is a uniform family of polynomial-size general quantum circuits, where the $n$'th circuit takes $2n$ input qubits partitioned into $n$-qubit registers $\reg A$ and $\reg B$, and performs the two-outcome projective measurement $\{ S_{\reg{AB}}, A_{\reg{AB}} \}$ on the input.
\end{clm}
\begin{proof}
	The circuit first applies a controlled-$\mai{Swap}_{\reg{AB}}$ gate, where the control qubit is an ancilla initialized in the $\ket+$ state. (This can be done by a circuit of linear size, because $\mai{Swap}_{\reg{AB}}$ is the tensor product of $n$ two-qubit swap transformations.) On input $\ket{\phi}_{\reg{AB}}$, the resulting state is
	\begin{equation*}
		\frac1{\sqrt2} \ket0 \otimes \ket{\phi} + \frac1{\sqrt2} \ket1 \otimes \mai{Swap} \ket{\phi}
		= \ket+ \otimes S\ket{\phi} + \ket- \otimes A\ket{\phi},
	\end{equation*}
	so it suffices to measure the ancilla qubit in the Hadamard basis.
\end{proof}

If $\reg A$ and $\reg B$ are zero-qubit registers, then $\mai{Swap}_{\reg{AB}} = \id_{\reg{AB}}$ and $S_{\reg{AB}} = \id_{\reg{AB}}, A_{\reg{AB}} = 0_{\reg{AB}}$, so in this case the measurement from \cref{clm:swap-test} vacuously outputs ``$S$" with probability 1.

\paragraph{The verifier.}

Let $t$ be a sufficiently large polynomial, to be specified later. In the rest of \cref{sec:synthesis} we fix $n$ and write $\ket\psi = \ket{\psi_n}$. \cref{alg:main} describes a $\cc{stateQIP}$ verifier for synthesizing $\ket\psi$. The verifier's workspace includes the following disjoint registers, which are initialized to the all-zeros state:
\begin{itemize}[leftmargin=*]
	\renewcommand\labelitemi{--}
	\item registers $\reg A_1, \dotsc, \reg A_n$ each consisting of a single qubit;
	\item registers $\reg B_1, \dotsc, \reg B_n$ each consisting of a single qubit;
	\item registers $\reg D_1, \dotsc, \reg D_{t(n)}$ each consisting of $\ell(n)$ qubits;
	\item registers $\reg W_1, \dotsc, \reg W_{t(n)}$ each consisting of $\poly(n)$ qubits.
\end{itemize}
For $0 \le k \le n$ we write $\reg A_{[k]}$ to denote the concatenated register $\reg A_1 \reg A_2 \dotsb \reg A_k$, and similarly for $\reg B_{[k]}$. For future convenience we also define zero-qubit registers $\reg A_{n+1}$ and $\reg B_{n+1}$, and write $\reg A_{[n+1]} = \reg A_{[n]}$ and $\reg B_{[n+1]} = \reg B_{[n]}$.

\begin{algorithm}[th!]
	\caption{Interactive synthesis of $\ket{\psi}$ with an untrusted prover.} \label{alg:main}
	\begin{algorithmic}[1]
		\State Set $k=0$, and sample $b = \Paren{b_1, \dotsc, b_{t(n)}} \unif \cube{t(n)}$.
		\For{$h = 1$ to $t(n)$}
		\IIf{$k = n + 1$} go to \Cref{line:m-end}. \EndIIf
		\Ctrl{the state $\ket x$ of $\reg A_{[k]}$ where $x \in \cube k$,} \hlabel{line:m1a} \Comment{\emph{start of ``forward-QIP"}}
		\State Send a copy of $x$ to the prover, and then let the prover act on $\reg D_h$. \hlabel{line:msendx-sec6}
		\Ctrl{the state $\ket\eta$ of $\reg D_h$ where $\eta \in \cube{\ell(n)}$,}
		\State Copy $(x,\eta)$ into a sub-register of $\reg W_h$, and then perform the following in $\reg W_h$:\hlabel{line:mcopyx-sec6}
		\If{$k<n$} unitarily simulate the verifier $V_\cp$ on input $(1^n, x, \eta)$, \hlabel{line:mqip1}
		\ElsIf{$k=n$} unitarily simulate the verifier $V_\ph$ on input $(x, \eta)$, \hlabel{line:mqip2}
		\EndIf
		\EndCtrl
		\EndCtrl \hlabel{line:m1b} \Comment{\emph{end of ``forward-QIP"}}
		\If{$b_h = 1$} \hlabel{line:m2a} \Comment{\emph{start of ``A-grow"}}
		\State Measure the flag qubit in $\reg W_h$ in the standard basis. Reject if the outcome is $0$. \hlabel{line:mmeas}
		\Ctrl{the state $\ket\eta$ of $\reg D_h$ where $\eta \in \cube{\ell(n)}$,}
		\If{$k < n$} interpreting $\eta$ as an element of $\D{\ell(n)}$, construct the state 
		\Statex \hspace{\algorithmicindent} \qquad \qquad $\sqrt\eta \ket0 + \sqrt{1-\eta} \ket1$ in $\reg A_{k+1}$.
		\ElsIf{$k=n$} interpreting $\eta$ as an element of $\U{\ell(n)}$, apply the phase $\eta$.
		\EndIf
		\EndCtrl
		\EndIf \hlabel{line:m2b} \Comment{\emph{end of ``A-grow"}}
		\State Run the forward-QIP step ``in reverse". \hlabel{line:m3} \Comment{\emph{``backward-QIP"}}
		\State Send the value of $b_h$ to the prover.
		\IIf{$b_h = 1$} let the prover act on $\reg B_{[k + b_h]}$. \EndIIf \hlabel{line:m4} \Comment{\emph{``B-grow"}}
		\State Perform the swap test on registers $\reg{A}_{[k+b_h]} \reg{B}_{k+b_h}$, and reject if the outcome is the 
		\Statex \hspace{\algorithmicindent} \qquad antisymmetric subspace.
		\State $k \gets k + b_h$.
		\EndFor
		\State \Return $\reg A_{[n]}$. \hlabel{line:m-end}
	\end{algorithmic}
\end{algorithm}

We define the following groups of lines in \cref{alg:main} for future convenience: the ``forward-QIP step" refers to lines~\ref{line:m1a} to \ref{line:m1b}, the ``A-grow step" refers to lines~\ref{line:m2a} to \ref{line:m2b}, the ``backward-QIP step" refers to line~\ref{line:m3}, and the ``B-grow step" refers to line~\ref{line:m4}.

Rather than referring to an explicit message register, we model the interaction between the verifier and prover as follows: If the verifier owns a $\kappa$-qubit register $\reg R$, then we write ``the verifier lets the prover act on $\reg R$" to denote that first the verifier swaps $\reg R$ with the first $\kappa$ qubits of the message register, then the prover acts on the message register, and finally the verifier swaps the first $\kappa$ qubits of the message register with $\reg R$.

Finally, we clarify the following lines of \cref{alg:main}:
\begin{itemize}[leftmargin=*]
	\renewcommand\labelitemi{--}
	\item In \cref{line:msendx-sec6} ``send a copy of $x$ to the prover" means the following: Write a copy of $x$ (encoded in some reasonable way as a string of length depending only on $n$) to a sub-register of $\reg W_h$ that will not be used for anything else, and then let the prover act on this sub-register.
	\item In \cref{line:mqip1,line:mqip2} ``unitarily simulate" means to simulate using only unitary transformations, and in particular to not measure the flag (accept/reject) qubit at the end.
	\item In \cref{line:mmeas} the ``flag qubit in $\reg W_h$" refers to the flag qubit from the simulation of a $\QIP[1/2]$ verifier in \cref{line:mqip1,line:mqip2}.
	\item In the backward-QIP step we mean the following: If the precise specification of the forward-QIP step is to successively apply unitaries $V_1, V_2, \dotsc, V_\kappa$ interspersed with actions by the prover, then now apply $\adj V_\kappa, \adj V_{\kappa-1}, \dotsc, \adj V_1$ interspersed with actions by the prover.
\end{itemize}

\subsection{Proof of completeness} \label{subsec:proof-comp}

For $x \in \cube{<n}$ define approximate conditional probabilities
\begin{align*}
	&\wt g_{x0} = \cp(1^n, x),
	&\wt g_{x1} = 1-\cp(1^n,x),
\end{align*}
and for $x \in \cube{\le n}$ define the approximate marginal probability
\begin{equation*}
	\wt{p}(x) = \prod_{j=1}^{|x|} \wt g_{x_{\leq j}}
\end{equation*}
(recall that $x_{\leq j}$ denotes the first $j$ bits of $x$). For $0 \le k \le n$ define the approximate intermediate state
\begin{equation*}
	\ket{\wt{\psi}\uppart{k}} = \sum_{\mathclap{x \in \cube k}} \sqrt{\wt{p}(x)} \, \ket{x},
\end{equation*}
and define the approximate target state
\begin{equation*}
	\ket{\wt{\psi}} = \ket{\wt{\psi}\uppart{n+1}} = \sum_{\mathclap{x \in \cube{n}}} \ph(x) \sqrt{\wt{p}(x)} \, \ket{x}~.
\end{equation*}
(The amplitudes of $\wt{\psi}$ are exponentially close to those of the ideal target state $\psi$, because of the exponentially small error induced by the $\PSPACE$ algorithms for tomography.)

The following honest prover helps the verifier synthesize $\ket{\wt{\psi}}$ exactly:
\begin{itemize}[leftmargin=*]
	\renewcommand\labelitemi{--}
	\item In the forward-QIP step of iteration $h$, controlled on receiving $x$ from the verifier,
	\begin{itemize}
		\item[*] If $|x|<n$, then write $\cp(1^n,x)$ to $\reg D_h$, and then simulate an honest prover corresponding to $V_\cp$ on input $(1^n, x, \cp(1^n, x))$.
		\item[*] If $|x|=n$, then write $\ph(x)$ to $\reg D_h$, and then simulate an honest prover corresponding to $V_\ph$ on input $(x, \ph(x))$.
	\end{itemize}
	
	\item In the backward-QIP step of iteration $h$, run the above bullet point ``in reverse", i.e.\ if the precise specification of the above bullet point is to successively apply unitaries $P_1, P_2, \dotsc, P_\kappa$ interspersed with actions by the verifier, then now apply $\adj P_\kappa, \adj P_{\kappa-1}, \dotsc, \adj P_1$ interspersed with actions by the verifier.
	
	\item In the B-grow step of iteration $h$ when $b_h = 1$, swap a copy of $\ket{\wt\psi \uppart{k+1}}$ into $\reg B_{[k+1]}$, where $k = \hw(b_{<h})$ (recall that $\hw(\cdot)$ denotes Hamming weight).
\end{itemize}

With this prover, it is easy to see by induction on $h$ that conditioned on the random string $b \in \cube{t(n)}$, for all $h \in [t(n)]$ with $k = \hw(b_{<h}) \le n$, the state of the register $\reg A_{[k+b_h]} \reg B_{[k+b_h]}$ at the end of iteration $h$ is $\ket{\wt\psi \uppart{k+b_h}}^{\otimes 2}$. Furthermore, the verifier never rejects.\footnote{If the verifier uses the Solovay-Kitaev theorem~\cite{dawson2006solovay} to implement the unitary from the A-grow step \emph{approximately} over a finite gate set rather than exactly, then at the B-grow step when $b_h = 1$, the honest prover should provide in $\reg B_{[k+b_h]}$ a copy of the \emph{actual} state in $\reg A_{[k+b_h]}$ rather than $\ket{\wt\psi \uppart{k+b_h}}$. Then the verifier still accepts with probability exactly 1.}

Thus, averaging over the random string $b \in \cube{t(n)}$, the output state $\rho$ of the verifier after $t(n)$ rounds is
\begin{equation*}
	\rho = \PR{\hw(b) \le n} \cdot \sigma + \PR{\hw(b) \ge n+1} \cdot \wt\psi
\end{equation*}
for some mixed state $\sigma$. By the convexity of trace distance, it follows that
\begin{align*}
	\td(\rho, \psi)
	&\le \PR{\hw(b) \le n} \cdot \td(\sigma, \psi) + \PR{\hw(b) \ge n+1} \cdot \td\Paren{\wt\psi, \psi} \\
	&\le \PR{\hw(b) \le n} + \td\Paren{\wt\psi, \psi}~.
\end{align*}
Along with \cref{lem:term-whp,lem:psi-approx} (stated below), this implies that $\td(\rho, \psi) \le \exp(-q(n))$ as desired, where we take $t(n) = 18q(n) + 3n + 54$ and $m(n) = 4q(n) + 12n$:

\begin{lem} \label{lem:term-whp}
	$\PR{\hw(b) \le n} \le \frac1{12} \cdot \exp(-q(n))$.
\end{lem}
\begin{proof}
	By a Chernoff bound,
	\begin{align*}
		\PR{\hw(b) \le n}
		&= \PR{\mrm{Binomial}(t(n), 1/2) \le n}
		\le \exp\Paren{-2 \cdot t(n) \cdot \Paren{\frac12 - \frac{n}{t(n)}}^2} \\
		&\le \exp\Paren{-2 \cdot (18q(n) + 54) \cdot \Paren{\frac12 - \frac{n}{3n}}^2}
		= \exp(-q(n) - 3) \\
		&< \frac{1}{12} \cdot \exp(-q(n))~. \qedhere
	\end{align*}
\end{proof}

\begin{lem} \label{lem:psi-approx}
	$\td\Paren{\wt\psi, \psi} \le \frac1{12} \cdot \exp(-q(n))$.
\end{lem}
\begin{proof}
	Recall the definitions of $\alpha_x, p_n(x), \gamma_n$ from \cref{subsec:desc-prot}, and write $p(\cdot) = p_n(\cdot)$ and $\gamma = \gamma_n$. Since $\gamma \wt\psi \gamma^* = \wt\psi$ (because $\gamma$ has unit magnitude) and then using \eqref{eq:td-fid2}, we have
	\begin{equation*}
		\td\Paren{\wt\psi, \psi}
		= \td\Paren{\gamma \wt\psi \gamma^*, \psi}
		\leq \Norm{\gamma \ket{\wt\psi} - \ket\psi}
		\le \Norm{\gamma \ket{\wt\psi} - \ket\psi}_1
		= \sum_{\mathclap{x \in \cube n}} \left|\gamma  \cdot \ph(x) \cdot \sqrt{\wt{p}(x)}- \alpha_x\right|~.
	\end{equation*}
	Fix $x \in \cube n$. By the triangle inequality and \eqref{eq:ph}, we have
	\begin{align*}
		\left|\gamma\cdot \ph(x) \cdot \sqrt{\wt{p}(x)} - \alpha_x\right|
		&= \left|\gamma  \cdot \ph(x)\cdot \Paren{\sqrt{\wt{p}(x)} - \sqrt{p(x)}} + \Paren{\gamma \cdot \ph(x) \cdot \sqrt{p(x)} - \alpha_x} \right| \\
		&\le \left|\sqrt{\wt{p}(x)} - \sqrt{p(x)} \right| + 2^{-m(n)}~.
	\end{align*}
	Furthermore,
	\begin{align*}
		\left|\sqrt{p(x)} - \sqrt{\wt{p}(x)}\right|
		&= \left| \sqrt{p(x)} - \prod_{j=1}^n \sqrt{\wt{g}_{x_{\leq j}}} \right| \\
		&= \left| \sum_{j = 1}^n \Paren{\sqrt{p(x_{\le j})}  - \sqrt{p(x_{\le j - 1}) \cdot \wt{g}_{x_{\leq j}}}} \cdot \prod_{i=j+1}^n \sqrt{\wt{g}_{x_{\leq i}}} \right| \\
		&\le \sum_{j = 1}^n  \left| \sqrt{p(x_{\le j})}  - \sqrt{p(x_{\le j - 1}) \cdot \wt{g}_{x_{\leq j}} }  \right| \\
		&\le \sum_{j=1}^n \sqrt{\left| p(x_{\le j}) - p(x_{\le j - 1}) \cdot \wt{g}_{x_{\leq j}} \right|} \\
		&\le n 2^{-m(n)/2},
	\end{align*}
	where the first line is by the definition of $\wt{p}(x)$, the second line is by a telescoping sum, the third line is by the triangle inequality and fact that $\wt g_y \le 1$ for all $y$, the fourth line is by the fact that $|a - b|^2 \leq |a^2 - b^2|$ for nonnegative $a,b$, and the last line is by \eqref{eq:cp1} and \eqref{eq:cp2}. Therefore
	\begin{equation*}
		\td\Paren{\wt\psi, \psi} \le 2^n (n+1) 2^{-m(n)/2},
	\end{equation*}
	from which the result follows by the definition of $m$.
\end{proof}

\subsection{Proof of soundness} \label{subsec:proof-sound}

Throughout this subsection we fix an arbitrary prover, and let $b = \Paren{b_1, \dotsc, b_{t(n)}} \unif \cube{t(n)}$ denote the random string used by the verifier as described in \cref{subsec:desc-prot}. Let $\rho$ denote the output state conditioned on accepting, and for $d \in \cube{t(n)}$ let $\rho_d$ denote the output state conditioned on $b=d$ and on accepting.

For $h \in [t(n)], a \in \cube{h}$ define
\begin{equation*}
	r_a = \PR{\text{verifier rejects in iteration $h$} \mid b_{\leq h} = a}~.
\end{equation*}
The crux of the proof of soundness is the following lemma, which given a fixed string $d \in \cube{t(n)}$ of random choices of the verifier, relates the distance between the output state at the end of the protocol (conditioned on accepting) and the ideal state to the rejection probabilities at each iteration of the protocol. Recall the definition of $\ket{\wt\psi}$, which is the state synthesized by the verifier and the honest prover from \cref{subsec:proof-comp}.
\begin{lem} \label{lem:main-sound}
	Let $d \in \cube{t(n)}$ be a string such that $\hw(d) \ge n+1$. Then
	\begin{equation*}
		\td\Paren{\rho_d, \wt\psi} \le 4 \cdot \left( t(n) \cdot \sum_{h=1}^{t(n)} \left(r_{d_{< h} 0} + r_{d_{< h} 1} \right) \right)^{1/4}~.
	\end{equation*}
\end{lem}
\noindent We defer the proof of \cref{lem:main-sound} to \cref{sssec:narrow}, and we first prove the soundness of the protocol assuming it. Recall that our goal is to prove that
\begin{equation} \label{eq:sound-remind}
	\pr{\text{verifier accepts}} \le \exp\Paren{- \frac{\td(\rho, \psi)^4 - \exp(-q(n))} {\poly(n)}}~.
\end{equation}

\begin{proof}[Proof of soundness assuming \cref{lem:main-sound}]
	We proceed by upper-bounding $\td(\rho, \psi)$. By the triangle inequality and \cref{lem:psi-approx},
	\begin{equation*}
		\td\Paren{\rho, \psi}
		\le \td\Paren{\rho, \wt\psi} + \td\Paren{\wt\psi, \psi}
		\le \td\Paren{\rho, \wt\psi} + \frac{1}{12} \cdot \exp(-q(n))~.
	\end{equation*}
	For $d \in \cube{t(n)}$ let
	\begin{equation*}
		\mu(d) = \pr{b=d \mid \text{verifier accepts}}~.
	\end{equation*}
	By the convexity of trace distance,
	\begin{equation*}
		\td\Paren{\rho, \wt\psi}
		= \td\Paren{\sum_d \mu(d) \rho_d, \, \wt\psi}
		\le \sum_d \mu(d) \td\Paren{\rho_d, \wt\psi},
	\end{equation*}
	so
	\begin{equation*}
		\td\Paren{\rho, \psi}
		\le \sum_d \mu(d) \td\Paren{\rho_d, \wt\psi} + \frac{1}{12} \cdot \exp(-q(n))~.
	\end{equation*}
	
	We now establish a similar statement for $d$ over the \emph{uniform} distribution rather than $\mu$, under the assumption that the verifier accepts with probability at least $1/2$. (If this assumption is not satisfied, then for all sufficiently large $n$ we have that
	\begin{equation*}
		\pr{\text{verifier accepts}}
		\le 1/2
		\le \exp\Paren{- \frac{1 - \exp(-q(n))} {\poly(n)}}
		\le \exp\Paren{- \frac{\td(\rho, \psi)^4 - \exp(-q(n))} {\poly(n)}}
	\end{equation*}
	and thus \eqref{eq:sound-remind} holds vacuously.) By Bayes' rule, for $d \in \cube{t(n)}$ it holds that
	\begin{equation*}
		\mu(d)
		= \frac{\pr{\text{verifier accepts} \mid b=d} \cdot \pr{b=d}} {\pr{\text{verifier accepts}}}
		\le 2 \, \pr{b=d},
	\end{equation*}
	and therefore
	\begin{equation*}
		\td\Paren{\rho, \psi}
		\le 2 \, \E \td\Paren{\rho_b, \wt\psi} + \frac{1}{12} \cdot \exp(-q(n))
	\end{equation*}
	(where $b \sim \cube{t(n)}$ is uniform random).
	
	By \cref{lem:main-sound,lem:term-whp} and Jensen's inequality,
	\begin{align*}
		\E \td\Paren{\rho_b, \wt\psi}
		&= \E \left[ \td\Paren{\rho_b, \wt\psi} \Ind{\hw(b) \le n} \right] + \E \left[ \td\Paren{\rho_b, \wt\psi} \Ind{\hw(b) > n} \right] \\
		&\le \PR{\hw(b) \le n} + \E \left[4 \cdot \left( t(n) \cdot \sum_{h=1}^{t(n)} \left(r_{b_{< h} 0} + r_{b_{< h} 1} \right) \right)^{1/4} \right] \\
		&\le \frac1{12} \exp(-q(n)) + 4 \cdot t(n)^{1/4} \cdot \E \left[\sum_{h=1}^{t(n)} \left(r_{b_{< h} 0} + r_{b_{< h} 1} \right) \right]^{1/4}~.
	\end{align*}
	Furthermore, recalling that $b$ is uniform random, we have that
	\begin{align*}
		\E \sum_{h=1}^{t(n)} \left(r_{b_{< h} 0} + r_{b_{< h} 1} \right)
		&= 2 \sum_{h=1}^{t(n)} \E \left[ r_{b_{\le h}} \right]
		= 2 \sum_{h=1}^{t(n)} \pr{\text{verifier rejects in iteration $h$}} \\
		&= 2 \, \pr{\text{verifier rejects}}
	\end{align*}
	where the first equality holds by linearity of expectation, the second equality holds by the definition of $r_{b_{\le h}}$, and the last equality holds because the verifier can reject in at most one iteration in any execution of the protocol. Recalling that $t$ is a polynomial, it follows that
	\begin{equation*}
		\td\Paren{\rho, \psi}
		\le \frac14 \cdot \exp(-q(n)) + \poly(n) \cdot \pr{\text{verifier rejects}}^{1/4}~.
	\end{equation*}
	
	Rearranging yields
	\begin{equation*}
		\pr{\text{verifier rejects}}^{1/4} \geq \frac{\td\Paren{\rho, \psi} - \frac14 \cdot \exp(-q(n))} {\poly(n)},
	\end{equation*}
	and by \cref{lem:4p} it follows that
	\begin{equation*}
		\pr{\text{verifier rejects}} \geq \frac{\td\Paren{\rho, \psi}^4 - \exp(-q(n))} {\poly(n)}~.
	\end{equation*}
	Finally, \eqref{eq:sound-remind} follows because
	\begin{align*}
		\pr{\text{verifier accepts}}
		&= 1 - \pr{\text{verifier rejects}}
		\le \exp(-\pr{\text{verifier rejects}}) \\
		&\le \exp\Paren{-\frac{\td\Paren{\rho, \psi}^4 - \exp(-q(n))} {\poly(n)}}~. \qedhere
	\end{align*}
\end{proof}
\subsubsection{Proof of \texorpdfstring{\cref{lem:main-sound}} {Lemma 6.4}} \label{sssec:narrow}
Fix a string $d \in \cube{t(n)}$ with $\hw(d) \ge n+1$, and for $a \in \cube{< t(n)}$ let
\begin{equation*}
	u_a = r_{a0} + r_{a1}~.
\end{equation*}
Recall that our goal is to prove that
\begin{equation*}
	\td\Paren{\rho_d, \wt\psi} \le 4 \cdot \left( t(n) \cdot \sum_{h=1}^{t(n)} u_{d_{< h}} \right)^{1/4}~.
\end{equation*}
Below we prove that
\begin{equation}
	\label{eq:soundness-ip}
	1 - \sqrt{\bra{\wt\psi} \rho_d \ket{\wt\psi}}
	\le 6 \sum_{h=1}^{t(n)} \sqrt{u_{d_{< h}}},
\end{equation}
from which \Cref{lem:main-sound} follows because by \eqref{eq:td-fid} and Cauchy-Schwarz,
\begin{equation*}
	\td\Paren{\rho_d, \wt\psi}
	\le \sqrt{2 \Paren{1 - \sqrt{\bra{\wt\psi} \rho_d \ket{\wt\psi}}}}
	\le \sqrt{12 \sum_{h=1}^{t(n)} \sqrt{u_{d_{< h}}}}
	\le 4 \cdot \Paren{t(n) \cdot \sum_{h=1}^{t(n)} u_{d_{<h}}}^{1/4}~.
\end{equation*}

Observe that the (pure) state of the entire system (i.e.\ both the verifier's and prover's registers) at any point in the protocol depends only on $b$ and on the verifier's measurement outcomes, since the prover is unitary. For $h \in [t(n)]$ and $a = (a_1, \dotsc, a_h) \in \cube{h}$, let
\begin{equation*}
	\pi(a) = \PR{\text{verifier has not rejected by the end of iteration $h$} \mid b_{\le h} = a},
\end{equation*}
and let $\ket{\varphi_a}/\sqrt{\pi(a)}$ denote the state of the system at the end of iteration $h$, conditioned on $b_{\leq h} = a$ and the verifier not having rejected yet. Note that $\ket{\varphi_a}$ may be \emph{subnormalized}, or in other words $\norm{\ket{\varphi_a}} = \sqrt{\pi(a)}$ may be less than 1.

Clearly
\begin{equation*}
	\sqrt{\bra{\wt\psi} \rho_d \ket{\wt\psi}}
	= \frac{1}{\sqrt{\pi(d)}} \Norm{\bra{\wt\psi}_{\reg A_{[n]}} \ket{\varphi_d}} \geq \Norm{\bra{\wt\psi}_{\reg A_{[n]}} \ket{\varphi_d}}~.
\end{equation*}
(To clarify the notation, if we group together all registers except for $\reg A_{[n]}$ into a register $\reg{R}$, then the notation $\bra{\wt\psi}_{\reg A_{[n]}} \ket{\varphi_d}$ indicates a (subnormalized) state vector in $\reg{R}$, because $\ket{\varphi_d}$ is a state in ${\reg A_{[n]}} \reg{R}$ and the bra operator $\bra{\wt\psi}_{\reg A_{[n]}}$ implicitly acts as the identity on $\reg{R}$.) Thus to establish \eqref{eq:soundness-ip} it suffices to prove that
\begin{equation}\label{eq:soundness-goal-1} 
	1 - \Norm{\bra{\wt\psi}_{\reg A_{[n]}} \ket{\varphi_d}}
	\le 6 \sum_{h=1}^{t(n)} \sqrt{u_{d_{<h}}}~.
\end{equation}

Recall the definitions of states $\ket{\wt \psi\uppart k}$ for $0 \le k \le n+1$ from \cref{subsec:proof-comp}, and in particular that $\ket{\wt \psi \uppart{0}}$ is a phase of 1 and $\ket{\wt \psi\uppart{n+1}} = \ket{\wt\psi}$. Also let $\ket{\varphi_\emptyset}$ denote the initial state of the entire system (i.e.\ $\emptyset$ denotes the empty string here). Recalling that $\hw(d) \ge n+1$, let $\tau$ be the least number $h \in [t(n)]$ such that $\hw \Paren{d_{\le h}} = n+1$. Then, since the verifier does nothing after iteration $\tau$, by a telescoping sum we have that
\begin{align*}
	1 - &\Norm{\bra{\wt\psi}_{\reg A_{[n]}} \ket{\varphi_d}} \\
	&= \Norm{\bra{\wt\psi \uppart{0}}_{\reg A_{[0]}} \ket{\varphi_\emptyset}} - \Norm{\bra{\wt\psi \uppart{n+1}}_{\reg A_{[n+1]}} \ket{\varphi_{d_{\le \tau}}}} \\
	&= \sum_{h=1}^\tau  \Paren{ \Norm{\bra{\wt\psi \uppart{\hw \Paren{d_{< h}}}}_{\reg A_{\left[\hw \Paren{d_{< h}}\right]}} \ket{\varphi_{d_{< h}}}} - \Norm{\bra{\wt\psi \uppart{\hw \Paren{d_{\leq h}}}}_{\reg A_{\left[ \hw \Paren{d_{\leq h}} \right]}} \ket{\varphi_{d_{\leq h}}}}}~.
\end{align*}
Thus to establish~\eqref{eq:soundness-goal-1} it suffices to prove the following claim:

\begin{clm} \label{clm:sound-progress}
	For all $a \in \cube{<\tau}$ and $c \in \bits$,
	\begin{equation*}
		\Norm{\bra{\wt\psi \uppart{\hw(a)}}_{\reg A_{[\hw(a)]}} \ket{\varphi_{a}}} -
		\Norm{\bra{\wt\psi \uppart{\hw(ac)}}_{\reg A_{[\hw(ac)]}} \ket{\varphi_{ac}}}
		\le 6\sqrt{u_a}~.
	\end{equation*}
\end{clm}

Intuitively, this claim states that if $a$ denotes the random choices of the verifier up to some iteration $h < \tau$ and the state in $\reg{A}_{[\hw(a)]}$ is close to $\ket{\wt \psi \uppart{\hw(a)}}$ (which is the state that should have been synthesized up to that point), then after the next iteration (where the next random choice is the bit $c$), the state in $\reg{A}_{[\hw(ac)]}$ should still be close to $\ket{\wt \psi \uppart{\hw(ac)}}$, and the degradation in closeness is a polynomial function of the rejection probability in iteration $h$ (conditioned on the random choices $a$ up to that point).

Fix a prefix $a \in \cube{<\tau}$, and let $h = |a|+1$. After introducing the necessary notation, we prove the $c=0$ and $c=1$ cases of \cref{clm:sound-progress} for this value of $a$. Specifically, we define the following states, registers, unitary transformations, and orthogonal projections:
\paragraph{States.}
Abbreviate $\ket{\theta_0} = \ket{\wt\psi \uppart{\hw(a)}}$ and $\ket{\theta_1} = \ket{\wt\psi \uppart{\hw(a)+1}}$.

\paragraph{Registers.}
Recalling the definitions of registers from \cref{subsec:desc-prot}, let $\reg A = \reg A_{[\hw(a)]} = \reg A_{[\hw(a0)]}$ and $\reg{A^+} = \reg A_{[\hw(a1)]}$, and similarly let $\reg B = \reg B_{[\hw(a)]} = \reg B_{[\hw(a0)]}$ and $\reg{B^+} = \reg B_{[\hw(a1)]}$. Also let $\reg D = \reg D_h$ (for ``dyadic") and $\reg W = \reg W_h$ (for ``workspace"). Let $\reg P$ (for ``prover's workspace, among other things") be the register consisting of all qubits not in $\reg{A^+ B^+ D W}$.

\paragraph{Unitaries.}
Consider the unitary operations applied in iteration $h$ of the protocol. Let $F$ (for ``forward'') denote the unitary jointly applied by the verifier and prover in the forward-QIP step, which acts on $\reg{ADWP}$. Let $G$ (for ``grow'') be the unitary applied by the verifier in the A-grow step (when $b_h = 1$), which acts on $\reg{A^+} \, \reg{D}$. Let $R$ (for ``reverse'') denote the unitary jointly applied by the verifier and prover in the backward-QIP step, which acts on $\reg{ADWP}$. Let $C$ (for ``copy'') denote the unitary applied by the prover in the B-grow step (when $b_h = 1$), which acts on $\reg{B^+} \, \reg{P}$. 

\paragraph{Projections.}
Let $Y$ (for ``yes'') denote the projection $\kb1$ acting on the flag qubit of $\reg W$ (which indicates whether the $\QIP[1/2]$ verifier accepts or rejects). Let $S$ denote the projection onto the symmetric subspace between $\reg{A}$ and $\reg{B}$, and similarly let $S^+$ denote the projection onto the symmetric subspace between $\reg{A^+}$ and $\reg{B^+}$. Let $A = \id - S$ and $A^+ = \id - S^+$. Define projections
\begin{align*}
	T &=
	\begin{cases}
		\sum_{x \in \cube{\hw(a)}} \kb{x}_{\reg A} \otimes \kb{\cp(x)}_{\reg D} &\text{if } \hw(a) < n, \\
		\sum_{x \in \cube{\hw(a)}} \kb{x}_{\reg A} \otimes \kb{\ph(x)}_{\reg D} &\text{if } \hw(a) = n,
	\end{cases} \\
	L &= \id - T
\end{align*}
(for ``truth" and ``lie" respectively) acting on $\reg{AD}$. The operator $T$ projects onto the unique correct ``answers'' of the $\QIP[1/2]$ protocols. 

\paragraph{}

It follows from the definition of the protocol that
\begin{equation} \label{eq:varphi_a0_def}
	\ket{\varphi_{a0}} = SRF \ket{\varphi_a}~.
\end{equation}
This is because $\ket{\varphi_a}/\sqrt{\pi(a)}$ is the state of the system at the beginning of iteration $h$ conditioned on not having rejected yet, and if iteration $h$ is a ``test round'' (i.e. $c = 0$) then the verifier and prover first perform the forward-QIP step (thus applying the unitary $F$), the backward-QIP step (thus applying the unitary $R$), and then the verifier performs a swap test and accepts if the projection onto $S$ succeeds. Therefore
\begin{equation*}
	\pi(a0) = \pi(a) \cdot \Norm{ SRF \frac{\ket{\varphi_a}}{\sqrt{\pi(a)}} }^2 = \| SRF \ket{\varphi_a} \|^2,
\end{equation*}
which implies~\eqref{eq:varphi_a0_def} with the correct normalization factor. Similarly, for the $c = 1$ case, the subnormalized state $\ket{\varphi_{a1}}$ can be written as
\begin{equation} \label{eq:varphi_a1_def}
	\ket{\varphi_{a1}} = S^+ C R G Y F \ket{\varphi_a}~.
\end{equation}

We prove the following claim using the fact that $S$ (and $S^+$) projects onto the symmetric subspace:

\begin{clm} \label{clm:symmetric-states}
	The states $\ket{\varphi_a}, \ket{\varphi_{a0}}, \ket{\varphi_{a1}}$ satisfy
	\begin{gather*}
		\ket{\varphi_a} = \mai{Swap} \cdot \ket{\varphi_a} \\
		\ket{\varphi_{a0}} = \mai{Swap} \cdot \ket{\varphi_{a0}} \\
		\ket{\varphi_{a1}} = \mai{Swap}^+ \cdot \ket{\varphi_{a1}}
	\end{gather*}
	where $Swap$ (resp. $Swap^+$) denotes the swap unitary between registers $\reg{A}$ and $\reg{B}$ (resp.\ $\reg{A^+}$ and $\reg{B^+}$).
\end{clm}
\begin{proof}
	By definition $S = \frac{I + \mai{Swap}} 2$, so $S = \mai{Swap} \cdot S$ and therefore $\ket{\varphi_{a0}} = \mai{Swap} \cdot \ket{\varphi_{a0}}$. Similarly it holds that $\ket{\varphi_{a1}} = \mai{Swap}^+ \cdot \ket{\varphi_{a1}}$.
	
	If $|a| > 0$ then we can write $a = a' c'$ for some string $a' \in \cube{|a|-1}$ and bit $c' \in \bits$, and then similar reasoning implies that $\ket{\varphi_a} = \mai{Swap} \cdot \ket{\varphi_a}$. Alternatively, if $a$ is the empty string then $\reg A$ and $\reg B$ are zero-qubit registers, so $\mai{Swap}$ is the identity operator and therefore $\ket{\varphi_a} = \mai{Swap} \cdot \ket{\varphi_a}$.
\end{proof}

Recall that we define
\begin{equation*}
	r_{ac} = \pr{\text{verifier rejects in iteration $h$} \mid b_{\le h} = ac}
\end{equation*}
for $c \in \bits$. We now prove the $c=0$ case of \cref{clm:sound-progress}, i.e.\ that
\begin{equation*}
	\norm{\bra{\theta_0}_{\reg A} \ket{\varphi_a}} - \norm{\bra{\theta_0}_{\reg A} \ket{\varphi_{a0}}}
	\le 6\sqrt{u_a},
\end{equation*}
via applications of \cref{clm:symmetric-states}.

\begin{proof}[Proof of the $c=0$ case of \cref{clm:sound-progress}]
	By the symmetry of $\ket{\varphi_{a0}}$ (\cref{clm:symmetric-states}) and \eqref{eq:varphi_a0_def}, we have
	\begin{align*}
		\norm{\bra{\theta_0}_{\reg A} \ket{\varphi_{a0}}} 
		&= \norm{\bra{\theta_0}_{\reg A} \, Swap \, \ket{\varphi_{a0}}} 
		= \norm{\bra{\theta_0}_{\reg B} \ket{\varphi_{a0}}}
		= \norm{\bra{\theta_0}_{\reg B} SRF \ket{\varphi_a}} \\
		&= \norm{\bra{\theta_0}_{\reg B} (I - A)RF \ket{\varphi_a}}
		\ge \norm{\bra{\theta_0}_{\reg B} RF \ket{\varphi_a}} - \norm{\bra{\theta_0}_{\reg B} ARF \ket{\varphi_a}},
	\end{align*}
	where the last inequality is by the triangle inequality. By Cauchy-Schwarz,
	\begin{equation*}
		\norm{\bra{\theta_0}_{\reg B} ARF \ket{\varphi_a}}
		\le \norm{ARF \ket{\varphi_a}}
		= \sqrt{r_{a0}}
		\le \sqrt{u_a}
		\le 6 \sqrt{u_a}~.
	\end{equation*}
	Furthermore,
	\begin{equation*}
		\norm{\bra{\theta_0}_{\reg B} RF\ket{\varphi_a}}
		= \norm{\bra{\theta_0}_{\reg B} \ket{\varphi_a}}
		= \norm{\bra{\theta_0}_{\reg B} \cdot Swap \cdot \ket{\varphi_a}}
		= \norm{\bra{\theta_0}_{\reg A} \ket{\varphi_a}},
	\end{equation*}
	where the first equality holds because $RF$ is unitary and does not act on the register $\reg B$ (which is what $\bra{\theta_0}$ acts on), and the second equality holds by the symmetry of $\ket{\varphi_a}$ (\cref{clm:symmetric-states}). Therefore $\norm{\bra{\theta_0}_{\reg A} \ket{\varphi_{a0}}} \ge \norm{\bra{\theta_0}_{\reg A} \ket{\varphi_a}} - 6\sqrt{u_a}$, and rearranging gives the desired inequality.
\end{proof}

Now we prove the $c=1$ case of \cref{clm:sound-progress}, i.e.\ that
\begin{equation*}
	\norm{\bra{\theta_0}_{\reg A} \ket{\varphi_a}} - \norm{\bra{\theta_1}_{\reg A^+} \ket{\varphi_{a1}}}
	\le 6\sqrt{u_a}~.
\end{equation*}
Since the proof of the $c=1$ case is longer than that of the $c=0$ case, we break the proof of the $c=1$ case into a series of claims. The first claim, at a high level, relates $\norm{\bra{\theta_1}_{\reg A^+} \ket{\varphi_{a1}}}$ to the scenario where we assume that the prover computes the conditional probabilities/phases honestly (i.e.\ we project the state of the protocol onto $T$):

\begin{clm} \label{clm:start}
	$
	\norm{\bra{\theta_1}_{\reg A^+} \ket{\varphi_{a1}}}
	\ge \norm{\bra{\theta_1}_{\reg A^+} RGYTF \ket{\varphi_a}}
	- \norm{YLF \ket{\varphi_a}}
	- \sqrt{r_{a1}}$~.
\end{clm}
\begin{proof}
	By \eqref{eq:varphi_a1_def}, the fact that $S^+ + A^+ = \id$, and the triangle inequality, we have
	\begin{align*}
		\norm{\bra{\theta_1}_{\reg A^+} \ket{\varphi_{a1}}} &= \norm{\bra{\theta_1}_{\reg A^+} S^+ C R G Y F \ket{\varphi_a}} \ge \norm{\bra{\theta_1}_{\reg A^+} C R G Y F \ket{\varphi_a}} -
		\norm{\bra{\theta_1}_{\reg A^+} A^+ C R G Y F \ket{\varphi_a}}~.
	\end{align*}
	By Cauchy-Schwarz, we have
	\begin{align*}
		\norm{\bra{\theta_1}_{\reg A^+} A^+ C R G Y F  \ket{\varphi_a}}^2
		\le \norm{ A^+ C R G Y F  \ket{\varphi_a}}^2
		\le r_{a1},
	\end{align*}
	where the last inequality follows from the fact that $\norm{ A^+ C R G Y F  \ket{\varphi_a}}^2$ denotes the probability that, conditioned on the random choices being $a1$, in iteration $h$ the verifier accepts the $\QIP[1/2]$ protocols (thus the $Y$ projector) but rejects in the swap test (thus the $A^+$ projector). This is at most the probability of rejecting in iteration $h$ overall, conditioned on the random choices $a1$ (which is $r_{a1}$). 
	
	Next, since the unitary $C$ acts on $\reg{B^+ P}$ but not on $\reg A^+$, it holds that
	\begin{align*}
		\norm{\bra{\theta_1}_{\reg A^+} C R G Y F \ket{\varphi_a}}
		&= \norm{\bra{\theta_1}_{\reg A^+} RGYF \ket{\varphi_a}} \\
		&= \norm{\bra{\theta_1}_{\reg A^+} RGY (T + L)F \ket{\varphi_a}} \\
		&\geq \norm{\bra{\theta_1}_{\reg A^+} RGYTF \ket{\varphi_a}} - \norm{\bra{\theta_1}_{\reg A^+} RGYLF \ket{\varphi_a}} \\
		&\geq \norm{\bra{\theta_1}_{\reg A^+} RGYTF \ket{\varphi_a}} - \norm{YLF \ket{\varphi_a}}
	\end{align*}
	where the last inequality follows from Cauchy-Schwarz. Combining these inequalities yields the claim.
\end{proof}

The next claim is proved using the soundness guarantee of the $\QIP[1/2]$ verifiers $V_\cp$ and $V_\ph$:
\begin{clm} \label{clm:qip-sound}
	$
	\Norm{Y LF \ket{\varphi_a}}^2
	\le \frac{1}{2} \cdot {\Norm{LF \ket{\varphi_a}}}^2
	$.
\end{clm}
\begin{proof}
	Write $F = F^{(2)}_{\reg{WP}} F^{(1)}_{\reg{ADWP}}$, where $F^{(1)}$ consists of the ``pre-processing stage" (Lines~\ref{line:msendx-sec6} through~\ref{line:mcopyx-sec6} of the protocol: the verifier sends $x$, receives $\eta$ from the prover, and copies $(x,\eta)$ into $\reg W$) and $F^{(2)}$ consists of the $\QIP[1/2]$ protocol itself. Then since $L$ (which acts on $\reg{AD}$) commutes with $F^{(2)}$, we have that
	\begin{equation*}
		\Norm{Y LF \ket{\varphi_a}}^2 = \Norm{Y F^{(2)} \ket{\zeta}}^2 \quad \text{for} \quad \ket{\zeta} = L F^{(1)} \ket{\varphi_a}~.
	\end{equation*}
	By the definition of the ``lie'' projection $L$, we can write
	\begin{equation*}
		\ket{\zeta} = \sum_{\mathclap{\text{bad } (x,\eta)}} \kappa_{x,\eta} \, \ket{x,\eta}_{\reg{A} \reg{D}} \otimes \ket{x,\eta,\zs}_{\reg{W}_h} \otimes \ket{\zeta_{x,\eta}}
	\end{equation*}
	where $\kappa_{x,\eta}$ are complex numbers, $\ket{\zeta_{x,\eta}}$ are unit-length vectors, and we call $(x,\eta)$ \emph{bad} if $\eta \neq \cp(x)$ (in the case that $\hw(a) < n$) or if $\eta \neq \ph(x)$ (in the case that $\hw(a) = n$). In other words, the workspace register contains a superposition ``no'' instances of the $\QIP$ language $L_\cp$ or $L_\ph$. Thus the $\QIP[1/2]$ verifiers $V_\cp$ or $V_\ph$, when run on these bad instances, will accept with probability at most $1/2$ (by the soundness property of the $\QIP[1/2]$ protocols). Since $Y F^{(2)}$ acts on the register $\reg{WP}$, which is disjoint from $\reg{AD}$, it follows that
	\[
	\Norm{Y F^{(2)} \ket\zeta}^2 = \sum_{\mathclap{\text{bad } (x,\eta)}} |\kappa_{x,\eta}|^2 \cdot \Norm { YF^{(2)} \ket{x,\eta,\zs}_{\reg{W}_h} \ket{\zeta_{x,\eta}} }^2 \leq \sum_{\mathclap{\text{bad } (x,\eta)}} |\kappa_{x,\eta}|^2 \cdot \frac{1}{2} = \frac{1}{2} \cdot \Norm{ \ket{\zeta} }^2~.
	\]
	Finally, since $F^{(2)}$ is unitary and commutes with $L$, we have
	\begin{equation*}
		\Norm{ \ket{\zeta} }^2
		= \Norm{ L F^{(1)} \ket{\varphi_a}}^2
		= \Norm{ L F^{(2)} F^{(1)} \ket{\varphi_a}}^2
		= \Norm{ L F \ket{\varphi_a}}^2~.
	\end{equation*}
	Combining these (in)equalities yields the claim.
\end{proof}

Roughly speaking, the next claim says that if the prover behaves honestly at the beginning of iteration $h$, then the verifier grows the state by a qubit (or if $\hw(a) = n$, applies a phase) in the correct way:
\begin{clm} \label{clm:grow}
	$
	\Norm{\bra{\theta_1}_{\reg A^+} RG YTF \ket{\varphi_a}}
	= \Norm{\bra{\theta_0}_{\reg A} RYTF\ket{\varphi_a}}
	$.
\end{clm}
\begin{proof}
	Define a unitary $\wt G$ acting on $\reg A^+$ as follows: If $\hw(a) < n$ then
	\begin{equation*}
		\wt G = \sum_{\mathclap{x \in \cube{\hw(a)}}} \kb{x}_{\reg A} \otimes Q^{(x)}_{\reg A_{\hw(a)+1}}
	\end{equation*}
	for some one-qubit unitaries $Q^{(x)}$ satisfying
	\begin{equation*}
		Q^{(x)} \ket0 = \sqrt{\cp(x)} \ket0 + \sqrt{1-\cp(x)} \ket1~.
	\end{equation*}
	Alternatively, if $\hw(a) = n$ then let
	\begin{equation*}
		\wt G = \sum_{\mathclap{x \in \cube n}} \ph(x) \kb{x}~.
	\end{equation*}
	In other words, the unitary $\wt{G}$ ``grows'' the state based on the true conditional probability or phase\footnote{Up to the overall phase $\gamma_n$ by which $\ket\psi$ differs from some close approximation of $\ket{\wt\psi}$.}, rather than based on the claimed value given by the prover (as with the unitary $G$).
	
	It follows from definitions that $GT = \wt G T$, and that $\wt G$ commutes with $R$. This latter point holds because we can write $R = \sum_{x \in \cube{\hw(a)}} \kb{x}_{\reg A} \otimes U^{(x)}$ for some unitaries $U^{(x)}$ acting on $\reg{DWP}$, and registers $\reg A_{\hw(a)+1}$ and $\reg{DWP}$ are disjoint. Thus, since $Y$ (which acts on $\reg W$) commutes with $T$ (which acts on $\reg{AD}$), we have
	\begin{align*}
		\Norm{\bra{\theta_1}_{\reg A^+} R GYTF \ket{\varphi_a}}
		= \Norm{\bra{\theta_1}_{\reg A^+} R \wt G YTF \ket{\varphi_a}} = \Norm{\bra{\theta_1}_{\reg A^+} \wt G R YT F \ket{\varphi_a}}~.
	\end{align*}
	
	Suppose that $\hw(a) < n$. Then $\wt{G} \ket{\theta_0}_{\reg{A}} \ket{0}_{\reg{A}_{\hw(a)+1}} = \ket{\theta_1}_{\reg{A^+}}$, so
	\begin{equation*}
		\Norm{\bra{\theta_1}_{\reg A^+} \wt G R YT F \ket{\varphi_a}} = \Norm{\bra{\theta_0,0}_{\reg{A^+}} RYTF \ket{\varphi_a}} = \Norm{\bra{\theta_0}_{\reg{A}} RYTF \ket{\varphi_a}},
	\end{equation*}
	where the last equality holds because the register ${\reg{A}_{\hw(a)+1}}$ in the state $RYTF \ket{\varphi_a}$ is in the state $\ket{0}$ (because the protocol has not acted on that register yet). Alternatively, suppose that $\hw(a) = n$. Then $\wt{G} \ket{\theta_0}_{\reg{A}} = \ket{\theta_1}_{\reg{A}}$, so in this case it also holds that
	\begin{equation*}
		\Norm{\bra{\theta_1}_{\reg A^+} \wt G R YT F \ket{\varphi_a}} =
		\Norm{\bra{\theta_0}_{\reg{A}} RYTF \ket{\varphi_a}}~.
	\end{equation*}
	The claim follows from combining the above equalities.
\end{proof}

Finally, we combine these claims to prove the $c=1$ case of \cref{clm:sound-progress}:

\begin{proof}[Proof of the $c=1$ case of \cref{clm:sound-progress}]
	By \cref{clm:start,clm:grow},
	\begin{align*}
		\norm{\bra{\theta_1}_{\reg A^+} \ket{\varphi_{a1}}}
		&\ge \norm{\bra{\theta_1}_{\reg A^+} RG Y TF  \ket{\varphi_a}}
		- \norm{YLF  \ket{\varphi_a}}
		- \sqrt{r_{a1}} \\
		&= \Norm{\bra{\theta_0}_{\reg A} RYTF \ket{\varphi_a}}
		- \norm{YLF \ket{\varphi_a}}
		- \sqrt{r_{a1}}~.
	\end{align*}
	\cref{clm:qip-sound} implies that
	\begin{align*}
		\Norm{YLF \ket{\varphi_a}}^2
		\le \frac{1}{2} {\Norm{LF \ket{\varphi_a}}}^2 = \frac{1}{2}\Norm{YLF \ket{\varphi_a}}^2
		+ \frac{1}{2} \Norm{(\id - Y) LF \ket{\varphi_a}}^2~.
	\end{align*}
	Rearranging, and then using that $(\id - Y)_{\reg W}$ and $L_{\reg{AD}}$ commute and $\norm{L}_{\mrm{op}} \le 1$, we get
	\begin{equation}
		\norm{YLF \ket{\varphi_a}}^2
		\le \norm{(\id - Y)LF \ket{\varphi_a}}^2
		= \norm{L(\id - Y)F \ket{\varphi_a}}^2
		\le \norm{(\id - Y)F \ket{\varphi_a}}^2
		\le r_{a1}~. \label{eq:c1-1}
	\end{equation}
	The last inequality holds because $\norm{(\id - Y)F \ket{\varphi_a}}^2$ denotes the probability that the verifier rejects in the A-grow step of iteration $h$ (conditioned on the random choices $a1$), which is at most the probability that the verifier rejects in iteration $h$ (conditioned on the random choices $a1$), a.k.a.\ $r_a$.
	
	Thus
	\begin{align*}
		\norm{\bra{\theta_1}_{\reg A^+} \ket{\varphi_{a1}}}
		&\ge \norm{\bra{\theta_0}_{\reg A} R Y T F \ket{\varphi_a}} - 2\sqrt{r_{a1}} \\
		&\geq \norm{\bra{\theta_0}_{\reg A} RYF \ket{\varphi_a}}
		- \norm{\bra{\theta_0}_{\reg A} RYLF \ket{\varphi_a}} - 2\sqrt{r_{a1}} \\
		&\ge \norm{\bra{\theta_0}_{\reg A} RF \ket{\varphi_a}}
		- \norm{\bra{\theta_0}_{\reg A} R(\id - Y)F \ket{\varphi_a}}
		- \norm{YLF \ket{\varphi_a}}  - 2\sqrt{r_{a1}} \\
		&\ge \norm{\bra{\theta_0}_{\reg A} RF \ket{\varphi_a}}
		- \norm{(\id - Y) F \ket{\varphi_a}}
		- 3\sqrt{r_{a1}} \\
		&\ge \norm{\bra{\theta_0}_{\reg A}RF \ket{\varphi_a}}
		- 4\sqrt{r_{a1}}~.
	\end{align*}
	The second line follows from using $T = \id - L$ and the triangle inequality. The third line follows from $\id + (Y - \id) = Y$ and the triangle inequality, as well as Cauchy-Schwarz. The fourth line follows from Cauchy-Schwarz, as well as \eqref{eq:c1-1}. The last line follows from the final inequality in \eqref{eq:c1-1} again.	
	
	Finally, by reasoning similar to that in the $c=0$ case,
	\begin{align*}
		\norm{\bra{\theta_0}_{\reg A} RF \ket{\varphi_a}}
		&\ge \norm{\bra{\theta_0}_{\reg A} SRF \ket{\varphi_a}}
		- \norm{\bra{\theta_0}_{\reg A} ARF\ket{\varphi_a}} \\
		&\ge \norm{\bra{\theta_0}_{\reg B} S RF \ket{\varphi_a}}
		- \norm{A RF\ket{\varphi_a}} \\
		&\ge \norm{\bra{\theta_0}_{\reg B} RF \ket{\varphi_a}}
		- \norm{\bra{\theta_0}_{\reg B} ARF \ket{\varphi_a}}
		- \sqrt{r_{a0}} \\
		&\ge \norm{\bra{\theta_0}_{\reg B} \ket{\varphi_a}}
		- \norm{ARF\ket{\varphi_a}}
		- \sqrt{r_{a0}} \\
		&= \norm{\bra{\theta_0}_{\reg A} \ket{\varphi_a}} - 2\sqrt{r_{a0}},
	\end{align*}
	so $\norm{\bra{\theta_1}_{\reg A^+} \ket{\varphi_{a1}}} \ge \norm{\bra{\theta_0}_{\reg A} \ket{\varphi_a}} - 2\sqrt{r_{a0}} - 4\sqrt{r_{a1}}$. Since $r_{a0}, r_{a1} \leq u_{a}$, this implies \cref{clm:sound-progress} in the $c=1$ case as desired.
\end{proof}

	\section{Interactive Unitary Synthesis}
\label{sec:transformations}
In this section we present our unitary synthesis protocol for unitaries in $\cc{unitaryPSPACE}$ with polynomial action, i.e.\ \Cref{thm:poly-action}. We also discuss a unitary synthesis protocol for general unitary families, provided that the verifier also receives a succinct description of a polynomial-dimensional subspace that contains the input state.
\subsection{Protocol for unitaries with polynomial action}
\label{sec:poly-action}
Our unitary synthesis protocol is based on the following algorithm, which we denote $\alg{LMR}$ after Lloyd, Mohseni, and Rebentrost who formulated it~\cite{lloyd2014quantum}. 

\begin{thm}[\cite{lloyd2014quantum,kimmel2017hamiltonian}] \label{thm:lmr}
	There exists a quantum polynomial-time algorithm $\alg{LMR}$ that takes as input a state $\tau \otimes \rho^{\otimes k} \otimes \kb{t}$, where $\tau$ and $\rho$ are $n$-qubit mixed states and $t \ge 0$ is a number written in binary, and outputs an $n$-qubit mixed state $\sigma$ such that
	\begin{equation*}
		\td(\sigma, \, W \tau \adj W) \le O(t^2/k) \qquad \text{for} \qquad W = \exp(2\pi i \cdot t \cdot \rho)~.
	\end{equation*}
\end{thm}

We introduce the following notation related to \cref{thm:lmr}. If $U$ is a unitary acting on $n$ qubits, then for $t, \eps \ge 0$ we call an $n$-qubit mixed state $\rho$ a \emph{program state for $U$ with evolution time $t$ and error $\eps$} if for all $n$-qubit states $\ket\phi$,
\begin{equation*}
	\td \Paren{U \phi \adj U, \, W \phi \adj W} \le \eps \qquad \text{for} \qquad W = \exp(2 \pi i \cdot t \cdot \rho)~.
\end{equation*}
For example, if 
\begin{equation} \label{eq:def-U-t}
	U = \sum_{j=1}^{2^n} e^{2\pi i \theta_j} \kb{v_j}
\end{equation}
is an eigendecomposition of $U$ where $0 \le \theta_j < 1$ for all $j$, then for
\begin{align} \label{eq:def-rho-t}
	&t = \sum_j \theta_j, &\rho = \frac1t \sum_j \theta_j \kb{v_j}
\end{align}
the state $\rho$ is a program state for $U$ with evolution time $t$ and error 0 (if $t>0$). We call $\rho$ and $t$ respectively the \emph{canonical program state for $U$} and the \emph{canonical evolution time for $U$}, and note that $\rho$ and $t$ do not depend on the chosen eigendecomposition of $U$.

Observe that if $U$ has action $a$ (i.e.\ $U$ acts nontrivially on a subspace of dimension $a$) then the canonical evolution time for $U$ is at most $a$, because
\begin{equation*}
	t = \sum_j \theta_j \le \sum_j \ceil{\theta_j} = a~.
\end{equation*}
Therefore \cref{thm:poly-action} follows from the following theorem, which is identical to \cref{thm:poly-action} but with ``canonical evolution time" in place of ``action":

\begin{thm} \label{thm:main-7p1}
	Let $(U_n)_{n \in \N}$ be a family in \unitaryPSPACE such that $U_n$ has canonical evolution time at most $\poly(n)$ for all $n$, and let $q$ be a polynomial. Then $(U_n)_{n \in \N} \in \cc{unitaryQIP}[c,s]$ for
	\begin{align*}
		&c(n) = \frac{1}{q(n)},
		&s(n,\delta) = \exp \Paren{\frac{1}{q(n)} - q(n) \cdot \delta^4}~.
	\end{align*}
\end{thm}

In the rest of this subsection we prove \cref{thm:main-7p1}. In what follows we relax the definition of \statePSPACE to allow the $n$'th state in a sequence to be on $\poly(n)$ qubits rather than exactly $n$ qubits. Clearly our state synthesis theorem (\cref{thm:prim-result}) holds even with this more general definition, a fact which we will use. Also let $\D{m}[k]$ be the set of integer multiples of $2^{-m}$ in the interval $[0,k)$ (e.g.\ $\D{m}[1] = \D{m}$), encoded as binary strings in the natural way. Throughout this subsection, asymptotic notation hides universal constants (as opposed to constants that depend on some named polynomial or sequence of unitaries). Our proof uses the following lemma:

\begin{lem} \label{lem:helper-7p1}
	Let $(U_n)_{n \in \N} \in \unitaryPSPACE$. Then there exists
	\begin{itemize}[leftmargin=*]
		\renewcommand\labelitemi{--}
		\item a $\PSPACE$-computable function $f$ such that $f(1^n) \in \D{\poly(n)}[2^n]$ and $f(1^n) \leq t_n + e^{-\Omega(n)}$, where $t_n$ is the canonical evolution time for $U_n$,
		\item a sequence $(\ket{\psi_n})_{n \in \N} \in \statePSPACE$,
	\end{itemize}
	such that for all $n$ the reduced state of $\ket{\psi_n}$ on the first $n$ qubits is a program state for $U_n$ with evolution time $f(1^n)$ and error $e^{-\Omega(n)}$.
\end{lem}

Our unitary synthesis protocol prepares copies of $\ket{\psi_n}$ and computes $f(1^n)$ (for $\ket{\psi_n}, f$ as defined in \cref{lem:helper-7p1}) using our state synthesis protocol, and then applies the $\alg{LMR}$ algorithm to the appropriate reduced state. First we prove \cref{lem:helper-7p1} in a certain special case (explained below), then we prove \cref{lem:helper-7p1} in the general case by reducing to the special case, and finally we prove \cref{thm:main-7p1} using \cref{lem:helper-7p1}.

A ``problem" with the definitions of $\rho$ and $t$ in \eqref{eq:def-rho-t} is that they are sensitive to small perturbations in the unitary $U$ from \eqref{eq:def-U-t}. For example, if $\theta_j = 0$ for some $j$, then an arbitrarily small perturbation to the $j$'th eigenvalue of $U$ could change $\theta_j$ to near 1. Furthermore $\rho$ is undefined when $t=0$, and is sensitive to a slight increase in any of the $\theta_j$ when $t$ is near zero. Motivated by these concerns, we define a notion of \emph{stability} of a unitary, and first prove \cref{lem:helper-7p1} in the case where $U_n$ is stable for all $n$.

\begin{dfn}[Stability of unitaries]
	An $n$-qubit unitary $U$ is \emph{stable} if all of its eigenvalues are of the form $e^{2 \pi i \theta}$ where $2^{-3n} \le \theta \le 1-2^{-3n}$.
\end{dfn}

The following equivalent definition of stability will sometimes be more convenient to work with. Define a metric $\Delta$ on the real numbers as follows:
\begin{equation*}
	\Delta(r,s)
	= \min_{k \in \Z} |r - s + k|
	= \min\Paren{r - s - \floor{r-s}, \, \ceil{r-s} - (r-s)}~.
\end{equation*}
Intuitively, this corresponds to mapping the real line to the unit circle by identifying all integer points with each other, and measuring the distance between two points on the resulting $1$-dimensional torus. Thus, an $n$-qubt unitary $U$ is stable if all of its eigenvalues are of the form $e^{2 \pi i \theta}$ where $\Delta(\theta, 0) \ge 2^{-3n}$.
\subsubsection{Proof of \texorpdfstring{\cref{lem:helper-7p1}}{Lemma 7.3} when \texorpdfstring{$U_n$}{Un} is stable}
The proof is organized as follows. First we review the well-known \emph{phase estimation algorithm}. Then, using the phase estimation algorithm, we define and analyze a sequence of quantum circuits $(C_n)_{n \in \N}$ which are used to define both $(\ket{\psi_n})_{n \in \N}$ and $f$. More specifically, we define $\ket{\psi_n}$ as the output state of $C_n$ when $C_n$ accepts, and $f(1^n)$ as $2^n$ times the approximation of the acceptance probability of $C_n$ obtained using \cref{tomog-0}. Then we prove that $(\ket{\psi_n})_{n \in \N}$ is in \statePSPACE. Finally we prove that $f$ satisfies the required properties.
\paragraph{The phase estimation algorithm.}
For an $n$-qubit unitary $U$ and number $m \in \N$, we denote by $\alg{PE}^{(U,m)}$ the instance of the phase estimation algorithm that acts on an $n$-qubit register $\reg A$ (the ``eigenvector register") and an $m$-qubit register $\reg B$ (the ``eigenvalue register") and makes oracle calls to $U_{\reg A}$ controlled on the content of $\reg B$. If $\ket v$ is an eigenvector of $U$ with eigenvalue $e^{2\pi i \theta}$, then
\begin{equation*}
	\alg{PE}^{(U,m)} \ket{v}_{\reg A} \ket{0^m}_{\reg B} = \ket{v}_{\reg A} \ket\eta_{\reg B}
\end{equation*}
for some state $\ket\eta$ (depending on $\ket v, U, m$), such that if $r \in \D m$ denotes the outcome of a standard-basis measurement of $\ket\eta$ then
\begin{equation} \label{eq:pe}
	\PR{\Delta(r,\theta) \ge \eps} \le O\Paren{2^{-m} / \eps}
\end{equation}
for all $\eps>0$~\cite[Chapter 5]{nielsen2000quantum}. Let $m(n) = 9n$ and $P_n = \alg{PE}^{\Paren{U_n, m(n)}}$. Since $(U_n)_{n \in \N}$ is space-uniform, the sequence $(P_n)_{n \in \N}$ is also space-uniform, as can be seen by inspection of the phase estimation algorithm.

\paragraph{The circuit $C_n$ and its properties.}

\begin{algorithm}
	\caption{The circuit $C_n$} \label{alg:Cn}
	\begin{algorithmic}[1]
		\State Initialize $n$-qubit registers $\reg A$ and $\reg B$ to the maximally entangled state $\ket{\Phi_n}_{\reg{AB}} = 2^{-n/2} \sum_{x \in \cube n} \ket{x}_{\reg A} \ket{x}_{\reg B}$.
		\State Initialize an $m(n)$-qubit register $\reg C$ to $\ket{0^{m(n)}}$, and apply the phase estimation circuit $P_n$ with eigenvector register $\reg A$ and eigenvalue register $\reg C$. \hlabel{line:Cnl2}
		\State Create a one-qubit register $\reg D$, and controlled on the state $\ket r$ of $\reg C$ where $r \in \D{m(n)}$, construct the state $\sqrt{r} \ket0 + \sqrt{1-r} \ket1$ in $\reg D$. \hlabel{line:Cnl3}
		\State Measure $\reg D$ in the standard basis. If the measurement outcome is 0 then accept and output $\reg{ABC}$, otherwise reject.
	\end{algorithmic}
\end{algorithm}

The circuit $C_n$ is described in \cref{alg:Cn}; clearly $(C_n)_{n \in \N}$ is space-uniform. For a fixed $n \in \N$, when analyzing $C_n$ we use the following notation. Let
\begin{equation*}
	U_n = \sum_{j=1}^{2^n} e^{2\pi i \theta_j} \kb{v_j}
\end{equation*}
be an eigendecomposition of $U_n$ where $0 \le \theta_j < 1$ for all $j$. (The phases $e^{2 \pi i \theta_j}$ and eigenvectors $\ket{v_j}$ depend on $n$, but for notational clarity we leave this dependence implicit.) Let
\begin{align*}
	&t = \sum_j \theta_j,
	&\rho = \frac1t \sum_j \theta_j \kb{v_j}
\end{align*}
respectively denote the canonical evolution time and canonical program state for $U_n$. Let $\ket{\eta_j} = \sum_{r \in \D{m(n)}} \alpha_{jr} \, \ket{r}$ be the $m(n)$-qubit state such that
\begin{equation*}
	P_n \ket{v_j} \ket{0^{m(n)}} = \ket{v_j} \ket{\eta_j}~.
\end{equation*}

Let $\ket{\overline v_j}$ denote the element-wise complex conjugate of $\ket{v_j}$ with respect to the standard basis. It is easily verified\footnote{E.g.\ using the fact that $(V \otimes V^*) \ket{\Phi_n} = \ket{\Phi_n}$ for all unitaries $V$ (in particular $V = \sum_j \ketbra{v_j}{j}$), where $V^*$ denotes the element-wise complex conjugate of $V$.} that
\begin{equation*}
	\ket{\Phi_n} = 2^{-n/2} \sum_{j \in [2^n]} \ket{v_j} \ket{\overline v_j}~.
\end{equation*}
Therefore the state after \cref{line:Cnl2} is $2^{-n/2} \sum_j \ket{v_j} \ket{ \overline v_j} \ket{ \eta_j}$, and so the state after \cref{line:Cnl3} is
\begin{equation*}
	2^{-n/2}	\sum_{\mathclap{j \in [2^n], r \in \D{m(n)}}}  \alpha_{jr} \cdot \ket{v_j} \ket{ \overline v_j} \ket{r} \otimes (\sqrt{r} \ket0 + \sqrt{1-r} \ket1)~.
\end{equation*}
Define
\begin{align*}
	&\wt\theta_j = \sum_{\mathclap{r \in \D{m(n)}}} |\alpha_{jr}|^2 \, r,
	&\wt t = \sum_j \wt \theta_j~.
\end{align*}
Then
\begin{equation}
	\label{eq:cn_acc_prob}
	\pr{C_n \text{ accepts}} = 2^{-n} \cdot \wt t,
\end{equation}
and when accepting $C_n$ outputs the state
\begin{equation*}
	\ket{\psi_n} = \wt{t}^{-1/2} \, \sum_{\mathclap{j \in [2^n], r \in \D{m(n)}}} \alpha_{jr} \cdot \sqrt{r} \cdot \ket{v_j} \ket{ \overline v_j} \ket{ r}~.
\end{equation*}
The reduced state on the first $n$ qubits of $\ket{\psi_n}$ is thus
\begin{equation*}
	\wt\rho = \frac1{\wt t} \sum_j \wt\theta_j \kb{v_j}~.
\end{equation*}

Now we argue that
\begin{equation} \label{eq:theta-approx}
	\left|\wt\theta_j - \theta_j \right| \le 2^{-4n}
\end{equation}
for all $j$, provided $n$ is sufficiently large. Set $\eps = \frac12 \cdot 2^{-4n}$. If $r$ denotes the outcome of a standard-basis measurement of $\ket{\eta_j}$, then
\begin{align*}
	\left| \wt\theta_j - \theta_j \right|
	&= | \E[r] - \theta_j |
	= |\E[r - \theta_j]|
	\le \E|r - \theta_j| \\
	&= \E[|r-\theta_j| \cdot \Ind{|r - \theta_j| \le \eps}] + \E[|r-\theta_j| \cdot \Ind{|r-\theta_j| > \eps}] \\
	&\le \eps + \pr{|r-\theta_j| >\eps}
	=\eps + \PR{\Delta(r,\theta_j) > \eps},
\end{align*}
where the last equality holds because $U_n$ is stable and $\eps < 2^{-3n}$. By \eqref{eq:pe} it follows that
\begin{equation*}
	\left| \wt\theta_j - \theta_j \right|
	\le \eps + O\Paren{2^{-m(n)} / \eps}
	= \eps + O\Paren{2^{-5n}}
	\le 2\eps,
\end{equation*}
which establishes \eqref{eq:theta-approx}.

\paragraph{Proof that $(\ket{\psi_n})_{n \in \N}$ is in \statePSPACE.}

Let $\sigma_n$ denote the output state of the general quantum circuit $D_n$ described in \cref{alg:Dn}, and let $\ell(n) = 2^{4n}$. Observe that $(D_n)_{n \in \N}$ is space-uniform. Clearly
\begin{equation*}
	\sigma_n = \pr{C_n \text{ rejects}}^{\ell(n)} \cdot \kb\zs + \Paren{1 - \pr{C_n \text{ rejects}}^{\ell(n)}} \cdot \kb{\psi_n},
\end{equation*}
so by the convexity of trace distance and \eqref{eq:cn_acc_prob},
\begin{align*}
	\td(\sigma_n, \psi_n)
	&\le \pr{C_n \text{ rejects}}^{\ell(n)} \td(\zs, \psi_n)
	\le \Paren{1 - 2^{-n} \wt t}^{\ell(n)}
	\le \exp\Paren{-2^{-n} \cdot \wt t \cdot \ell(n)} \\
	&= \exp\Paren{-2^{3n} \cdot \wt t}~.
\end{align*}
By \eqref{eq:theta-approx} and the fact that $U_n$ is stable,
\begin{equation} \label{eq:ttilb}
	\wt t
	= \sum_j \wt \theta_j
	= \sum_j \Paren{\theta_j - \Paren{\theta_j - \wt \theta_j}}
	\ge \sum_j (2^{-3n} - 2^{-4n})
	\ge \Omega(2^{-2n}),
\end{equation}
so $\td(\sigma_n, \psi_n) \le \exp\Paren{-\Omega(2^n)} \le \exp\Paren{-n^{\omega(1)}}$ as desired.

\begin{algorithm}
	\caption{The circuit $D_n$} \label{alg:Dn}
	\begin{algorithmic}[1]
		\For{$2^{4n}$ times}
		\State Execute $C_n$.
		\IIf{$C_n$ accepts} \Return the output state of $C_n$ and \textbf{abort}.
		\EndIIf
		\EndFor
		\State \Return $\ket\zs$.
	\end{algorithmic}
\end{algorithm}

\paragraph{The function $f$ and its properties.}

Let $C'_n$ be identical to $C_n$ except that $C'_n$ only outputs an accept/reject bit (as opposed to also outputting registers $\reg{ABC}$). By \cref{tomog-0} applied to $(C'_n)_{n \in \N}$, and \eqref{eq:cn_acc_prob}, there exists a $\PSPACE$-computable function $g$ such that for all $n \in \N$ it holds that $g(1^n) \in \D{\poly(n)}$ and
\begin{equation*}
	\left| g(1^n) - 2^{-n} \wt t \right| \le 2^{-4n}~.
\end{equation*}
Let $f(1^n) = 2^n g(1^n)$, and observe that $f$ is computable in $\PSPACE$ (since $g$ is) and that
\begin{equation} \label{eq:sttil}
	\left|f(1^n) - \wt t \right| = 2^n \left| g(1^n) - 2^{-n} \wt t \right| \le 2^{-3n}~.
\end{equation}
By \eqref{eq:theta-approx}, \eqref{eq:sttil} and the triangle inequality,
\begin{equation*}
	|f(1^n) - t|
	\le \left| f(1^n) - \wt t \right| + \left| \wt t - t \right|
	\le 2^{-3n} + \sum_j \left| \wt \theta_j - \theta_j \right|
	\le 2^{-3n} + 2^n \cdot 2^{-4n} = 2 \cdot 2^{-3n},
\end{equation*}
so $f(1^n) \le t + e^{-\Omega(n)}$ as required.

Fixing $n$, all that remains is to prove that $\wt\rho$ is a program state for $U = U_n$ with evolution time $f = f(1^n)$ and error $e^{-\Omega(n)}$. In other words, given an arbitrary $n$-qubit state $\ket\phi$, we would like to prove that
\begin{equation} \label{eq:fin-goal}
	\td\Paren{U \phi \adj U, \, \wt U \phi \adj{\wt U}} \le e^{-\Omega(n)}
\end{equation}
for
\begin{equation*}
	\wt U = \exp\Paren{2\pi i \cdot f  \cdot \wt\rho}
	= \sum_j \exp\Paren{2\pi i \cdot f \wt\theta_j / \wt t} \kb{v_j}~.
\end{equation*}
By \eqref{eq:td-fid2},
\begin{equation*}
	\td\Paren{U \phi \adj U, \, \wt U \phi \adj{\wt U}}
	\le \Norm{U\ket\phi - \wt U \ket\phi}
	\le \Norm{U - \wt U}_{\mrm{op}}~.
\end{equation*}
Since $\ket{v_1}, \dotsc, \ket{v_{2^n}}$ are orthogonal and the operator norm equals the largest singular value,
\begin{align*}
	\Norm{U - \wt U}_{\mrm{op}}
	&= \Norm{\sum_j \Paren{\exp\Paren{2\pi i \cdot \theta_j} - \exp\Paren{2\pi i \cdot f \wt\theta_j / \wt t}} \kb{v_j}}_{\mrm{op}} \\
	&= \max_j \left| \exp\Paren{2\pi i \cdot \theta_j} - \exp\Paren{2\pi i \cdot f \wt\theta_j / \wt t} \right| \\
	&\le 2\pi \max_j \left| \theta_j - f \wt\theta_j / \wt t \right|,
\end{align*}
where the last inequality is by \eqref{eq:integral}. Finally, for all $j$,
\begin{equation*}
	\left| \theta_j - \frac{f \wt\theta_j} {\wt t} \right|
	= \left| \frac{\theta_j \Paren{\wt t - f} + f \Paren{\theta_j - \wt \theta_j}} {\wt t} \right|
	\le \frac{\theta_j \left|\wt t - f\right| + f \left|\theta_j - \wt\theta_j \right|} {\left| \wt t\right|}
	\le \frac{\left|\wt t - f \right| + 2^n \left|\theta_j - \wt\theta_j \right|} {\left| \wt t\right|},
\end{equation*}
and by \eqref{eq:theta-approx}, \eqref{eq:ttilb}, \eqref{eq:sttil} it holds that
\begin{equation*}
	\frac{\left|\wt t - f \right| + 2^n \left|\theta_j - \wt\theta_j \right|} {\left| \wt t\right|}
	\le O\Paren{\frac{2^{-3n} + 2^n \cdot 2^{-4n}} {2^{-2n}}}
	\le e^{-\Omega(n)}
\end{equation*}
from which \eqref{eq:fin-goal} follows.

\subsubsection{Proof of \texorpdfstring{\cref{lem:helper-7p1}}{Lemma 7.3} in the general case}
Again, given $n$ let $U_n = \sum_{j=1}^{2^n} e^{2\pi i \theta_j} \kb{v_j}$ be an eigendecomposition of $U_n$ where $0 \le \theta_j < 1$ for all $j$. We now consider the case where $U_n$ may not be stable for all $n$. To remedy this, we reduce to the stable case via the following claim:

\begin{clm} \label{clm:712}
	There exists a $\PSPACE$-computable function $\phi$ such that for all $n \in \N$ it holds that $\phi(1^n) \in \D{\poly(n)}, \phi(1^n) \le O\Paren{2^{-2n}}$ and the unitary $e^{2 \pi i \phi(1^n)} U_n$ is stable.
\end{clm}

First we prove \cref{lem:helper-7p1} assuming \cref{clm:712}, and then we prove \cref{clm:712}.
\begin{proof}[Proof of \cref{lem:helper-7p1} assuming \cref{clm:712}]
	Let $U_n' = e^{2 \pi i \phi(1^n)} U_n$. The family $(U_n')_{n \in \N}$ is space-uniform, because $(U_n)_{n \in \N}$ is space-uniform and $\phi(1^n)$ is $\PSPACE$-computable.\footnote{Technically, the phase $\exp\Paren{-2\pi i \phi \Paren{1^n}}$ may not be implementable exactly using our assumed gate set; however it can be approximated with exponentially small error that does not alter the analysis. For clarity we assume that the phase $\exp\Paren{-2\pi i \phi \Paren{1^n}}$ can be implemented exactly.} Since $U_n'$ is furthermore stable for all $n$, by the special case of \cref{lem:helper-7p1} proved above there exists
	\begin{itemize}[leftmargin=*]
		\renewcommand\labelitemi{--}
		\item a $\PSPACE$-computable function $f$ such that $f(1^n) \in \D{\poly(n)}[2^n]$ and $f(1^n) \leq t_n' + e^{-\Omega(n)}$, where $t_n'$ is the canonical evolution time for $U'_n$,
		\item a sequence $(\ket{\psi_n})_{n \in \N} \in \statePSPACE$,
	\end{itemize}
	such that for all $n$ the reduced state of $\ket{\psi_n}$ on the first $n$ qubits is a program state for $U'_n$ with evolution time $f(1^n)$ and error $e^{-\Omega(n)}$. Since $U_n'$ and $U_n$ differ only by an overall phase, the reduced state of $\ket{\psi_n}$ on the first $n$ qubits is also a program state for $U_n$ with evolution time $f(1^n)$ and error $e^{-\Omega(n)}$. Finally, letting $t_n = \sum_j \theta_j$ denote the canonical evolution time for $U_n$, we have
	\begin{equation*}
		t_n' = \sum_j (\theta_j + \phi(1^n) - \ceil{\theta_j + \phi(1^n)})
		\le \sum_j (\theta_j + \phi(1^n))
		= \sum_j \theta_j + 2^n \cdot \phi(1^n)
		\le t_n + 2^n \cdot O\Paren{2^{-2n}},
	\end{equation*}
	so $f(1^n) \le t_n' + e^{-\Omega(n)} \le t_n + e^{-\Omega(n)}$ as desired.
\end{proof}
Now we prove \cref{clm:712}. The function $\phi$ is defined relative to the circuit $E_n$ described in \cref{alg:En}. Here, $(P_n)_{n \in \N}$ denotes the family of phase estimation circuits (like described in the proof of the stable case of \cref{lem:helper-7p1}), where the eigenvalue register has $m(n)$ qubits for a sufficiently large polynomial $m$ to be specified later. Given an implicit parameter $n$, let $\delta = 2^{-3n}$ and $\eps = 2 \cdot 2^{-3n}$.

\begin{algorithm}
	\caption{The circuit $E_n$} \label{alg:En}
	\begin{algorithmic}[1]
		\Require $r \in \D{m(n)}$
		\State Initialize $n$-qubit registers $\reg A$ and $\reg B$ to the maximally entangled state $\ket{\Phi_n}_{\reg{AB}} = 2^{-n/2} \sum_{x \in \cube n} \ket{x}_{\reg A} \ket{x}_{\reg B}$.
		\State Initialize an $m(n)$-qubit register $\reg C$ to $\ket{0^{m(n)}}$, and apply the phase estimation circuit $P_n$ with eigenvector register $\reg A$ and eigenvalue register $\reg C$. \hlabel{line:Enl2}
		\State Measure $\reg C$ in the standard basis, and let $s \in \D{m(n)}$ denote the measurement outcome.
		\State If $\Delta(s,-r) > \eps$ then accept, otherwise reject.
	\end{algorithmic}
\end{algorithm}

Let $2^{-n/2} \sum_{j \in [2^n]} \ket{v_j} \ket{\overline v_j} \ket{\eta_j}$ denote the state of the circuit after \cref{line:Enl2} (where $\ket{v_j}, \ket{\overline v_j}, \ket{\eta_j}$ are defined as in the analysis of \cref{alg:Cn}). Let $s_j \in \D{m(n)}$ denote the outcome of a standard-basis measurement of $\ket{\eta_j}$; then
\begin{equation*}
	\pr{\text{$E_n(r)$ rejects}} = \E_{j \sim [2^n]} \pr{\Delta(s_j,-r) \le \eps}~.
\end{equation*}
Since $(P_n)_{n \in \N}$ is space-uniform, so is $(E_n)_{n \in \N}$, so by \cref{tomog-0} and the above equality there exists a $\PSPACE$-computable function $h$ such that for all $n \in \N, r \in \D{m(n)}$ it holds that $h(1^n, r) \in \D{\poly(n)}$ and
\begin{equation} \label{eq:h}
	\left| h(1^n, r) - \E_{j \sim [2^n]} \pr{\Delta(s_j, -r) \le \eps} \right| < 2^{-2n}~.
\end{equation}
Define
\begin{equation}
	\label{eq:phi-def}
	\phi(1^n) = \min \left \{r \in \D{m(n)}: h(1^n, r) < 2 \cdot 2^{-2n} \right \}~.
\end{equation}
(It is not immediately clear that $\phi(1^n)$ is well defined, i.e.\ that there exists $r \in \D{m(n)}$ such that $h(1^n, r) < 2 \cdot 2^{-2n}$, but we will see that this is the case.)

First we prove that $\phi(1^n)$ is well defined, at most $9 \cdot 2^{-2n}$, and $\PSPACE$-computable, and then we prove that $U_n' = e^{2 \pi i \phi(1^n)} U_n$ is stable, thus establishing \cref{clm:712}.
\begin{proof}[Proof that $\phi(1^n)$ is well-defined, at most $9 \cdot 2^{-2n}$, and $\PSPACE$-computable.]
	We first show by a counting argument that there exists an $r \in \D{m(n)}$ such that $r \le 9 \cdot 2^{-2n}$ and $\Delta(\theta_j, -r) > 2 \eps$ for all $j$. This holds because on one hand, we have
	\begin{equation*}
		\left| \left\{r \in \D{m(n)} \mid r \le 9 \cdot 2^{-2n} \right\} \right| \ge 9 \cdot 2^{m(n) - 2n}~.
	\end{equation*}
	On the other hand, we have
	\begin{align*}
		&\left| \left\{ r \in \D{m(n)} \mid \exists j: \Delta(\theta_j,-r) \le 2 \eps \right\} \right|
		\leq \sum_{j=1}^{2^n} \left| \left\{ r \in \D{m(n)} \mid \Delta(\theta_j, -r)
		\le 2 \eps \right\} \right| \\
		&\qquad \qquad \leq \sum_{j=1}^{2^n} \Paren{2 \cdot 2\eps \cdot 2^{m(n)} + 1}
		= 8 \cdot 2^{m(n) - 2n} + 2^n
		< 9 \cdot 2^{m(n) - 2n},
	\end{align*}
	where in the last inequality we take $m$ to be a sufficiently large polynomial. This implies the existence of such an $r$.
	
	We now prove that $h(1^n, r) < 2 \cdot 2^{-2n}$, which implies that $\phi(1^n)$ is well defined and at most $9 \cdot 2^{-2n}$ as required. By \eqref{eq:h} it holds that
	\begin{equation*}
		h(1^n, r) < \E_{j \sim [2^n]} \pr{\Delta(s_j, -r) \le \eps} + 2^{-2n},
	\end{equation*}
	so it suffices to prove that $\pr{\Delta(s_j,-r) \le \eps} < 2^{-2n}$ for all $j$. By the definition of $r$ and the triangle inequality for $\Delta$, the event $\Delta(s_j, -r) \le \eps$ implies the event
	\begin{equation*}
		2 \eps
		< \Delta(\theta_j, -r)
		\le \Delta(\theta_j, s_j) + \Delta(s_j, -r)
		\le \Delta(\theta_j, s_j) + \eps,
	\end{equation*}
	i.e.\ $\Delta(\theta_j, s_j) > \eps$. So by \eqref{eq:pe},
	\begin{equation*}
		\pr{\Delta(s_j,-r) \le \eps}
		\le \pr{\Delta(\theta_j, s_j) > \eps}
		\le O\Paren{2^{-m(n)} / \eps}
		< 2^{-2n},
	\end{equation*}
	where the last inequality follows by taking $m$ to be a sufficiently large polynomial. 
	
	Finally, since $h$ is $\PSPACE$-computable, $\phi$ is as well.
\end{proof}
\begin{proof}[Proof that $U_n'$ is stable.]
	Since $U_n'$ has eigenvalues $\exp \Paren{2\pi i (\theta_j + \phi(1^n))}$, the condition that $U_n'$ is stable is equivalent to the condition that $\Delta(\theta_j + \phi(1^n),0) \geq \delta$ for all $j \in [2^n]$, which in turn is equivalent to $\Delta(\theta_j,-\phi(1^n)) \geq \delta$ for all $j$. Suppose for contradiction that there exists a $j^*$ such that $\Delta(\theta_{j^*},-\phi(1^n)) < \delta$. By the definitions of $h$ and $\phi(1^n)$ (i.e.\ \eqref{eq:h} and \eqref{eq:phi-def}), we have
	\begin{align*}
		\pr{\Delta(s_{j^*}, -\phi(1^n)) \le \eps}
		&\le 2^n \, \E_{j \sim [2^n]} \pr{\Delta(s_j, -\phi(1^n)) \le \eps}
		\le 2^n \Paren{2^{-2n} + h(1^n, \phi(1^n))} \\
		&\le 2^n \Paren{2^{-2n} + 2 \cdot 2^{-2n}}
		= e^{-\Omega(n)}.
	\end{align*}
	On the other hand,
	\begin{align*}
		\pr{\Delta(s_{j^*},-\phi(1^n)) > \eps}
		&\le \pr{\Delta(s_{j^*}, \theta_{j^*}) + \Delta(\theta_{j^*},-\phi(1^n)) > \eps}
		\le \pr{\Delta(s_{j^*}, \theta_{j^*}) > \eps - \delta} \\
		&= \pr{\Delta(s_{j^*}, \theta_{j^*}) > 2^{-3n}}
		\le O\Paren{2^{-m(n) + 3n}}
		\le e^{-\Omega(n)},
	\end{align*}
	where the first inequality is by the triangle inequality for $\Delta$, the second is by the definition of $j_*$, the third is by the definitions of $\delta$ and $\eps$, the fourth is by \eqref{eq:pe}, and the last is by taking $m$ to be a sufficiently large polynomial. Therefore
	\begin{equation*}
		1 = \pr{\Delta(s_{j^*}, -\phi(1^n)) \le \eps} + \pr{\Delta(s_{j^*},-\phi(1^n)) > \eps} \le e^{-\Omega(n)},
	\end{equation*}
	which gives the desired contradiction.
\end{proof}
\subsubsection{Proof of \texorpdfstring{\cref{thm:main-7p1}}{Theorem 7.2}}

By \cref{lem:helper-7p1} and the fact that $U_n$ has canonical evolution time at most $\poly(n)$, there exists
\begin{itemize}[leftmargin=*]
	\renewcommand\labelitemi{--}
	\item a $\PSPACE$-computable function $f$ such that $f(1^n) \in \D{\poly(n)}[\poly(n)]$ for all $n$, and
	\item a sequence $(\ket{\psi_n})_{n \in \N} \in \statePSPACE$,
\end{itemize}
such that for all $n$ the reduced state $\rho_n$ on the first $n$ qubits of $\ket{\psi_n}$ is a program state for $U_n$ with evolution time $f(1^n)$ and error $e^{-\Omega(n)}$.

Assume without loss of generality that $q(n) \ge n$. Let $k$ be a sufficiently large polynomial to be chosen later, and let
\begin{equation*}
	\ket{\varphi_n} = \ket{\psi_n}^{\otimes k(n)} \otimes \ket{f(1^n)}~.
\end{equation*}
It is easy to see that $(\ket{\varphi_n})_{n \in \N}$ is in \statePSPACE, so by \cref{thm:prim-result} there exists a $\stateQIP{c',s'}$ verifier $V$ for $(\ket{\varphi_n})_{n \in \N}$ where
\begin{align*}
	&c'(n) = \exp(-q(n)),
	&s'(n,\delta) = \exp \Paren{e^{-q(n)} - q(n) \cdot \delta^4}~.
\end{align*}
Let $V_n$ be the $n$'th quantum verifier circuit in $V$ (i.e.\ the one for synthesizing $\ket{\varphi_n}$). Given $n$, let $\reg R$ be the register such that
\begin{equation*}
	\tr_{\reg R} (\varphi_n) = \rho_n^{\otimes k(n)} \otimes \kb{f\Paren{1^n}}~.
\end{equation*}
The following describes a $\cc{unitaryQIP}[c,s]$ verifier for $(U_n)_{n \in \N}$ as applied to an $n$-qubit input state $\ket\theta$: Simulate $V_n$, if $V_n$ rejects then reject, and if $V_n$ accepts and outputs a state $\sigma$ then accept and output $\alg{LMR}(\theta, \tr_{\reg R}(\sigma))$. (This verifier runs in polynomial time because $V$ and $\alg{LMR}$ do.)

Assume for simplicity that $n$ is sufficiently large. First we argue that
\begin{equation} \label{eq:asdf-sec7}
	\td\Paren{\alg{LMR}(\theta, \tr_{\reg R}(\sigma)), U_n \theta \adj U_n}
	\le \td(\sigma, \varphi_n) + \frac1{5 q(n)^2}
\end{equation}
for all $n$-qubit states $\ket\theta$ and all mixed states $\sigma$ on the same number of qubits as $\varphi_n$; we will use \eqref{eq:asdf-sec7} in the proofs of both completeness and soundness. Fix $n$, write $f = f(1^n), \rho = \rho_n$ etc.\ and let $W = \exp(2\pi i \cdot f \cdot \rho)$. By the triangle inequality,
\begin{align*}
	&\td\Paren{\alg{LMR}(\theta, \tr_{\reg R}(\sigma)), U \theta \adj U} \le \td \Paren{\alg{LMR}(\theta, \tr_{\reg R}(\sigma)), \alg{LMR}(\theta, \tr_{\reg R}(\varphi))} \\
	& \qquad \qquad \qquad + \td\Paren{\alg{LMR}(\theta, \tr_{\reg R}(\varphi)), W \theta \adj W}
	+ \td\Paren{W \theta \adj W, U \theta \adj U}~.
\end{align*}
By the definition of trace distance
\begin{equation*}
	\td\Paren{\alg{LMR}(\theta, \tr_{\reg R}(\sigma)), \alg{LMR}(\theta, \tr_{\reg R}(\varphi))} \le \td(\sigma, \varphi),
\end{equation*}
by the definition of $\reg R$ and \cref{thm:lmr}
\begin{equation*}
	\td\Paren{\alg{LMR}(\theta, \tr_{\reg R}(\varphi)), W \theta \adj W}
	= \td\Paren{\alg{LMR}(\theta, \rho^{\otimes k}, f), W \theta \adj W}
	\le O(f^2/k),
\end{equation*}
and since $\rho$ is a program state for $U$ with evolution time $f$ and error $e^{-\Omega(n)}$
\begin{equation*}
	\td\Paren{W \theta \adj W, U \theta \adj U} \le e^{-\Omega(n)},
\end{equation*}
so it follows that
\begin{equation*}
	\td\Paren{\alg{LMR}(\theta, \tr_{\reg R}(\sigma)), U \theta \adj U}
	\le \td(\sigma, \varphi) + O(f^2/k) + e^{-\Omega(n)}~.
\end{equation*}
Therefore, since $f \le \poly(n)$ there exists a polynomial $k$ such that \eqref{eq:asdf-sec7} holds.

Now we prove completeness. An example of an honest prover is one that simulates an honest prover in the $\stateQIP{c',s'}$ protocol. Then $V_n$ accepts with probability 1 and outputs a state $\sigma$ such that $\td(\sigma, \varphi_n) \le c'(n) \le e^{-\Omega(n)}$, so by \eqref{eq:asdf-sec7} for all $n$-qubit states $\ket \theta$ it holds that
\begin{equation*}
	\td\Paren{\alg{LMR}(\theta, \tr_{\reg R}(\sigma)), U_n \theta \adj U_n}
	\le e^{-\Omega(n)} + \frac1{5q(n)^2}
	< \frac1{q(n)}
	= c(n)~.
\end{equation*}

Finally we prove soundness. Consider an arbitrary prover (which may depend on the verifier's $n$-qubit input state $\ket\theta$), let $\sigma$ denote the output state of $V_n$ conditioned on accepting, and let $\delta = \td\Paren{\alg{LMR}(\theta, \tr_{\reg R}(\sigma)), U_n \theta \adj U_n}$. By \eqref{eq:asdf-sec7} and \cref{lem:4p},
\begin{equation*}
	\td(\sigma, \varphi_n)^4 \ge \delta^4 - \frac45 \cdot \frac1{q(n)^2},
\end{equation*}
so
\begin{align*}
	\pr{\text{verifier accepts}}
	&\le s'(n, \td(\sigma, \varphi_n))
	= \exp\Paren{e^{-q(n)} - q(n) \cdot \td(\sigma, \varphi_n)^4} \\
	&\le \exp\Paren{e^{-q(n)} - q(n) \cdot \delta^4 + \frac45 \cdot \frac1{q(n)}} \\
	&\le \exp\Paren{\frac1{q(n)} - q(n) \cdot \delta^4}
	= s(n, \delta)~. \qedhere
\end{align*}

\subsection{Protocol for general unitaries, but with restricted inputs}

Let $U = (U_n)_{n \in \N}$ denote a family of quantum circuits in \unitaryPSPACE, not necessarily with polynomial action (i.e.\ the unitaries can act nontrivially on the entire Hilbert space). We argue that there is a $\cc{unitaryQIP}$ protocol for $U$ provided that the verifier also receives as input a \emph{succinct description} of a polynomial-dimensional subspace $S$ that contains the input state (this was presented as \Cref{cor:poly-input-intro} in the introduction). By succinct description, we mean that there is a polynomial-space Turing machine $M$ that on input $1^n$, outputs the description of a polynomial-space quantum circuit $R$ that on input $\ket{0} \otimes \ket{\phi}$ outputs $\ket{0} \otimes (\id - \Pi_S) \ket{\phi} + \ket{1} \otimes \Pi_S \ket{\phi}$ (for all $n$-qubit states $\ket\phi$)
where $\Pi_S$ is the projection onto $S$. In other words, the Turing machine $M$ succinctly describes a circuit $R$ that ``recognizes'' states from the subspace $S$. 

The protocol for applying $U$ essentially reduces to using the protocol for polynomial-action unitary families from \cref{sec:poly-action}. Consider an input $(1^n, \ket{\phi}, M)$, where $\ket{\phi}$ is an $n$-qubit state and $M$ is a succinct description of a polynomial-space quantum circuit $R$ that recognizes a polynomial-dimensional subspace $S$ that contains $\ket{\phi}$. Let $S^\prime$ be the $(n+1)$-qubit subspace spanned by $\ket\varphi \ket0$ for $\ket\varphi$ in $S$, and by $\ket\varphi \ket1$ for $\ket\varphi$ in the subspace $U_n S$. Let $V$ be the $(n+1)$-qubit unitary defined by $V \ket\varphi \ket0 = U_n \ket\varphi \otimes \ket1$ and $V \ket\varphi \ket1 = \adj U_n \ket\varphi \otimes \ket0$ (for all $n$-qubit states $\ket\varphi$) and observe that $S^\prime$ is closed under action by $V$. Furthermore, since $R$ and $U_n$ are polynomial-space quantum circuits, there clearly exists a polynomial-space quantum circuit $R^\prime$ that recognizes $S^\prime$. The verifier can run the protocol to synthesize the unitary $V$ (on an input in $S^\prime$) which acts as $U_n \otimes I$ on the subspace $S \otimes \ket0$. The key to this reduction is to show that the program state corresponding to the unitary $V$ can be generated in quantum polynomial space. We sketch an argument below.

The same phase-estimation-based approach can be used, except before applying phase estimation to the $n$-qubit maximally entangled state $\ket{\Phi}$, an ancilla $\ket{0}$ qubit is adjoined and the polynomial-space circuit $R^\prime$ is run on the first $n+1$ qubits of $\ket{0} \otimes \ket{\Phi}$. Then, the first qubit is measured and the circuit post-selects on it being in the state $\ket{1}$. In other words, the maximally entangled state $\ket{\Phi}$ is projected to be supported only on the subspace $S^\prime$, and the resulting entangled state is
\[
\ket{\Phi_{S}^\prime} = \frac{1}{\sqrt{\dim S^\prime}} \sum_k \ket{w_k} \otimes \ket{w_k^*}
\]
where $\{ \ket{w_k} \}$ is a basis for the subspace $S^\prime$ in which $V$ is diagonal (i.e.\ they are eigenvectors of $V$). All other steps of \Cref{alg:Cn}, \Cref{alg:Dn}, and \Cref{alg:En} are the same. Thus, the resulting program states represent the unitary $V$ restricted to the subspace $S^\prime$, and thus are program states for $V$. These program states can be generated in polynomial space because the post-selection probability of projecting $\ket{\Phi}$ to $\ket{\Phi_S}$ is at least $2^{-n}$ (assuming that the subspace $S$ is at least one-dimensional). 

	\section{Multiple Entangled Provers}
\label{sec:multiple}

The class $\cc{QMIP}$ is defined analogously to $\cc{QIP}$, but with multiple provers who may share arbitrarily many entangled qubits. The following theorem characterizes this class:

\begin{thm}[\cite{ji2020mip}] \label{thm:qmip=re}
	$\cc{QMIP} = \cc{RE}$.
\end{thm}

Similarly, we define classes $\cc{stateQMIP}[c,s]$ and $\cc{unitaryQMIP}[c,s]$ analogously to $\stateQIP{c,s}$ and $\cc{unitaryQIP}[c,s]$ respectively, but with multiple provers who may share arbitrarily many entangled qubits. We also define the following analogues of the class $\cc{R}$ of computable languages, where by a ``computable sequence of quantum circuits" we mean a sequence in which the description of the $n$'th circuit can be computed as a function of $n$:

\begin{dfn}[$\cc{stateR}$]
	$\cc{stateR}$ is the class of all sequences $(\ket{\psi_n})_{n \in \N}$ such that each $\ket{\psi_n}$ is a state on $n$ qubits, and for every polynomial $q$ there exists a computable family of general quantum circuits $(C_n)_{n \in \N}$ such that for all $n \in \N$, the circuit $C_n$ takes no inputs and $C_n$ outputs a density matrix $\rho$ such that $\td(\rho, \psi_n) \leq \exp(-q(n))$.
\end{dfn}

\begin{dfn}[$\cc{unitaryR}$]
	$\cc{unitaryR}$ is the class of all computable sequences $(U_n)_{n \in \N}$ of unitary quantum circuits such that each $U_n$ acts on $n$ qubits.
\end{dfn}

We prove a statement identical to \cref{thm:prim-result} but with $\cc{stateR}$ and $\cc{stateQMIP}$ in place of $\cc{statePSPACE}$ and $\cc{stateQIP}$ respectively, and we also prove the converse statement:

\begin{thm} \label{thm:qmip-prim-result}
	$\cc{stateR} = \bigcap_{\textnormal{polynomials $q$}} \cc{stateQMIP}[c_q, s_q]$ for
	\begin{align*}
		&c_q(n) = \exp(-q(n)),
		&s_q(n,\delta) = \exp \Paren{e^{-q(n)} - q(n) \cdot \delta^4}~.
	\end{align*}
\end{thm}

\begin{proof}
	First we prove that $\cc{stateR} \subseteq \bigcap_{\textnormal{polynomials $q$}} \cc{stateQMIP}[c_q, s_q]$. That is, given a state family $(\ket{\psi_n})_{n \in \N}$ in $\cc{stateR}$ and a polynomial $q$, we prove that $(\ket{\psi_n})_{n \in \N}$ is in $\cc{stateQMIP}[c_q, s_q]$. Clearly there exist functions $\cp$ and $\ph$ like those defined in \cref{subsec:desc-prot}, except that these functions are merely computable rather than $\PSPACE$-computable. Then the associated languages $L_\cp$ and $L_\ph$ (again, defined in \cref{subsec:desc-prot}) are in $\cc{R}$, and hence admit $\cc{QMIP}$ verifiers $V_\cp$ and $V_\ph$ by \cref{thm:qmip=re} and the fact that $\cc{R} \subseteq \cc{RE}$.\footnote{We remark that with multiple \emph{un}entangled provers, obtaining the appropriate analogous statement appears to be nontrivial.} The rest of the proof is essentially the same as that of \cref{thm:prim-result}, observing that our $\cc{stateQIP}$ soundness amplification lemma (\cref{lem:stateQIP-amplification}) also holds for $\cc{stateQMIP}$ by essentially the same proof.\footnote{However, the analogous statement for multiple \emph{un}entangled provers does \emph{not} obviously hold, because during the first iteration of the soundness amplification procedure the provers could obtain entangled qubits for use in future iterations.}
	
	Now we prove that $\bigcap_{\textnormal{polynomials $q$}} \cc{stateQMIP}[c_q, s_q] \subseteq \cc{stateR}$. Let $(\ket{\psi_n})_{n \in \N}$ be a state family in $\bigcap_{\textnormal{polynomials $q$}} \cc{stateQMIP}[c_q, s_q]$, and let $p$ be an arbitrary polynomial. By the definition of $\cc{stateR}$, it suffices to prove that there exists a computable family of general quantum circuits $(C_n)_{n \in \N}$ such that for all $n \in \N$, the circuit $C_n$ takes no inputs and $C_n$ outputs a density matrix $\rho$ such that $\td(\rho, \psi_n) \leq \exp(-p(n))$. Let $q(n) = 4p(n) + 10$. The circuit $C_n$ does the following, where $V$ is a $\cc{stateQMIP}[c_q, s_q]$ verifier:
	\begin{enumerate}[leftmargin=*]
		\item Brute force over a discretization of the set of all provers for a $\cc{stateQMIP}[c_q, s_q]$ verifier, until finding a prover $P$ that $V_n$ accepts with probability at least $\exp\Paren{-e^{-q(n)}}$.
		\item Output the state $\rho$ produced by $V_n \interact P$ conditioned on accepting.
	\end{enumerate}
	Such a prover $P$ can be found because there exists an ``honest" prover that $V_n$ accepts with probability 1, and there exists an arbitrarily good approximation of the honest prover in the discretization of the set of provers.\footnote{For comparison, the fact that $\cc{QMIP} \subseteq \cc{RE}$ also follows by brute-forcing over a discretization of the set of all provers. But unlike the state synthesis algorithm described above, this $\cc{QMIP} \subseteq \cc{RE}$ algorithm is not guaranteed to terminate, which is why \cref{thm:qmip-prim-result} is not directly analogous to \cref{thm:qmip=re}.} By the soundness guarantee of $\cc{stateQMIP}[c_q, s_q]$ verifiers it holds that
	\begin{equation*}
		\exp\Paren{-e^{-q(n)}}
		\le \PR{\text{$V_n$ accepts $P$}}
		\le s_q\Paren{n, \td\Paren{\rho, \psi_n}}
		= \exp \Paren{e^{-q(n)} - q(n) \cdot \td\Paren{\rho, \psi_n}^4}~.
	\end{equation*}
	Rearranging yields
	\begin{equation*}
		\td\Paren{\rho, \psi_n}
		\le \Paren{\frac{2 e^{-q(n)}}{q(n)}}^{1/4}
		\le 2^{1/4} \exp(-q(n)/4)
		\le \exp(-p(n))
	\end{equation*}
	as required.
\end{proof}

We also prove a statement identical to \cref{thm:poly-action} but with $\cc{unitaryR}$ and $\cc{unitaryQMIP}$ in place of $\cc{unitaryPSPACE}$ and $\cc{unitaryQIP}$ respectively:

\begin{restatable}{thm}{qmip-polyaction}
	\label{thm:qmip-poly-action}
	Let $U = (U_n)_{n \in \N}$ be a family in $\cc{unitaryR}$ with polynomial action, and let $q$ be a polynomial. Then $U \in \cc{unitaryQMIP}[c,s]$ for
	\begin{align*}
		&c(n) = \frac{1}{q(n)},
		&s(n,\delta) = \exp \Paren{\frac{1}{q(n)} - q(n) \cdot \delta^4}~.
	\end{align*}
\end{restatable}

\begin{proof}
	The proof is essentially the same as that of \cref{thm:poly-action}, except that the natural analogue of \cref{lem:helper-7p1} (i.e.\ with ``$\cc{stateR}$" in place of ``$\statePSPACE$" and with ``computable" in place of ``$\PSPACE$-computable") holds trivially, and one should apply \cref{thm:qmip-prim-result} in place of \cref{thm:prim-result}.
\end{proof}

	\printbibliography[heading=bibintoc]
	
\end{document}